\documentclass[submission,copyright,creativecommons,conference]{eptcs}
\usepackage{hyperref}
\usepackage{enumerate}
\usepackage{amsmath, amsthm, amssymb}
\usepackage[latin1]{inputenc}
\usepackage[svgnames]{xcolor}
\usepackage{tikz}
\usepackage{stmaryrd}
\usepackage{braket}
\usepackage{xspace}
\usepackage{amsmath,amssymb,amsthm}
\usepackage{youngtab}

\usepackage{etoolbox}
\usepackage{tikz}
\usetikzlibrary{positioning}
\usetikzlibrary{calc}
\usetikzlibrary{arrows}
\usetikzlibrary{decorations.markings}
\usetikzlibrary{decorations.pathreplacing}
\usetikzlibrary{backgrounds}
\usetikzlibrary{shapes.geometric}

\pgfdeclarelayer{background}
\pgfdeclarelayer{foreground}
\pgfsetlayers{background,main,foreground}

\tikzset{comod/.style={rectangle, minimum width=25pt, minimum height=10pt, draw, inner sep=1pt}}

\def\fillgray{black!20}
\def\fillA{black!20}
\def\fillB{red!50}
\def\fillC{blue!50}
\def\fillClight{blue!40}
\def\fillD{green!50}
\def\fillcomp{\fillB}
\def\C{\mathbb{C}}
\newcommand\Ps{\ensuremath{\textcolor{blue}{P_{\mathrm{s}}}}}
\newcommand\Pd{\ensuremath{\textcolor{blue}{P_{\mathrm{d}}}}}
\def\id{\ensuremath{\mathrm{id}}}

\renewcommand{\-}[0]{\nobreakdash-\hspace{0pt}}

\tikzset{arrow/.style={decoration={
    markings,
    mark=at position #1 with \arrow{angle 60}},
    postaction=decorate}
}
\tikzset{reversed arrow/.style={decoration={
    markings,
    mark=at position #1 with \arrow{angle 60 reversed}},
    postaction=decorate}
}

\tikzset{keyvertexcolour/.initial=black}
\tikzset{vertex colour/.style={keyvertexcolour={#1}}}
\tikzset {triangle/.style={
        draw,
        shape border rotate=-90,
        isosceles triangle,
        isosceles triangle apex angle=80,        
        minimum height=2em}}
\tikzset {triangleup/.style={
        draw,
        shape border rotate=90,
        isosceles triangle,
        isosceles triangle apex angle=80,
         minimum height=2em}}
\newlength\vertexradius
\newlength\innerradius
\def\halfanglesep{24}
\def\tripleanglesep{24}
\def\sideangle{30}
\def\tripleanglesep{40}
\makeatletter
\pgfdeclareshape{Vertex}
{
    \savedanchor\centerpoint
    {
        \pgf@x=0pt
        \pgf@y=0pt
    }
    \anchor{n1}
    {
        \pgfextractx{\pgf@x}{\pgfpointpolar{90+\halfanglesep}{\innerradius}}
        \pgfextracty{\pgf@y}{\pgfpointpolar{90+\halfanglesep}{\innerradius}}
    }
    \anchor{n2}
    {
        \pgfextractx{\pgf@x}{\pgfpointpolar{90-\halfanglesep}{\innerradius}}
        \pgfextracty{\pgf@y}{\pgfpointpolar{90-\halfanglesep}{\innerradius}}
    }
    \anchor{n}
    {
        \pgfextractx{\pgf@x}{\pgfpointpolar{90}{\innerradius}}
        \pgfextracty{\pgf@y}{\pgfpointpolar{90}{\innerradius}}
    }
    \anchor{ne1}
    {
        \pgfextractx{\pgf@x}{\pgfpointpolar{\sideangle+\halfanglesep}{\innerradius}}
        \pgfextracty{\pgf@y}{\pgfpointpolar{\sideangle+\halfanglesep}{\innerradius}}
    }
    \anchor{ne2}
    {
        \pgfextractx{\pgf@x}{\pgfpointpolar{\sideangle-\halfanglesep}{\innerradius}}
        \pgfextracty{\pgf@y}{\pgfpointpolar{\sideangle-\halfanglesep}{\innerradius}}
    }
    \anchor{ne}
    {
        \pgfextractx{\pgf@x}{\pgfpointpolar{\sideangle}{\innerradius}}
        \pgfextracty{\pgf@y}{\pgfpointpolar{\sideangle}{\innerradius}}
    }
    \anchor{se1}
    {
        \pgfextractx{\pgf@x}{\pgfpointpolar{-\sideangle+\halfanglesep}{\innerradius}}
        \pgfextracty{\pgf@y}{\pgfpointpolar{-\sideangle+\halfanglesep}{\innerradius}}
    }
    \anchor{se2}
    {
        \pgfextractx{\pgf@x}{\pgfpointpolar{-\sideangle-\halfanglesep}{\innerradius}}
        \pgfextracty{\pgf@y}{\pgfpointpolar{-\sideangle-\halfanglesep}{\innerradius}}
    }
    \anchor{s1}
    {
        \pgfextractx{\pgf@x}{\pgfpointpolar{-90-\halfanglesep}{\innerradius}}
        \pgfextracty{\pgf@y}{\pgfpointpolar{-90-\halfanglesep}{\innerradius}}
    }
    \anchor{s2}
    {
        \pgfextractx{\pgf@x}{\pgfpointpolar{-90+\halfanglesep}{\innerradius}}
        \pgfextracty{\pgf@y}{\pgfpointpolar{-90+\halfanglesep}{\innerradius}}
    }
    \anchor{sA}
    {
        \pgfextractx{\pgf@x}{\pgfpointpolar{-90-\tripleanglesep}{\innerradius}}
        \pgfextracty{\pgf@y}{\pgfpointpolar{-90-\tripleanglesep}{\innerradius}}
    }
    \anchor{sB}
    {
        \pgfextractx{\pgf@x}{\pgfpointpolar{-90}{\innerradius}}
        \pgfextracty{\pgf@y}{\pgfpointpolar{-90}{\innerradius}}
    }
    \anchor{sC}
    {
        \pgfextractx{\pgf@x}{\pgfpointpolar{-90+\tripleanglesep}{\innerradius}}
        \pgfextracty{\pgf@y}{\pgfpointpolar{-90+\tripleanglesep}{\innerradius}}
    }
    \anchor{s}
    {
        \pgfextractx{\pgf@x}{\pgfpointpolar{-90}{\innerradius}}
        \pgfextracty{\pgf@y}{\pgfpointpolar{-90}{\innerradius}}
    }
    \anchor{sw1}
    {
        \pgfextractx{\pgf@x}{\pgfpointpolar{180+\sideangle+\halfanglesep}{\innerradius}}
        \pgfextracty{\pgf@y}{\pgfpointpolar{180+\sideangle+\halfanglesep}{\innerradius}}
    }
    \anchor{sw2}
    {
        \pgfextractx{\pgf@x}{\pgfpointpolar{180+\sideangle-\halfanglesep}{\innerradius}}
        \pgfextracty{\pgf@y}{\pgfpointpolar{180+\sideangle-\halfanglesep}{\innerradius}}
    }
    \anchor{sw}
    {
        \pgfextractx{\pgf@x}{\pgfpointpolar{180+\sideangle}{\innerradius}}
        \pgfextracty{\pgf@y}{\pgfpointpolar{180+\sideangle}{\innerradius}}
    }
    \anchor{nw1}
    {
        \pgfextractx{\pgf@x}{\pgfpointpolar{180-\sideangle+\halfanglesep}{\innerradius}}
        \pgfextracty{\pgf@y}{\pgfpointpolar{180-\sideangle+\halfanglesep}{\innerradius}}
    }
    \anchor{nw2}
    {
        \pgfextractx{\pgf@x}{\pgfpointpolar{180-\sideangle-\halfanglesep}{\innerradius}}
        \pgfextracty{\pgf@y}{\pgfpointpolar{180-\sideangle-\halfanglesep}{\innerradius}}
    }
    \anchor{nw}
    {
        \pgfextractx{\pgf@x}{\pgfpointpolar{180-\sideangle}{\innerradius}}
        \pgfextracty{\pgf@y}{\pgfpointpolar{180-\sideangle}{\innerradius}}
    }
    \anchor{center}{\centerpoint}
    \anchorborder{\centerpoint}
    \backgroundpath
    {
        \pgfkeysgetvalue{/tikz/keyvertexcolour}{\vcol}
        \begin{pgfonlayer}{foreground}
        \pgfsetfillcolor{\vcol}
        \pgfsetfillopacity{1}
        \pgfpathcircle{\pgfpoint{0cm}{0cm}}{\vertexradius}
        \pgfusepath{fill}
        \pgfsetstrokecolor{black}
        \pgfsetlinewidth{0.8pt}
        \pgfpathcircle{\pgfpoint{0cm}{0cm}}{\vertexradius}
        \pgfusepath{stroke}
        \end{pgfonlayer}
    }
}
\makeatother

\def\labelsize{0.7}

\def\nwangle{180-\sideangle}
\def\neangle{\sideangle}
\def\swangle{180+\sideangle}
\def\seangle{-\sideangle}

\usepgflibrary{shapes.geometric}

\theoremstyle{plain}
\newtheorem{theorem}{Theorem}
\newtheorem{lemma}[theorem]{Lemma}

\theoremstyle{definition}
\newtheorem{definition}[theorem]{Definition}

\renewcommand{\-}[0]{\nobreakdash-\hspace{0pt}}
\newcommand\vc[1]{\begin{tabular}{@{}l}#1\end{tabular}}
\makeatletter
\def\calign@preamble{%
   &\hfil\strut@
    \setboxz@h{\@lign$\m@th\displaystyle{##}$}%
    \ifmeasuring@\savefieldlength@\fi
    \set@field
    \hfil
    \tabskip\alignsep@
}
\let\cmeasure@\measure@
\patchcmd\cmeasure@{\divide\@tempcntb\tw@}{}{}{}
\patchcmd\cmeasure@{\divide\@tempcntb\tw@}{}{}{}
\patchcmd\cmeasure@{\ifodd\maxfields@
  \global\advance\maxfields@\@ne
  \fi}{}{}{}    
\newenvironment{calign}
{%
  \let\align@preamble\calign@preamble
  \let\measure@\cmeasure@
  \align
}
{%
  \endalign
}  
\makeatother
\newcommand\ignore[1]{}

\begin{document}

\title{A 2-Categorical Analysis of Complementary Families, Quantum Key Distribution and the Mean King Problem}

\author{Krzysztof Bar
\institute{Department of Computer Science,
University of Oxford}
\email{krzysztof.bar@cs.ox.ac.uk}
\and
Jamie Vicary
\institute{Centre for Quantum Technologies, University of Singapore\\
and Department of Computer Science, University of Oxford}
\email{jamie.vicary@cs.ox.ac.uk}
}
\def\titlerunning{2-Categories, Complementary Families, Quantum Key Distribution and the Mean King Problem}
\def\authorrunning{Krzysztof Bar \& Jamie Vicary}

\def\vertexsize{1.0}

\newcommand\cat[1]{\ensuremath{\text{\textbf{#1}}}}

\maketitle

\begin{abstract}
This paper explores the use of 2\-categorical technology for describing and reasoning about complex quantum procedures. We give syntactic definitions of a family of complementary measurements, and of quantum key distribution, and show that they are equivalent. We then show abstractly that either structure gives a solution to the Mean King problem, which we also formulate 2\-categorically.
\end{abstract}

\section{Introduction}

The 2\-categorical approach to quantum information is now well-developed in its basic aspects~\cite{BarVicary,StayVicary,Vicary:2012hqt}. The central aim is to use the structure of a symmetric monoidal 2\-category to describe quantum procedures in an abstract way, such that the ordinary versions of these procedures are recovered when we apply the formalism in \cat{2Hilb}, the symmetric monoidal 2\-category of 2\-Hilbert spaces~\cite{b97-hda2}. This builds on the highly successful categorical quantum mechanics research programme of Abramsky, Coecke and collaborators~\cite{Abramsky:2004ac,ac08-cqm}, in which quantum procedures are axiomatized in terms of monoidal 1\-categories.

The key advantage of the 2\-categorical setup is that many important structures, such as teleportation, dense coding and complementary observables, can be defined by single 2\-categorical equations. These typically have direct physical interpretations, with the defining equation for a structure following immediately from a careful physical description of its required properties. A computer algebra system \emph{TwoVect}~\cite{r11-2vect, 2vect} allows these equations to be directly evaluated computationally.

In this paper, we show that the formalism can be applied successfully to more sophisticated quantum procedures: measurements in a complementary family of bases, quantum key distribution (QKD), and the Mean King problem. For each scenario, we write down a 2\-categorical equation that defines the entire procedure in a precise way. For example, here is the defining diagram for BB84 QKD:
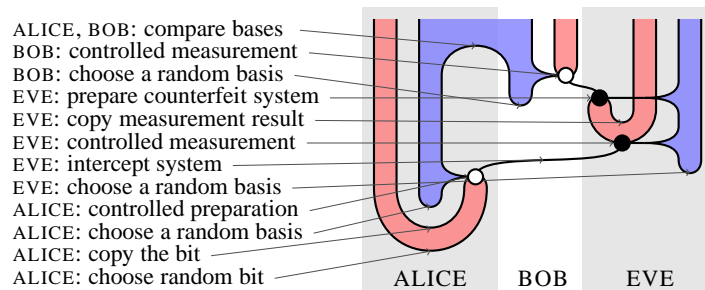
\begin{align*}
\begin{aligned}
\begin{tikzpicture} [scale=0.3,thick]
\draw [fill=black!10, draw=none] (-9,-7.25) rectangle +(6,12.75);
\draw [fill=black!10, draw=none] (0.75,-7.25) rectangle +(5.65,12.75);
\node (A) [Vertex, scale=\vertexsize, vertex colour=white] at (-4,-2) {};
\node (B) [Vertex, scale=\vertexsize] at (2.5,-0.5) {};
\node (C) [Vertex, scale=\vertexsize, vertex colour=white] at (0, 2.5) {};
\node (D) [Vertex, scale=\vertexsize] at (1.5, 1.5) {};
\draw [out=up, in=down, looseness=0.5] (-4,-2) to (2.5,-0.5);
\draw [out=up, in=down] (1.5, 1.5) to (0, 2.5);
\fill [fill=\fillcomp, draw, fill opacity=0.8]  (0.5,5)
to (0.5, 3)
to [out=down, in=down, looseness=2] (-0.5, 3)
to (-0.5,5);
\fill [fill=\fillcomp, draw, fill opacity=0.8]  (3,5) 
to (3,1)
to [out=down, in=down, looseness=2](2,1)
to (2, 1)
to [out=up, in=up, looseness=2] (1,1)
to [out=down, in=left](2.5,-0.5)
to [out=right, in=down](4,1)
to (4,5);
\fill [fill=\fillcomp, fill opacity=0.8, draw] (-7.5, 5)
to (-7.5, -2.5)
to [out=down, in=down, looseness=2](-4.5, -2.5)
to [out=up, in=up, looseness=2](-3.5, -2.5)
to [out=down, in=down, looseness=1.9](-8.5, -2.5)
to (-8.5, 0)
to (-8.5, 5);
\fill [fill=\fillC, fill opacity=0.8, draw] (-6.5, 5)
to (-6.5, -2.75)
to [out=down, in=down, looseness=2](-5.5, -2.75)
to [out=up, in=left] (-4,-2)
to [out=left, in=down] (-5.5, -1.25)
to (-5.5, 2.5)
to [out=up, in=up, looseness=1.5] (-2.5, 2.5)
to (-2.5, 1.75)
to [out=down, in=down, looseness=2] (-1.5,1.75)
to (-1.5,1.75)
to [out=up, in=left] (0, 2.5)
to [out=left, in=down] (-1.5, 3.25)
to [out=up, in=down](-1.5, 5);
\fill [fill=\fillC, fill opacity=0.8, draw] (5,5)
to (5,2.25)
to [out=down, in=right](3.5, 1.5)
to [out=right, in=up](5, 0.75)
to (5, 0.25)
to [out=down, in=right](3.5,-0.5)
to [out=right, in=up](5,-1.25)
to [out=down, in=down, looseness=2](6,-1.25)
to (6,5);
\draw (3.5, -0.5) to (2.5, -0.5);
\draw (3.5, 1.5) to (1.5, 1.5);
\node [] at (-6,-6.5){\sc alice};
\node [] at (-1,-6.5){\sc bob};
\node [] at (3.75,-6.5){\sc eve};
\foreach \x/\ya/\yb/\text in
  {-6/-6.5/-5.3/{\sc alice}: choose random bit,
   -6/-5.5/-4.3/{\sc alice}: copy the bit,
   -6/-4.5/-3.3/{\sc alice}: choose a random basis,
   -4.3/-3.5/-2.0/{\sc alice}: controlled preparation,
   5.5/-2.5/-1.8/{\sc eve}: choose a random basis,
   -1/-1.5/-1.25/{\sc eve}: intercept system,
   2.1/-0.5/-0.5/{\sc eve}: controlled measurement,
   2.5/0.5/0.4/{\sc eve}: copy measurement result,
   1.1/1.5/1.5/{\sc eve}: prepare counterfeit system,
   -2/2.5/1.15/{\sc bob}: choose a random basis,
   -0.3/3.5/2.5/{\sc bob}: controlled measurement,
   -4/4.5/3.8/{\sc alice, bob}: compare bases   }
{
  \node (z) at (-25,\ya) [anchor=west, font=\footnotesize] {\text\vphantom{p|}};
  \draw [black!70, ->, ultra thin] (z.east) to (\x,\yb);
}
\end{tikzpicture}
\end{aligned}
\end{align*}
We have annotated this diagram with text to show how the different parts are to be interpreted, but we emphasize that this annotation is superfluous: everything about the flow of quantum and classical information is captured by the diagram itself. To complete the abstract definition of BB84 quantum key distribution, we require that this diagram is equal to a second diagram which encodes the intended result of the procedure.

The main contributions of this paper are as follows:
\begin{itemize}
\item Definitions~\ref{def:controlledcomplementarity},~\ref{BB84QKD},~\ref{E91QKD} and \ref{Def:Mean king problem scheme} give 2\-categorical equations whose solutions in \cat{2Hilb} correspond exactly to implementations of a family of complementary observables, BB84 QKD, E91 QKD, and solutions of the Mean King problem respectively.
\item In Theorem~\ref{thm:equivalent} we show that the 2\-categorical definition for a family of complementary measurements is equivalent to that for  QKD. While an equivalence between these notions seems generally expected in the community, we are not able to find an existing crisp proof in the literature.
\item In Theorem~\ref{lemma: MKP correctness} we give a graphical proof of correctness of Klappenecker and Rottleer's solution~\cite{Gothic} to the Mean King problem. This is roughly the same complexity as the original proof, but quite different in nature. The graphical proof makes clear the role played by complementarity.
\end{itemize}
A significant result on the categorical basis of quantum key distribution was given by Coecke and Perdrix in~\cite[Proposition~7.4]{coeckeperdrix}, which demonstrates the correctness of QKD based on a pair of complementary observables. Our work goes beyond this result, as we work with arbitrary families of complementary observables rather than a single pair, and we further show that every implementation of QKD gives rise to a family of complementary observables.

A primary avenue of future work arising from our results will be investigating the existence of nonstandard models. It has been shown that a category of groupoids, profunctors and spans admits combinatorial `toy models' of teleportation, as solutions to a 2\-categorical equation, from which ordinary quantum teleportation can be recovered by applying a 2\-functor into \cat{2Hilb}~\cite{BarVicary}. It will be interesting to use the results of this paper to investigate whether combinatorial toy models of quantum key distribution can also be built in that setting.

\paragraph{Remark on colour and transparency.} The diagrams in this paper make essential use of colour and transparency. We therefore recommend reading this paper on a screen, or as a colour printout. For printing we recommend Adobe Reader, as some other PDF viewers do not correctly handle transparency.

\subsection{The 2\-category \cat{2Hilb} of 2\-Hilbert spaces}

It has been argued in~\cite{Vicary:2012hqt} that the 2\-category of 2\-Hilbert spaces~\cite{b97-hda2} is the correct 2\-categorical setting in which to analyze quantum informatic procedures. We recall the following construction of \cat{2Hilb}, which is most useful for calculation. For more details, see the papers cited above.
\begin{definition}
The symmetric monoidal 2\-category \cat{2Hilb} has \textit{objects} given by natural numbers, \textit{1\-morphisms} given by matrices of finite-dimensional Hilbert spaces, and \textit{2\-morphisms} given by matrices of linear maps. Details of the compositional structure of \cat{2Hilb} are available in the references given.
\end{definition}

\noindent
This gives the formal categorical semantics that forms the primary model of our abstract syntax, introduced in Section~\ref{sec: The topological formalism}.

\subsection{The topological formalism}
\label{sec: The topological formalism}

The basic 2\-categorical structures on which the theory is built have simple graphical representations~\cite{Vicary:2012hqt}, thanks to the graphical notation for monoidal 2\-categories. This graphical formalism involves surfaces, lines and vertices. Their basic interpretation is as follows:
\begin{center}
\begin{tabular}{lll}
\bf Category theory & \bf Geometry & \bf Interpretation
\\
Objects & Surfaces & Classical information
\\
1\-Morphisms & Lines & Quantum systems
\\
2-Morphisms & Vertices & Physical operations
\end{tabular}
\end{center}
Composite diagrams involving many vertices are interpreted as a series of actions that place over time, with time flowing from bottom to top. In the graphical calculus, composition of 1\-morphisms is given by horizontal juxtaposition, and composition of 2\-morphisms by vertical juxtaposition. The tensor product is given by `overlaying' regions one above the other, perpendicular to the plane of the page and the tensor unit is expressed by an unlabelled, empty region.

We desire the ability to take the formal adjoint of 2\-cells, represented graphically by flipping a diagram about a horizontal axis.
\begin{definition}
A \textit{dagger 2\-category} is a 2\-category equipped with an involutive operation $\dag$ on 2\-cells, such that for all $\mu: F \Rightarrow G$ we have $\mu ^\dag: G \Rightarrow F$, which is functorial and compatible with the rest of the monoidal 2\-category structure.
\end{definition}
\begin{definition}
A 2\-cell $\mu$ is \emph{unitary} when $\mu \circ \mu ^\dag = \id$ and $\mu ^\dag \circ \mu = \id$.
\end{definition}

The core graphical theory makes use of only a small number of graphical components. They give the formal syntax for our theory. We summarize them here, along with their interpretations, which are motivated in detail in~\cite{Vicary:2012hqt}.
\allowdisplaybreaks
\def\tys{0.8}
\def\aascale{1.4}
\def\aaspace{\hspace{30pt}}
\setlength\fboxsep{0pt}
\def\sep{5pt}
\def\innersep{3pt}
\def\littlegap{12pt}
\def\boxmargin{0.15cm}
\newcommand{\centerdia}[1]{#1}
\newcommand\newtwocell[2]
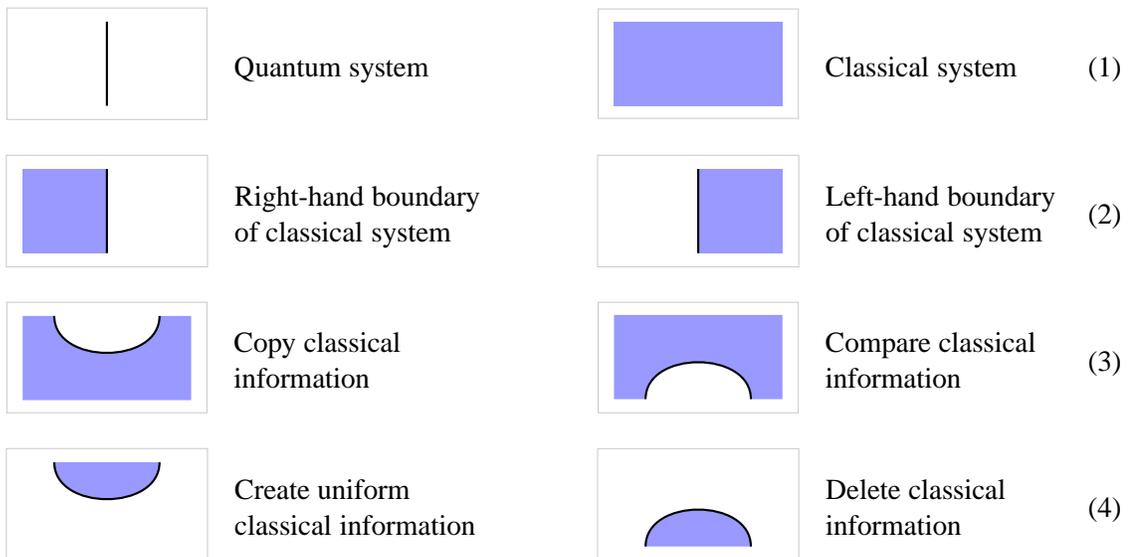
{\begin{aligned}
\begin{tikzpicture}[scale=\aascale, yscale=\tys]
    #1
    \draw [black!20]
        ([xshift=-\boxmargin, yshift=-\boxmargin] current bounding box.south west)
        rectangle
        ([xshift=\boxmargin, yshift=\boxmargin] current bounding box.north east);
\end{tikzpicture}
\end{aligned}
\hspace{10pt} \makebox[80pt][l]{\vc{#2}}}
\newcommand\separatetwocells{\\[\sep]}
\begin{calign}
\newtwocell{
    \draw [fill=white, draw=none] (0.2,-0.5)
        to (0.5,-0.5)
        to (1.8,-0.5)
        to (1.8,-1.5)
        to (0.2,-1.5)
        to (0.2,-0.5);
    \draw [thick] (1,-0.5)
        to (1,-1.5);
}
{Quantum system}
\hspace{\littlegap}&\hspace{\littlegap}
\newtwocell{
    \draw [fill=\fillClight, draw=none] (0.2,0.5)
        to (0.5,0.5)
        to (1.8,0.5)
        to (1.8,1.5)
        to (0.2,1.5)
        to (0.2,0.5);
}
{Classical system}
\separatetwocells
\label{eq:top1}
\newtwocell{
    \draw [white] (1.8,-0.5) rectangle (0.2,-1.5);
    \draw [fill=\fillClight, draw=none] (0.2,-0.5)
        to (0.5,-0.5)
        to (1,-0.5)
        to (1,-1.5)
        to (0.2,-1.5)
        to (0.2,-0.5);
    \draw [thick] (1,-0.5)
        to (1,-1.5);
}
{Right-hand boundary\\of classical system}
\hspace{\littlegap}&\hspace{\littlegap}
\newtwocell{
    \draw [white] (1.8,-0.5) rectangle (0.2,-1.5);
    \draw [fill=\fillClight, draw=none] (1.8,-0.5)
        to (1,-0.5)
        to (1,-1.5)
        to (1.8,-1.5)
        to (1.8,-0.5);
    \draw [thick] (1,-0.5)
        to (1,-1.5);
}
{Left-hand boundary\\of classical system}
\separatetwocells
\newtwocell{
    \draw [fill=\fillClight, draw=none] (0.2,-0.5)
        to (0.5,-0.5)
        to [out=down, in=down, looseness=1.5] (1.5, -0.5)
        to (1.8,-0.5)
        to (1.8,-1.5)
        to (0.2,-1.5)
        to (0.2,-0.5);
    \draw [thick] (0.5,-0.5)
        to [out=down, in=down, looseness=1.5] (1.5,-0.5);
}
{Copy classical\\information}
\hspace{\littlegap}&\hspace{\littlegap}
\newtwocell{
    \draw [fill=\fillClight, draw=none] (0.2,0.5)
        to (0.5,0.5)
        to [out=up, in=up, looseness=1.5] (1.5,0.5)
        to (1.8,0.5)
        to (1.8,1.5)
        to (0.2,1.5)
        to (0.2,0.5);
    \draw [thick] (0.5,0.5)
        to [out=up, in=up, looseness=1.5] (1.5,0.5);
}
{Compare classical\\information}
\separatetwocells
\label{eq:top3}
\newtwocell{
    \draw [white] (0.2,1) to (1.8,1);
    \draw [fill=\fillClight, thick] (0.5,1.5)
        to [out=down, in=down, looseness=1.5] (1.5,1.5);
    \draw [thick, white] (1,0.5) to (1,1);
}
{Create uniform\\classical information}
\hspace{\littlegap}&\hspace{\littlegap}
\newtwocell{
    \draw [white] (0.2,-1) to (1.8,-1);
    \draw [fill=\fillClight, thick] (0.5,-1.5)
        to [out=up, in=up, looseness=1.5] (1.5,-1.5);
    \draw [thick, white] (1,-0.5) to (1,-1);
}
{Delete classical\\information}
\end{calign}
These components are required to satisfy a set of axioms, which amount to saying that the boundary of a region is topological, and that holes can be eliminated:
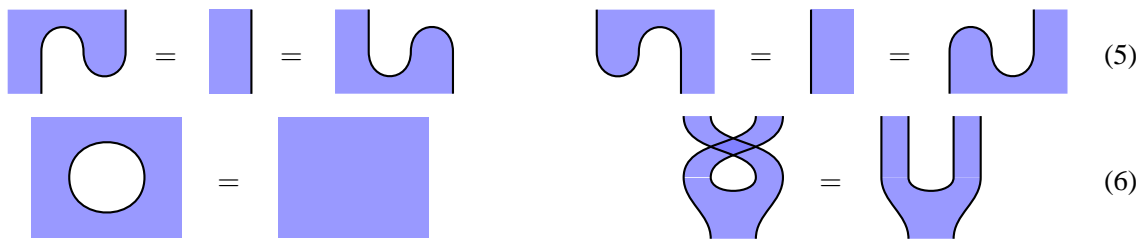
\def\aascale{0.7}
\begin{calign}
\label{eq:topeq1}
\def\quad{\hspace{0.3cm}}
\begin{aligned}
\begin{tikzpicture}[scale=\aascale,xscale=0.8, yscale=\tys]
\draw [use as bounding box, draw=none] (-0.5,0) rectangle (2.3,2);
\draw [white] (-0.5,0) to (3.1,2);
\draw [fill=\fillClight, draw=none] (-0.5,0) to (0.3,0) to (0.3,1)
    to [out=up, in=up, looseness=2] (1.3,1)
    to [out=down, in=down, looseness=2] (2.3,1)
    to (2.3,2) to (-0.5,2);
\draw [thick] (0.3,0) to (0.3,1)
    to [out=up, in=up, looseness=2] (1.3,1)
    to [out=down, in=down, looseness=2] (2.3,1)
    to (2.3,2);
\end{tikzpicture}
\end{aligned}
\quad=\quad
\begin{aligned}
\begin{tikzpicture}[scale=\aascale,xscale=0.8, yscale=\tys]
\draw [use as bounding box, draw=none] (1,0) rectangle (0,2);
\draw [white] (0,0) to (2,2);
\draw [fill=\fillClight, draw=none] (0,0)
    to (1,0)
    to (1,2)
    to (0,2);
\draw [thick] (1,0) to (1,2);
\end{tikzpicture}
\end{aligned}
\quad=\quad
\begin{aligned}
\begin{tikzpicture}[scale=\aascale,xscale=0.8, yscale=\tys]
\draw [use as bounding box, draw=none] (-0.5,0) rectangle (2.3,-2);
\draw [white] (-0.5,0) to (3.1,-2);
\draw [fill=\fillClight, draw=none] (-0.5,0) to (0.3,0) to (0.3,-1)
    to [out=down, in=down, looseness=2] (1.3,-1)
    to [out=up, in=up, looseness=2] (2.3,-1)
    to (2.3,-2) to (-0.5,-2);
\draw [thick] (0.3,0) to (0.3,-1)
    to [out=down, in=down, looseness=2] (1.3,-1)
    to [out=up, in=up, looseness=2] (2.3,-1)
    to (2.3,-2);
\end{tikzpicture}
\end{aligned}
&&
\begin{aligned}
\begin{tikzpicture}[scale=\aascale,xscale=0.8, yscale=\tys]
\draw [use as bounding box, draw=none] (-2.3,0) rectangle (0.5,2);
\draw [white] (0.5,0) to (-3.1,2);
\draw [fill=\fillClight, draw=none] (0.5,0) to (-0.3,0) to (-0.3,1)
    to [out=up, in=up, looseness=2] (-1.3,1)
    to [out=down, in=down, looseness=2] (-2.3,1)
    to (-2.3,2) to (0.5,2);
\draw [thick] (-0.3,0) to (-0.3,1)
    to [out=up, in=up, looseness=2] (-1.3,1)
    to [out=down, in=down, looseness=2] (-2.3,1)
    to (-2.3,2);
\end{tikzpicture}
\end{aligned}
\quad=\quad
\begin{aligned}
\begin{tikzpicture}[scale=\aascale,xscale=0.8, yscale=\tys]
\draw [use as bounding box, draw=none] (-1,0) rectangle (0,2);
\draw [white] (0,0) to (-2,2);
\draw [fill=\fillClight, draw=none] (0,0)
    to (-1,0)
    to (-1,2)
    to (-0,2);
\draw [thick] (-1,0) to (-1,2);
\end{tikzpicture}
\end{aligned}
\quad=\quad
\begin{aligned}
\begin{tikzpicture}[scale=\aascale,xscale=0.8, yscale=\tys]
\draw [use as bounding box, draw=none] (-2.3,0) rectangle (0.5,-2);
\draw [white] (0.5,0) to (-3.1,-2);
\draw [fill=\fillClight, draw=none] (0.5,0) to (-0.3,0) to (-0.3,-1)
    to [out=down, in=down, looseness=2] (-1.3,-1)
    to [out=up, in=up, looseness=2] (-2.3,-1)
    to (-2.3,-2) to (0.5,-2);
\draw [thick] (-0.3,0) to (-0.3,-1)
    to [out=down, in=down, looseness=2] (-1.3,-1)
    to [out=up, in=up, looseness=2] (-2.3,-1)
    to (-2.3,-2);
\end{tikzpicture}
\end{aligned}
\\
\label{eq:topeq2}
\begin{aligned}
\begin{tikzpicture}[yscale=\tys]
\draw [fill=\fillClight, draw=none] (0.5,0.5) rectangle (2.5,2.5);
\draw [fill=white, thick] (1,1.5)
    to [out=up, in=up, looseness=2] (2,1.5)
    to [out=down, in=down, looseness=2] (1,1.5);
\end{tikzpicture}
\end{aligned}
\quad=\quad
\begin{aligned}
\begin{tikzpicture}[yscale=\tys]
\draw [fill=\fillClight, draw=none] (0.5,0.5) rectangle (2.5,2.5);
\end{tikzpicture}
\end{aligned}
&&
\begin{aligned}
\begin{tikzpicture}[scale=0.6, yscale=-0.85, yscale=\tys]
\draw [thin, red, fill=\fillClight, opacity=0.8] (1.5,1.27) to (1,1) to (1.5,0.73) to (2,1);
\draw [thin, red, fill=\fillC, opacity=1] (1.5,1.27) to (1,1) to (1.5,0.73) to (2,1);
\draw [fill=\fillClight, draw=none] (1,2) to [out=down, in=up, in looseness=1.3] (2.6,0) to (2,0) to [out=up, in=down, in looseness=1.3] (0.4,2)
    to (1,2)
    to [out=up, in=up, looseness=1.5] (2,2) to [out=down, in=up, in looseness=1.3] (0.4,0) to (1,0)
    to [out=up, in=down, in looseness=1.3] (2.6,2)
    to [out=up, in=down] (2,4)
    to (1,4)
    to [out=down, in=up] (0.4,2);
\draw [thick, opacity=1] (0.4,0)
    to [out=up, in=down, out looseness=1.3] (2,2)
    to [out=up, in=up, looseness=1.5] (1,2)
    to [out=down, in=up, in looseness=1.3] (2.6,0);
\draw [thick] (1,0)
    to [out=up, in=down, in looseness=1.3] (2.6,2)
    to [out=up, in=down] (2,4);
\draw [thick] (2,0)
    to [out=up, in=down, in looseness=1.3] (0.4,2)
    to [out=up, in=down] (1,4);
\end{tikzpicture}
\end{aligned}
\quad=\quad
\begin{aligned}
\begin{tikzpicture}[scale=0.6, yscale=-0.85, yscale=\tys]
\draw [fill=\fillClight, draw=none] (0.4,0)
    to [out=up, in=down, out looseness=1.5] (0.4,2)
    to (1,2)
    to [out=down, in=up, out looseness=1.5] (1,0);
\draw [fill=\fillClight, draw=none] (2.6,0)
    to [out=up, in=down, out looseness=1.5] (2.6,2)
    to (2,2)
    to [out=down, in=up, out looseness=1.5] (2,0);
\draw [fill=\fillClight, draw=none] (0.4,2)
    to (1,2)
    to [out=up, in=up, looseness=1.5] (2,2)
    to (2.6,2)
    to [out=up, in=down] (2,4)
    to (1,4) to [out=down, in=up] (0.4,2);
\draw [thick, opacity=1] (2.6,0)
    to [out=up, in=down, in looseness=1.5] (2.6,2)
    to [out=up, in=down] (2,4);
\draw [thick] (0.4,0)
    to [out=up, in=down, in looseness=1.5] (0.4,2)
    to [out=up, in=down] (1,4);
\draw [thick] (2,0) 
    to (2,2)
    to [out=up, in=up, looseness=1.5] (1,2)
    to (1,0);
\end{tikzpicture}
\end{aligned}
\end{calign}
All rotations and mirrored versions of the last of these axioms are also imposed. The net effect of these axioms is that any two connected networks of copying, comparison, creation and deletion operations, with the same number of inputs and the same number of outputs, will be equal. It follows that every such region carries the structure of a commutative dagger-Frobenius algebra in a canonical way. Note that the symmetric monoidal 2\-category structure is used crucially in the last equation here, allowing one region to pass above another.
\begin{definition}
In a symmetric monoidal 2\-category, an object has a \emph{topological boundary} if it is equipped with the data~\eqref{eq:top1}--\eqref{eq:top3} satisfying equations~\eqref{eq:topeq1}--\eqref{eq:topeq2}.
\end{definition}

\noindent
We will assume throughout that we are working with dagger 2\-categories whose objects are equipped with topological boundaries.

\subsection{Controlled operations}
In this paper, a key role will be played by the concept of a \textit{controlled family of measurements} which we define here in a new way. This has the following definition in the 2\-categorical formalism.
\begin{definition}
\label{def:controlledfamily}
A \textit{controlled family of measurements} is a unitary 2\-cell of the following type:
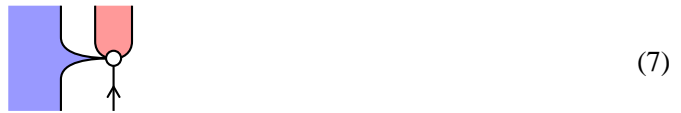
\begin{equation}
\label{eq:controlledexample}
\begin{aligned}
\begin{tikzpicture} [thick,scale=0.7, xscale=-1, yscale=\tys]
\node (A)[Vertex, scale=\vertexsize, vertex colour=white] at (0,3.25){};
\draw  [arrow={0.5}]  (0,2)  to (0,3.25);
\draw [fill=\fillcomp, fill opacity=0.8] (-0.35, 4.5)
to (-0.35, 3.65)
to [out=down, in=down, looseness=2] (0.35, 3.65)
to (0.35, 4.5);
\draw [fill=\fillC, fill opacity=0.8, draw=none](2, 2) to (1,2)
to (1, 2)
to (1, 2.75)
to [out=up, in=right](0,3.25)
to [out=right, in=down](1, 3.75)
to (1, 4.5)
to (2,4.5);
\draw (1, 2)
to (1, 2.75)
to [out=up, in=right](0,3.25)
to [out=right, in=down](1, 3.75)
to (1, 4.5);
\end{tikzpicture}
\end{aligned}
\end{equation}
\end{definition}

\noindent
The left-hand pool of classical information represents classical data that will determine the measurement basis, which we will always draw in blue. The line at the bottom-right of the diagram represents the quantum system to be measured. The upper-right pool of classical information represents the classical result of the measurement, which we will always draw in red.

We interpret these measurements as perfectly fine-grained (that is, non-degenerate), and projective. The motivation for the definition above is made clear by analyzing its models in \cat{2Hilb}.
\begin{lemma}
In \cat{2Hilb}, a controlled family of measurements corresponds precisely to a Hilbert space equipped with a list of orthonormal bases.
\end{lemma}
\begin{proof}
Given a 2\-cell $\zeta$ of the type~\eqref{eq:controlledexample} in \cat{2Hilb}, write $n$ for the dimension of the blue object, and $m$ for the dimension of the red object. Then it is immediate that $\zeta$ constitutes a list of length $n$, whose entries are $m$-by-$m$ matrices~\cite{Vicary:2012hqt}. For $\zeta$ to be unitary means exactly that each $m$-by-$m$ matrix is unitary. So we have a list of $n$ unitary operators. However, the red region comes equipped with a canonical commutative dagger-Frobenius algebra structure, and hence a canonical orthonormal basis. Writing the unitaries in terms of this basis, it is clear that the data of $\zeta$ is canonically equivalent to a list of $n$ orthonormal bases for the incoming $m$-dimensional Hilbert space.
\end{proof}

\noindent
The unitarity property of~\eqref{eq:controlledexample} takes the following graphical form:
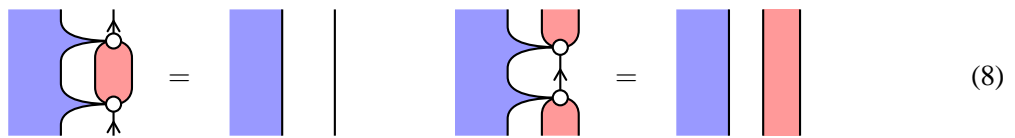
\begin{equation}
\label{eq:controlledmeasurementunitarity}
\begin{aligned}
\begin{tikzpicture} [thick,scale=0.7, xscale=-1, yscale=\tys]
\node (A)[Vertex, scale=\vertexsize, vertex colour=white] at (0,3.25){};
\node (A)[Vertex, scale=\vertexsize, vertex colour=white] at (0,4.75){};
\draw [thick, arrow={0.5}] (0,2.5) to  (0, 3.25);
\draw [thick, arrow={0.7}] (0, 4.75) to  (0,5.5);
\draw [fill=\fillcomp, fill opacity=0.8] (-0.35, 4.35)
to (-0.35, 3.65)
to [out=down, in=down, looseness=2] (0.35, 3.65)
to (0.35, 4.35)
to [out=up, in=up, looseness=2](-0.35, 4.35);
\draw [fill=\fillC, fill opacity=0.8, draw=none](2, 2.5) to (1,2.5)
to (1, 2.5)
to (1, 2.75)
to [out=up, in=right](0,3.25)
to [out=right, in=down](1, 3.75)
to (1, 4.25)
to [out=up, in=right](0, 4.75)
to [out=right, in=down](1, 5.25)
to (1, 5.5)
to (2,5.5);
\draw (1,2.5)
to (1, 2.5)
to (1, 2.75)
to [out=up, in=right](0,3.25)
to [out=right, in=down](1, 3.75)
to (1, 4.25)
to [out=up, in=right](0, 4.75)
to [out=right, in=down](1, 5.25)
to (1, 5.5);
\end{tikzpicture}
\end{aligned}
\quad=\quad
\begin{aligned}
\begin{tikzpicture} [thick,scale=0.7, xscale=-1, yscale=\tys]
\draw [thick] (0,2.5) to (0,5.5);
\draw [fill=\fillC, fill opacity=0.8, draw=none](2, 2.5) to (1,2.5)
to (1, 2.5)
to (1, 2.75)
to (1, 4.25)
to (1, 5.5)
to (2,5.5);
\draw (1,2.5)
to (1, 2.5)
to (1, 2.75)
to (1, 4.25)
to (1, 5.5);
\end{tikzpicture}
\end{aligned}
\qquad\qquad
\begin{aligned}
\begin{tikzpicture} [thick,scale=0.7, xscale=-1, yscale=\tys]
\node (A)[Vertex, scale=\vertexsize, vertex colour=white] at (0,3.4){};
\node (A)[Vertex, scale=\vertexsize, vertex colour=white] at (0,4.6){};
\draw [thick, arrow={0.6}] (0,3.4) to (0,4.6);
\draw [fill=\fillcomp, fill opacity=0.8] (-0.35, 5.5)
to (-0.35, 5)
to [out=down, in=down, looseness=2] (0.35, 5)
to (0.35, 5.5);
\draw [fill=\fillcomp, fill opacity=0.8] (0.35, 2.5)
to (0.35, 3)
to [out=up, in=up, looseness=2](-0.35, 3)
to (-0.35, 2.5);
\draw [fill=\fillC, fill opacity=0.8, draw=none](2, 2.5) to (1,2.5)
to (1, 2.5)
to (1, 2.9)
to [out=up, in=right](0,3.4)
to [out=right, in=down](1, 3.9)
to (1, 4.1)
to [out=up, in=right](0, 4.6)
to [out=right, in=down](1, 5.1)
to (1, 5.5)
to (2,5.5);
\draw (1,2.5)
to (1, 2.5)
to (1, 2.9)
to [out=up, in=right](0,3.4)
to [out=right, in=down](1, 3.9)
to (1, 4.1)
to [out=up, in=right](0, 4.6)
to [out=right, in=down](1, 5.1)
to (1, 5.5);
\end{tikzpicture}
\end{aligned}
\quad=\quad
\begin{aligned}
\begin{tikzpicture} [thick,scale=0.7, xscale=-1, yscale=\tys]
\draw [fill=\fillcomp, fill opacity=0.8, draw=none] (-0.35, 5.5)
to (-0.35, 2.5)
to (0.35, 2.5)
to (0.35, 5.5);
\draw (-0.35, 2.5) to (-0.35, 5.5);
\draw (0.35, 2.5) to (0.35, 5.5);
\draw [fill=\fillC, fill opacity=0.8, draw=none](2, 2.5) to (1,2.5)
to (1, 2.5)
to (1, 2.75)
to (1, 4.25)
to (1, 5.5)
to (2,5.5);
\draw (1,2.5)
to (1, 2.5)
to (1, 2.75)
to (1, 4.25)
to (1, 5.5);
\end{tikzpicture}
\end{aligned}
\end{equation}

\begin{definition}[Conjugate measurement bases]
Following the standard conventions, a controlled measurement with respect to the \textit{conjugate} set of bases is represented by mirroring the diagram about a vertical axis:
\begin{equation}
\label{eq:conjugatemeasurement}
\begin{aligned}
\begin{tikzpicture} [thick,scale=0.7, yscale=\tys]
\node (A)[Vertex, scale=\vertexsize, vertex colour=white] at (0,3.25){};
\draw  [reversed arrow={0.5}]  (0,2)  to (0,3.25);
\draw [fill=\fillcomp, fill opacity=0.8] (-0.35, 4.5)
to (-0.35, 3.65)
to [out=down, in=down, looseness=2] (0.35, 3.65)
to (0.35, 4.5);
\draw [fill=\fillC, fill opacity=0.8, draw=none](2, 2) to (1,2)
to (1, 2)
to (1, 2.75)
to [out=up, in=right](0,3.25)
to [out=right, in=down](1, 3.75)
to (1, 4.5)
to (2,4.5);
\draw (1, 2)
to (1, 2.75)
to [out=up, in=right](0,3.25)
to [out=right, in=down](1, 3.75)
to (1, 4.5);
\end{tikzpicture}
\end{aligned}
\quad:=\quad
\left(
\begin{aligned}
\begin{tikzpicture} [thick,scale=0.7, xscale=-1, yscale=\tys]
\node (A)[Vertex, scale=\vertexsize, vertex colour=white] at (0,3.25){};
\draw  [arrow={0.5}]  (0,2)  to (0,3.25);
\draw [fill=\fillcomp, fill opacity=0.8] (-0.35, 4.5)
to (-0.35, 3.65)
to [out=down, in=down, looseness=2] (0.35, 3.65)
to (0.35, 4.5);
\draw [fill=\fillC, fill opacity=0.8, draw=none](2, 2) to (1,2)
to (1, 2)
to (1, 2.75)
to [out=up, in=right](0,3.25)
to [out=right, in=down](1, 3.75)
to (1, 4.5)
to (2,4.5);
\draw (1, 2)
to (1, 2.75)
to [out=up, in=right](0,3.25)
to [out=right, in=down](1, 3.75)
to (1, 4.5);
\end{tikzpicture}
\end{aligned}
\right) _*
\quad\equiv\quad
\begin{aligned}
\begin{tikzpicture} [thick,scale=0.7, yscale=\tys]
\node (A)[Vertex, scale=\vertexsize, vertex colour=white] at (0,3.25){};
\draw [thick, arrow={0.7}] (0,3.25) to (0,3.5)
to [out=up, in=up, looseness=2](0.5, 3.5)
to  (0.5,3) to (0.5,2);
\draw [fill=\fillcomp, fill opacity=0.8] (-1.25, 4.5)
to (-1.25, 2.85)
to [out=down, in=down, looseness=1.1](-0.25, 2.85)
to [out=up, in=up, looseness=2] (0.25, 2.85)
to [out=down, in=down, looseness=1.1](-1.75, 2.85)
to (-1.75, 4.5);
\draw [fill=\fillC, fill opacity=0.8, draw=none](2, 2) to (1,2)
to (1, 2)
to (1, 3.5)
to [out=up, in=up, looseness=1.2](-0.5, 3.5)
to [out=down, in=left](0, 3.25)
to [out=left, in=up](-0.5, 3)
to [out=down, in=down, looseness=2](-1, 3)
to (-1, 4.5)
to (2,4.5);
\draw (1,2)
to (1, 2)
to (1, 3.5)
to [out=up, in=up, looseness=1.2](-0.5, 3.5)
to [out=down, in=left](0, 3.25)
to [out=left, in=up](-0.5, 3)
to [out=down, in=down, looseness=2](-1, 3)
to (-1, 4.5);
\end{tikzpicture}
\end{aligned}
\end{equation}
Here we decompose the conjugation operation into a composition of adjoint and transpose operations. The blue classical data controlling the choice of basis is now naturally on the right-hand side.
\end{definition}

However, we may want to change the side of the classical data controlling the choice of basis \textit{without} passing to the conjugate set of bases. To do this, we use the symmetric monoidal 2\-category structure to directly move the blue classical region to the other side.  In order to distinguish this from the conjugate controlled measurement~\eqref{eq:conjugatemeasurement}, we represent it as a black vertex.
\begin{definition}[Control from the other side]
\label{def:flipcontrol}
We use a black vertex to indicate control of the measurement and encoding vertices from the other side:
\begin{calign}
\begin{aligned}
\begin{tikzpicture} [thick,scale=0.7, xscale=-1, yscale=-1, yscale=\tys]
\node (A)[Vertex, scale=\vertexsize] at (0,3.25){};
\draw  [arrow={0.4}]  (0,2)  to (0,3.25);
\draw [fill=\fillcomp, fill opacity=0.8] (-0.35, 4.5)
to (-0.35, 3.65)
to [out=down, in=down, looseness=2] (0.35, 3.65)
to (0.35, 4.5);
\draw [fill=\fillC, fill opacity=0.8, draw=none](2, 2) to (1,2)
to (1, 2)
to (1, 2.75)
to [out=up, in=right](0,3.25)
to [out=right, in=down](1, 3.75)
to (1, 4.5)
to (2,4.5);
\draw (1, 2)
to (1, 2.75)
to [out=up, in=right](0,3.25)
to [out=right, in=down](1, 3.75)
to (1, 4.5);
\end{tikzpicture}
\end{aligned}
\quad:=\quad
\begin{aligned}
\begin{tikzpicture} [thick,scale=0.7, scale=-1, yscale=\tys]
\node (A)[Vertex, scale=\vertexsize, vertex colour=white] at (0,3.25){};
\draw [fill=\fillC, fill opacity=0.8, draw=none](2, 2)to (1,2)
to (1, 2)
to [out=up, in=right](0,2.5)
to [out=left, in=left, looseness=2](0, 4)
to [out=right, in=down](1, 4.5)
to (2,4.5);
\draw [fill=\fillC, fill opacity=0.8, draw](0,2.5)
to [out=left, in=left, looseness=3](0,3.25)
to [out=left, in=left, looseness=3](0, 4)
to [out=left, in=left, looseness=2](0, 2.5);
\draw (1, 2)
to [out=up, in=right](0,2.5)
to [out=left, in=left, looseness=3](0,3.25)
to [out=left, in=left, looseness=3](0, 4)
to [out=right, in=down](1, 4.5);
\draw  [arrow={0.4}]  (0,2)  to (0,3.25);
\draw [fill=\fillcomp, fill opacity=0.8] (-0.35, 4.5)
to (-0.35, 3.65)
to [out=down, in=down, looseness=2](0.35, 3.65)
to (0.35, 4.5);
\end{tikzpicture}
\end{aligned}
&
\begin{aligned}
\begin{tikzpicture} [thick,scale=0.7, yscale=\tys]
\node (A)[Vertex, scale=\vertexsize] at (0,3.25){};
\draw  [arrow={0.4}]  (0,2)  to (0,3.25);
\draw [fill=\fillcomp, fill opacity=0.8] (-0.35, 4.5)
to (-0.35, 3.65)
to [out=down, in=down, looseness=2] (0.35, 3.65)
to (0.35, 4.5);
\draw [fill=\fillC, fill opacity=0.8, draw=none](2, 2) to (1,2)
to (1, 2)
to (1, 2.75)
to [out=up, in=right](0,3.25)
to [out=right, in=down](1, 3.75)
to (1, 4.5)
to (2,4.5);
\draw (1, 2)
to (1, 2.75)
to [out=up, in=right](0,3.25)
to [out=right, in=down](1, 3.75)
to (1, 4.5);
\end{tikzpicture}
\end{aligned}
\quad:=\quad
\begin{aligned}
\begin{tikzpicture} [thick,scale=0.7, yscale=\tys]
\node (A)[Vertex, scale=\vertexsize, vertex colour=white] at (0,3.25){};
\draw [fill=\fillC, fill opacity=0.8, draw=none](2, 2)to (1,2)
to (1, 2)
to [out=up, in=right](0,2.5)
to [out=left, in=left, looseness=2](0, 4)
to [out=right, in=down](1, 4.5)
to (2,4.5);
\draw [fill=\fillC, fill opacity=0.8, draw](0,2.5)
to [out=left, in=left, looseness=3](0,3.25)
to [out=left, in=left, looseness=3](0, 4)
to [out=left, in=left, looseness=2](0, 2.5);
\draw (1, 2)
to [out=up, in=right](0,2.5)
to [out=left, in=left, looseness=3](0,3.25)
to [out=left, in=left, looseness=3](0, 4)
to [out=right, in=down](1, 4.5);
\draw  [arrow={0.4}]  (0,2)  to (0,3.25);
\draw [fill=\fillcomp, fill opacity=0.8] (-0.35, 4.5)
to (-0.35, 3.65)
to [out=down, in=down, looseness=2](0.35, 3.65)
to (0.35, 4.5);
\end{tikzpicture}
\end{aligned}
\end{calign}
\end{definition}

\noindent
Arrows indicating dual objects will generally be omitted for simplicity.

\subsection{Projectors}

We will often need to constrain the value held by a pool of classical data, which we do with projectors of different kinds.

\begin{definition}
Given an object $\cat C \equiv \cat{Hilb}^n$ in \cat{2Hilb}, a \textit{classical data projector} is an element of the canonical $n$-element basis for the vector space $\mathrm{Hom} _{\cat{2Hilb}} (\id _\cat C, \id _\cat C)$.
\end{definition}

\noindent
These projectors act to constrain the classical data stored in a region to a particular value. We write them as floating labels that decorate our regions. The following lemma establishes some of their key properties.
\begin{lemma}[Properties of classical data projectors]
\label{lem:projectorproperties}
For diagrams in \cat{2Hilb}, we can use classical data projectors to decompose the identity, and two adjacent projectors annihilate unless they are identical:
\begin{calign}
\begin{aligned}
\begin{tikzpicture}[scale=0.8, yscale=\tys]
\draw [draw=none, fill=\fillC, fill opacity=0.8] (0,0) rectangle (2,2);
\end{tikzpicture}
\end{aligned}
\quad=\quad
\sum _{a=1} ^n \,
\begin{aligned}
\begin{tikzpicture}[scale=0.8, yscale=\tys]
\draw [draw=none, fill=\fillC, fill opacity=0.8] (0,0) rectangle (2,2);
\node at (1,1) {$a$};
\end{tikzpicture}
\end{aligned}
&
\begin{aligned}
\begin{tikzpicture}[scale=0.8, yscale=\tys]
\draw [draw=none, fill=\fillC, fill opacity=0.8] (0,0) rectangle (2,2);
\node at (1,1) {$a\hspace{0.4cm}b$};
\end{tikzpicture}
\end{aligned}
\quad=\quad
\delta _{a,b} \,
\begin{aligned}
\begin{tikzpicture}[scale=0.8, yscale=\tys]
\draw [draw=none, fill=\fillC, fill opacity=0.8] (0,0) rectangle (2,2);
\node at (1,1) {$a$};
\end{tikzpicture}
\end{aligned}
\end{calign}
The projectors can move freely around within regions, much like scalars in the theory of monoidal categories. Furthermore, labelled regions can be connected and disconnected arbitrarily:
\begin{calign}
\begin{aligned}
\begin{tikzpicture}[thick, scale=0.8, yscale=\tys]
\draw [draw=none, fill=\fillC, fill opacity=0.8] (0,0) to (0,2) to (0.5,2) to [out=down, in=down, looseness=1.5] (1.5,2) to (2,2) to (2,0) to (1.5,0) to [out=up, in=up, looseness=1.5] (0.5,0) to (0,0);
\draw [fill=white, fill opacity=1] (0.5,2) to [out=down, in=down, looseness=1.5] (1.5,2);
\draw [fill=white, fill opacity=1] (0.5,0) to [out=up, in=up, looseness=1.5] (1.5,0);
\node at (1,1) {$a$};
\end{tikzpicture}
\end{aligned}
\quad=\quad
\begin{aligned}
\begin{tikzpicture}[thick , scale=0.8, yscale=\tys]
\draw [draw=none, fill=\fillC, fill opacity=0.8] (0,0) rectangle (2,2);
\draw [draw=none, fill=white, fill opacity=1] (0.5,2) rectangle (1.5,0);
\draw (0.5,2) to (0.5,0);
\draw (1.5,2) to (1.5,0);
\node at (1.75,1) {$a$};
\node at (0.25,1) {$a$};
\end{tikzpicture}
\end{aligned}
&
\begin{aligned}
\begin{tikzpicture}[thick, scale=0.8, yscale=\tys]
\draw [draw=none, fill=\fillC, fill opacity=0.8] (0,0) rectangle (1,2);
\node at (0.5,1) {$a$};
\draw (0,0) to (0,2);
\draw (1,0) to (1,2);
\end{tikzpicture}
\end{aligned}
\quad=\quad
\begin{aligned}
\begin{tikzpicture}[thick, scale=0.8, yscale=\tys]
\draw [fill=\fillC, fill opacity=0.8] (0,0) to [out=up, in=up, looseness=2] (1,0);
\draw [fill=\fillC, fill opacity=0.8] (0,2) to [out=down, in=down, looseness=2] (1,2);
\node at (0.5,0.25) {$a$};
\node at (0.5,1.75) {$a$};
\end{tikzpicture}
\end{aligned}
\end{calign}
\end{lemma}
\begin{proof}
Straightforward, but omitted for reasons of space.
\end{proof}

We can also define a different type of projector, which constrains the values of two separate regions of classical data to be the same, or to be different. We will always apply these projectors to regions that are coloured blue in our notation. There will always be exactly 2 blue regions in every diagram where we use the projectors, so it will be unambiguous to which regions they `attach'.
\begin{definition}[Same-value and different-value projectors]
In a symmetric monoidal 2\-category whose hom-categories are \cat{Ab}-enriched, for an object with topological boundary, the \textit{same-value projector} $\Ps$ and \textit{different-value projector} $\Pd$ are defined as follows:
\begin{align}
\Ps \quad&:=\quad \begin{aligned}
\begin{tikzpicture}
 \draw [fill=\fillC, fill opacity=0.8, draw=none] (0.2,0.6)
        to(0.5,0.6)
        to [out=down, in=down, looseness=1.5](1.5,0.6)
        to(1.8,0.6)
        to(1.8, -0.6)
        to(1.5, -0.6)
        to [out=up, in=up, looseness=1.5](0.5, -0.6)
        to(0.2, -0.6);
    \draw [thick] (0.5,0.6)
        to [out=down, in=down, looseness=1.5] (1.5,0.6);    
    \draw [thick] (0.5,-0.6)
        to [out=up, in=up, looseness=1.5] (1.5, -0.6);
\end{tikzpicture}
\end{aligned}
\\
\Pd \quad&:=\quad \begin{aligned}
\begin{tikzpicture}
    \draw [fill=\fillC, fill opacity=0.8, draw=none] (1.5,0.6)
        to(1.8,0.6)
        to(1.8, -0.6)
        to(1.5, -0.6);
    \draw [fill=\fillC, fill opacity=0.8, draw=none] (0.2,0.6)
        to(0.5,0.6)
        to(0.5, -0.6)
        to(0.2, -0.6);  
    \draw [thick] (0.5,0.6)
        to (0.5,-0.6);     
    \draw [thick] (1.5,0.6)
        to (1.5, -0.6);
\end{tikzpicture}
\end{aligned}
\quad-\quad
\begin{aligned}
\begin{tikzpicture}
 \draw [fill=\fillC, fill opacity=0.8, draw=none] (0.2, 0.6)
        to(0.5, 0.6)
        to [out=down, in=down, looseness=1.5](1.5, 0.6)
        to(1.8, 0.6)
        to(1.8, -0.6)
        to(1.5, -0.6)
        to [out=up, in=up, looseness=1.5](0.5, -0.6)
        to(0.2, -0.6);
    \draw [thick] (0.5, 0.6)
        to [out=down, in=down, looseness=1.5] (1.5, 0.6);
    \draw [thick] (0.5,-0.6)
        to [out=up, in=up, looseness=1.5] (1.5, -0.6);
\end{tikzpicture}
\end{aligned}
\end{align}
\end{definition}

\noindent
The main 2\-category we are concerned with is \cat{2Hilb}, in which hom-categories are indeed $\cat{Ab}$-enriched.
\begin{lemma}
The projectors $\Ps$ and $\Pd$ satisfy $\Ps ^2 = \Ps$, $\Pd ^2 = \Pd$, $\Ps \circ \Pd = \Pd \circ \Ps = 0$ and $\Ps + \Pd = \id$.
\end{lemma}
\begin{proof}
Straightforward graphical proof, omitted for reasons of space.
\end{proof}

\noindent
We will assume throughout that we are working in an \cat{Ab}-enriched 2\-category, and so these projectors are well-defined. The tensor product, vertical and horizontal composition in the 2\-category all distribute over the additive structure introduced by these projectors. Distributivity of 2\-cell composition with respect to addition is illustrated by the following. Note, that we can apply the projectors $\Ps$ , $\Pd$ whenever the appropriate regions have any open boundary:
\begin{equation}
\Pd\,
\begin{aligned}
\begin{tikzpicture}[scale=0.4, thick]
\draw [fill=\fillC, fill opacity=0.8, draw=none] (0, 0) to (0,3)
to (1,3)
to (1,0) ;
\draw [fill=\fillC, fill opacity=0.8, draw=none] (2, 0) to (2,3)
to (3,3)
to (3,0) ;
\draw (0,0) to (0,3);
\draw (1,0) to (1,3);
\draw (2,0) to (2,3);
\draw (3,0) to (3,3);
\end{tikzpicture}
\end{aligned}
\quad=\quad
\left(
 \begin{aligned}
\begin{tikzpicture}
    \draw [fill=\fillC, fill opacity=0.8, draw=none] (1.5,0.6)
        to(1.8,0.6)
        to(1.8, -0.6)
        to(1.5, -0.6);
    \draw [fill=\fillC, fill opacity=0.8, draw=none] (0.2,0.6)
        to(0.5,0.6)
        to(0.5, -0.6)
        to(0.2, -0.6);  
    \draw [thick] (0.5,0.6)
        to (0.5,-0.6);     
    \draw [thick] (1.5,0.6)
        to (1.5, -0.6);
\end{tikzpicture}
\end{aligned}
\quad-\quad
\begin{aligned}
\begin{tikzpicture}
 \draw [fill=\fillC, fill opacity=0.8, draw=none] (0.2, 0.6)
        to(0.5, 0.6)
        to [out=down, in=down, looseness=1.5](1.5, 0.6)
        to(1.8, 0.6)
        to(1.8, -0.6)
        to(1.5, -0.6)
        to [out=up, in=up, looseness=1.5](0.5, -0.6)
        to(0.2, -0.6);
    \draw [thick] (0.5, 0.6)
        to [out=down, in=down, looseness=1.5] (1.5, 0.6);
    \draw [thick] (0.5,-0.6)
        to [out=up, in=up, looseness=1.5] (1.5, -0.6);
\end{tikzpicture}
\end{aligned}
\right)
\,\,\circ\,\,
\begin{aligned}
\begin{tikzpicture}[scale=0.4, thick]
\draw [fill=\fillC, fill opacity=0.8,draw=none] (0, 0) to (0,3)
to (1,3)
to (1,0) ;
\draw [fill=\fillC, fill opacity=0.8,draw=none] (2, 0) to (2,3)
to (3,3)
to (3,0) ;
\draw (0,0) to (0,3);
\draw (1,0) to (1,3);
\draw (2,0) to (2,3);
\draw (3,0) to (3,3);
\end{tikzpicture}
\end{aligned}
\quad=\quad
\begin{aligned}
\begin{tikzpicture}[scale=0.4, thick]
\draw [fill=\fillC, fill opacity=0.8, draw=none] (0, 0) to (0,3)
to (1,3)
to (1,0) ;
\draw [fill=\fillC, fill opacity=0.8, draw=none] (2, 0) to (2,3)
to (3,3)
to (3,0) ;
\draw (0,0) to (0,3);
\draw (1,0) to (1,3);
\draw (2,0) to (2,3);
\draw (3,0) to (3,3);
\end{tikzpicture}
\end{aligned}
\,\,-\,\,
\begin{aligned}
\begin{tikzpicture}[scale=0.4, thick]
\draw [fill=\fillC, fill opacity=0.8, draw=none] (1,0)
to (1,2)
to [out=up, in=up](2,2) 
to (2,0)
to (3,0)
to (3,3)
to (2,3)
to [out=down, in=down](1,3)
to (0,3)
to (0,0);
\draw (0,0) to (0,3);
\draw (1,0) to (1,2)
to [out=up, in=up](2,2) to (2,0);
\draw (3,0) to (3,3);
\draw (1,3) [out=down, in=down] to (2,3);
\end{tikzpicture}
\end{aligned}
\end{equation}

\subsection{Attaching controlled phases}

\begin{definition}
A \textit{controlled phase} $\phi$ is an unitary endomorphism of a family of boundaries:
\begin{equation}
\begin{aligned}
\begin{tikzpicture}[scale=0.35,thick,yscale=0.8]
\draw [fill=\fillcomp, draw=none, fill opacity=0.8] (-3,-2.75)
to (-1,-2.75)
to (-1,3.25)
to (-3, 3.25);
\draw (-1,3.25) to (-1, -2.75);
\draw [fill=\fillcomp, draw=none, fill opacity=0.8] (3,-2.75)
to (1,-2.75)
to (1,3.25)
to (3, 3.25);
\draw (1,3.25) to (1, -2.75);
\draw [fill=\fillC, draw=none, fill opacity=0.8] (4,-3)
to (2,-3)
to (2,3)
to (4, 3);
\draw (2,3) to (2, -3);
\draw [fill=\fillC, draw=none, fill opacity=0.8] (-4,-3)
to (-2,-3)
to (-2,3)
to (-4, 3);
\draw (-2,3) to (-2, -3);
\node (p)[minimum width=65pt, minimum height=19pt,draw, fill=white, fill opacity=1, scale=0.8] at (0, 0) {$\phi$};
\node (c) [Vertex, scale=0.5, vertex colour=white] at ([xshift=-28.5pt, yshift=1pt]p.south) {};
\node (c) [Vertex, scale=0.5, vertex colour=white] at ([xshift=28.5pt, yshift=1pt]p.south) {};
\node (c) [Vertex, scale=0.5, vertex colour=white] at ([xshift=57pt, yshift=1pt]p.south){};
\node (c) [Vertex, scale=0.5, vertex colour=white] at ([xshift=-57pt, yshift=1pt]p.south) {};
\end{tikzpicture}
\end{aligned}
\end{equation}
The white nodes decorating the 2\-cell $\phi$ indicate the boundaries to which it attaches.
\end{definition}

In \cat{2Hilb}, such a structure gives a controlled phase in the ordinary sense: a family of unit complex numbers, indexed by the values of the classical information of the regions to which the phase is connected. The result of such a controlled phase is to render the overall wavefunction of the system entangled, without introducing any classical statistical correlation between local measurement results~\cite{Vicary:2012hqt}.

\section{Complementary families of measurements}

In Definition~\ref{def:controlledfamily} we introduced a 2\-categorical axiomatization of a controlled family of measurements. In this Section, we add the extra requirement that any two distinct measurements in the family are \textit{complementary}. In this case, we say that we have a \textit{complementary family} of measurements. These play an essential role in quantum key distribution and the Mean King problem, which we study in Sections~\ref{sec:qkd} and~\ref{sec:meanking}.

For a single pair of nondegenerate measurements to be complementary is a standard condition in quantum information, sometimes also known as \textit{unbiasedness}.
\begin{definition}
Two bases $\{\ket {a_i}\}$, $\{\ket{b_j}\}$ of a finite-dimensional Hilbert space $H$ are \textit{complementary}, or \textit{unbiased}, when for all $i$, $j$ we have ${\lvert\langle a_i|b_j\rangle\rvert}^2=\text{dim}(H) ^{-1}$.
\end{definition}

\noindent
A first characterization of this property in terms of monoidal categories was given by Coecke and Duncan~\cite{Duncan:2009}, and a 2\-categorical definition was given in~\cite{Vicary:2012hqt}.

\subsection{Basic definition}

\begin{definition}[Complementary family]
\label{def:controlledcomplementarity}
A \textit{complementary family of measurements}, or simply a \emph{complementary family}, is an ordinary family of measurements as given in Definition~\ref{def:controlledfamily}, such that there exists some unitary 2\-cell $\phi$ satisfying the following equation:
\begin{equation}
\label{eq:complementaryfamily}
\begin{aligned}
\begin{tikzpicture}[scale=0.4, thick]
\node (a) [Vertex, scale=\vertexsize, vertex colour=white]at (0.5,-0.25) {};
\node (b) [Vertex, scale=\vertexsize, vertex colour=white]at (1.5, 2.) {};
\node (c) [Vertex, scale=\vertexsize] at (1.5, 3.00) {};
\draw [fill=\fillcomp, fill opacity=0.8, draw=none] (-1,4)
    to (-1,1.75)
    to [out=down, in=\nwangle, out looseness=1.5] (a.center)
    to [out=\neangle, in=down, in looseness=1.5] (2,1.5)
    to [out=up, in=\seangle] (b.center)
    to [out=\swangle, in=up] +(-0.5,-0.5)
    to [out=down, in=down, looseness=2] +(-1,0)
    to (0.0,4);
\draw (-1,4)
    to (-1,1.75)
    to [out=down, in=\nwangle, out looseness=1.5] (a.center)
    to [out=\neangle, in=down, in looseness=1.5] (2,1.5)
    to [out=up, in=\seangle] (b.center)
    to [out=\swangle, in=up] +(-0.5,-0.5)
    to [out=down, in=down, looseness=2] +(-1,0)
    to (0.0,4);
\draw (b.center)
    to [out=up, in=down] (c.center);  
\draw [fill=\fillcomp, fill opacity=0.8, draw=none] (1,4) to (1,3.5)
    to [out=down, in=\nwangle] (c.center)
    to [out=\neangle, in=down] +(0.5,0.5)
    to (2,4);
\draw (1,4) to (1,3.5)
    to [out=down, in=\nwangle] (c.center)
    to [out=\neangle, in=down] +(0.5,0.5)
    to (2,4);
\draw [fill=\fillC, fill opacity=0.8, draw=none] (-3,4.0)
        to (-2.0,4.0)
        to (-2, 3)
        to [out=down, in=left] (-1,2)
        to [out=left, in=up] (-2, 1)
        to (-2.0, 0.25)
        to [out=down, in=left] (-1, -0.25)
        to [out=left, in=up] (-2.0, -0.75)
        to (-3,-0.75);        
\draw(-2.0,4.0)
        to (-2, 3)
        to [out=down, in=left] (-1,2)
        to [out=left, in=up] (-2, 1)
        to (-2.0, 0.25)
        to [out=down, in=left] (-1, -0.25)
        to [out=left, in=up] (-2.0, -0.75);         
\draw [fill=\fillC, fill opacity=0.8, draw=none] (4, 4.0)
    to  (3.0, 4.0)
    to [out=down, in=right] (c.center)
    to [out=right, in=up] (3.0, 2.0)
    to (3.0, -0.75)
    to (4, -0.75);
\draw(3.0, 4.0)
    to [out=down, in=right] (c.center)
    to [out=right, in=up] (3.0, 2.0)
    to (3.0, -0.75);
\draw(-1,-0.25) to (0.5, -0.25);
\draw(-1,2) to (1.5, 2);
\draw (0.5,-0.75 -| a.center) to (a.center);
\node at (-4,1.5) {$\textcolor{blue}{\Pd}$};
\foreach \x/\ya/\yb/\text in
  {0.21/0/-.18/Measure in left basis,
   0.5/1/0.9/Copy result,
   1.2/2/1.9/Encode in left basis,
   1.2/3/3/Measure in right basis}
{
  \node (z) at (-12,\ya) [anchor=west, font=\scriptsize] {\text\vphantom{p|}};
  \draw [black!70, ->, ultra thin] (z.east) to (\x,\yb);
}
\end{tikzpicture}
\end{aligned}
\quad\,\,=\quad
\frac{\Pd}{n}\,
\begin{aligned}
\begin{tikzpicture}[scale=0.4, thick]
\node (a) [Vertex, scale=\vertexsize, vertex colour=white] at (0.0, 0.1) {};
\draw [fill=\fillcomp, fill opacity=0.8, draw=none] (-0.5,3.5)
    to (-0.5,0.5)
    to [out=down,  in=left] (a.center) 
    to [out=right,  in=down](0.5,0.5)
    to (0.5,3.5);
\draw (-0.5,3.5)
    to (-0.5,0.5)
    to [out=down,  in=left] (a.center) 
    to [out=right,  in=down](0.5,0.5)
    to (0.5,3.5);
\draw [fill=\fillcomp, fill opacity=0.8, draw=none] (1.5,3.5)
    to (1.5,0.5)
    to [out=down,  in=left] (2,0)
    to [out=right, in=down](2.5,0.5)
    to (2.5,3.5);    
\draw (1.5,3.5)
    to (1.5,0.5)
    to [out=down,  in=left] (2,0)
    to [out=right, in=down](2.5,0.5)
    to (2.5,3.5);
\draw (0.0,-1.0) to (a.center);
\draw [fill=\fillC, fill opacity=0.8, draw=none] (-2.5,3.5)
to (-1.5,3.5)
    to (-1.5,1.5)
    to [out=down,  in=left] (-0.5,0)
    to [out=left, in=up] (-1.5, -1)
    to (-2.5, -1);
\draw (-1.5,3.5)
    to (-1.5,1.5)
    to [out=down,  in=left] (-0.5,0)
    to [out=left, in=up] (-1.5, -1);
\draw [fill=\fillC, fill opacity=0.8, draw=none] (4.5,3.5)
    to (3.5,3.5)
    to (3.5,-1)
    to (4.5, -1);
\draw (3.5,3.5)to (3.5,-1);
\draw (-0.5, 0) to (0,0);
\node (p)[minimum width=65pt, draw, fill=white, fill opacity=1] at (1, 2) {$\phi$};
\node [Vertex, vertex colour=white, scale=0.5] at ([xshift=-71pt, yshift=1pt]p.south){};   
\node [Vertex, vertex colour=white, scale=0.5] at ([xshift=-14.5pt, yshift=1pt]p.south){};   
\node [Vertex, vertex colour=white, scale=0.5] at ([xshift=14.5pt, yshift=1pt]p.south){};   
\node [Vertex, vertex colour=white, scale=0.5] at ([xshift=71pt, yshift=1pt]p.south){};   
\foreach \x/\ya/\yb/\text in
  {2/0.25/0/Create random data,
   0.2/1.25/0/Measure in left basis,
   1.5/2.25/2.0/Controlled phase}
{
  \node (z) at (6,\ya) [anchor=west, font=\scriptsize] {\text\vphantom{p|}};
  \draw [black!70, ->, ultra thin] (z.west) to (\x,\yb);
}
\end{tikzpicture}
\end{aligned}
\end{equation}
\end{definition}

\noindent
The black measurement vertex is as defined in~\ref{def:flipcontrol}. Attenuation of the region controlling the measurement choice is a notation allowing to avoid obstructing the rest of the diagram. This definition has an immediate physical motivation. On the left-hand side, we measure a quantum system in some particular basis, copy the result, and then re-encode the result back into a quantum state using the same basis. Then, according to a second basis guaranteed to be different to the first thanks to the projector $\Pd$, we make a new measurement, represented by the black vertex, on the re-encoded state. The right-hand side of the equation says that this entire procedure must be equivalent to doing the original measurement with respect to the original basis, but then choosing the second measurement result uniformly at random, up to the application of some overall phase that allows the wavefunctions to be entangled without introducing any classical correlation.

That this definition is correct for ordinary quantum theory is immediate from previous results on the 2\-categorical characterization of complementary measurements.
\begin{lemma}
In \cat{2Hilb}, the complementary families are exactly Hilbert spaces equipped with a collection of pairwise-complementary orthonormal bases.
\end{lemma}
\begin{proof}
Labelling the left- and right-hand blue regions on each side of equation~\eqref{eq:complementaryfamily} with distinct projectors $a$ and $b$ respectively, we obtain the ordinary 2\-categorical condition for a complementary pair of orthonormal bases~\cite{Vicary:2012hqt}. Conversely, suppose we have a family of orthonormal bases; then by writing the identity as a sum of projectors using Lemma~\ref{lem:projectorproperties}, equation~\eqref{eq:complementaryfamily} follows.
\end{proof}

\subsection{Alternative characterizations}

Here we examine alternative characterizations of the complementary family definition.

\begin{lemma}[Complementarity through unitarity]
\label{lemma:Complementarity through unitarity}
A controlled family of measurements is complementary if and only if the following 2-cell is unitary on the support of the projector $\Pd$:
\begin{align}
\alpha\quad:=\quad
\begin{aligned}
\begin{tikzpicture}[scale=0.4, thick, yscale=0.8]
\node (b) [fill opacity=1, Vertex, scale=\vertexsize, vertex colour=white, fill opacity=1] at (1.5, 2.25) {};
\node (c) [Vertex, scale=\vertexsize] at (1.5, 3.25) {};
\draw [fill=\fillcomp, fill opacity=0.8, draw=none] (-1,5.25)
    to (-1,0.5)
    to  (2,0.5)
    to(2,1.75)
    to [out=up, in=\seangle] (b.center)
    to [out=\swangle, in=up] +(-0.5,-0.5)
    to [out=down, in=down, looseness=2] +(-1,0)
    to (0.0,5.25);
    \draw  (2,0.5)
    to(2,1.75)
    to [out=up, in=\seangle] (b.center)
    to [out=\swangle, in=up] +(-0.5,-0.5)
    to [out=down, in=down, looseness=2] +(-1,0)
    to (0.0,5.25);
\draw (b.center)
    to [out=up, in=down] (c.center);   
\draw [fill=\fillcomp, fill opacity=0.8, draw=none] (1,5.25) to (1,3.75)
    to [out=down, in=\nwangle] (c.center)
    to [out=\neangle, in=down] +(0.5,0.5)
    to (2,4)
    to [out=up, in=left](2.5,4.5)
    to [out=right, in=up](3,4)
    to (3,0.5)
    to(4,0.5)
    to (4,5.25);
 \draw  (1,5.25) to (1,3.75)
    to [out=down, in=\nwangle] (c.center)
    to [out=\neangle, in=down] +(0.5,0.5)
    to (2,4)
    to [out=up, in=left](2.5,4.5)
    to [out=right, in=up](3,4)
    to (3,0.5)   ;
\draw [fill=\fillC, fill opacity=0.8, draw=none, fill opacity=0.8] (-1.5,5)
        to (-0.5,5)
        to (-0.5, 3.25)
        to [out=down, in=left] (0.75,2.25)
        to [out=left, in=up](-0.5, 1.25)
        to (-0.5, 0.25)
        to (-1.5,0.25);  
\draw  (-0.5,5)
        to (-0.5, 3.25)
        to [out=down, in=left] (0.75,2.25)
        to [out=left, in=up](-0.5, 1.25)
        to (-0.5, 0.25);              
\draw  (0.75,2.25) to (1.5,2.25);   
\draw [fill=\fillC, fill opacity=0.8, draw=none, fill opacity=0.8] (4.5, 5)
    to  (3.5, 5)
    to (3.5, 4.25)
    to [out=down, in=right] (2.25,3.25)
    to [out=right, in=up] (3.5, 2.25)
    to (3.5, 0.25)
    to (4.5, 0.25);
\draw (3.5, 5)
    to (3.5, 4.25)
    to [out=down, in=right] (2.25,3.25)
    to [out=right, in=up] (3.5, 2.25)
    to (3.5, 0.25);
    \draw (2.25,3.25) to (1.5,3.25);
\end{tikzpicture}
\end{aligned}
\end{align}
\end{lemma}
\begin{proof}
See Appendix.
\end{proof}

The following alternative formulations of complementarity will prove useful in a later section.
\begin{lemma} [Complementarity condition under horizontal reflection]
\label{Complementarity condition under horizontal reflection}
A family of controlled measurement operations is complementary if and only if the following equation is satisfied:
\begin{equation}
\label{eq:flippedcomplementarity}
\textcolor{blue}{\Pd}
\,
\begin{aligned}
\begin{tikzpicture}[scale=0.5, thick, xscale=-1, yscale=0.8]
\node (a) [Vertex, scale=\vertexsize, vertex colour=black]at (0.5,-0.25) {};
\node (b) [Vertex, scale=\vertexsize, vertex colour=black]at (1.5, 2.) {};
\node (c) [Vertex, scale=\vertexsize, vertex colour=white] at (1.5, 3.00) {};
\draw [fill=\fillcomp, fill opacity=0.8, draw=none] (-1,4)
    to (-1,1.75)
    to [out=down, in=\nwangle, out looseness=1.5] (a.center)
    to [out=\neangle, in=down, in looseness=1.5] (2,1.5)
    to [out=up, in=\seangle] (b.center)
    to [out=\swangle, in=up] +(-0.5,-0.5)
    to [out=down, in=down, looseness=2] +(-1,0)
    to (0.0,4);
\draw (-1,4)
    to (-1,1.75)
    to [out=down, in=\nwangle, out looseness=1.5] (a.center)
    to [out=\neangle, in=down, in looseness=1.5] (2,1.5)
    to [out=up, in=\seangle] (b.center)
    to [out=\swangle, in=up] +(-0.5,-0.5)
    to [out=down, in=down, looseness=2] +(-1,0)
    to (0.0,4);
\draw (b.center)
    to [out=up, in=down] (c.center);  
\draw [fill=\fillcomp, fill opacity=0.8, draw=none] (1,4) to (1,3.5)
    to [out=down, in=\nwangle] (c.center)
    to [out=\neangle, in=down] +(0.5,0.5)
    to (2,4);
\draw (1,4) to (1,3.5)
    to [out=down, in=\nwangle] (c.center)
    to [out=\neangle, in=down] +(0.5,0.5)
    to (2,4);
\draw [fill=\fillC, fill opacity=0.8, draw=none] (-3,4.0)
        to (-2.0,4.0)
        to (-2, 3)
        to [out=down, in=left] (-1,2)
        to [out=left, in=up] (-2, 1)
        to (-2.0, 0.25)
        to [out=down, in=left] (-1, -0.25)
        to [out=left, in=up] (-2.0, -0.75)
        to (-3,-0.75);        
\draw(-2.0,4.0)
        to (-2, 3)
        to [out=down, in=left] (-1,2)
        to [out=left, in=up] (-2, 1)
        to (-2.0, 0.25)
        to [out=down, in=left] (-1, -0.25)
        to [out=left, in=up] (-2.0, -0.75);         
\draw [fill=\fillC, fill opacity=0.8, draw=none] (4, 4.0)
    to  (3.0, 4.0)
    to [out=down, in=right] (c.center)
    to [out=right, in=up] (3.0, 2.0)
    to (3.0, -0.75)
    to (4, -0.75);
\draw(3.0, 4.0)
    to [out=down, in=right] (c.center)
    to [out=right, in=up] (3.0, 2.0)
    to (3.0, -0.75);
\draw(-1,-0.25) to (0.5, -0.25);
\draw(-1,2) to (1.5, 2);
\draw (0.5,-0.75 -| a.center) to (a.center);
\end{tikzpicture}
\end{aligned}
\quad=\quad
\frac{\Pd}{\sqrt{n}}\,
\begin{aligned}
\begin{tikzpicture}[scale=0.5, thick, xscale=-1, yscale=0.8]
\node (a) [Vertex, scale=\vertexsize, vertex colour=black] at (0.0, 0.1) {};
\draw [fill=\fillcomp, fill opacity=0.8, draw=none] (-0.5,3.5)
    to (-0.5,0.5)
    to [out=down,  in=left] (a.center) 
    to [out=right,  in=down](0.5,0.5)
    to (0.5,3.5);
\draw (-0.5,3.5)
    to (-0.5,0.5)
    to [out=down,  in=left] (a.center) 
    to [out=right,  in=down](0.5,0.5)
    to (0.5,3.5);
\draw [fill=\fillcomp, fill opacity=0.8, draw=none] (1.5,3.5)
    to (1.5,0.5)
    to [out=down,  in=left] (2,0)
    to [out=right, in=down](2.5,0.5)
    to (2.5,3.5);    
\draw (1.5,3.5)
    to (1.5,0.5)
    to [out=down,  in=left] (2,0)
    to [out=right, in=down](2.5,0.5)
    to (2.5,3.5);
\draw (0.0,-1.0) to (a.center);
\draw [fill=\fillC, fill opacity=0.8, draw=none] (-2.4,3.5)
to (-1.5,3.5)
    to (-1.5,1.5)
    to [out=down,  in=left] (-0.5,0)
    to [out=left, in=up] (-1.5, -1)
    to (-2.4, -1);
\draw (-1.5,3.5)
    to (-1.5,1.5)
    to [out=down,  in=left] (-0.5,0)
    to [out=left, in=up] (-1.5, -1);
\draw [fill=\fillC, fill opacity=0.8, draw=none] (4.5,3.5)
    to (3.5,3.5)
    to (3.5,-1)
    to (4.5, -1);
\draw (3.5,3.5)to (3.5,-1);
\draw (-0.5, 0) to (0,0);
\node (p)[minimum width=80pt, draw, fill=white, fill opacity=1] at (1, 2) {$\phi^{\dagger}$};
\node [Vertex, vertex colour=white, scale=0.5] at ([xshift=-71pt, yshift=1pt]p.south){};   
\node [Vertex, vertex colour=white, scale=0.5] at ([xshift=-14.5pt, yshift=1pt]p.south){};   
\node [Vertex, vertex colour=white, scale=0.5] at ([xshift=14.5pt, yshift=1pt]p.south){};   
\node [Vertex, vertex colour=white, scale=0.5] at ([xshift=71pt, yshift=1pt]p.south){};   
\end{tikzpicture}
\end{aligned}
\end{equation}
\end{lemma}
\begin{proof}
By Lemma~\ref{lemma:Complementarity through unitarity}, both $\alpha$ and $\alpha^{\dagger}$ are unitary on the support of the projector $\Pd$. The condition~\ref{eq:flippedcomplementarity} can be obtained by elementary 2\-cell operations from the unitarity of $\alpha^{\dagger}$.
\end{proof}

\begin{lemma}[Alternative formulation of complementarity]
\label{Alternative formulation of complementarity}
A controlled family of measurement operations is complementary if and only if the following condition is satisfied:
\begin{equation}
\label{th:controlled}
\Pd\,
\begin{aligned}
\begin{tikzpicture} [thick, scale=0.4, yscale=0.7]
\draw (1,1) to (1,-1);
\draw (-1,-3) to (-1,3);
\draw [fill=\fillcomp, fill opacity=0.8, draw=none] (1.5, 5)
to (1.5, 1.6)
to [out=down, in=down, looseness=2](0.5, 1.6)
to (0.5, 3.6)
to [out=up, in=up, looseness=2](-0.5, 3.6)
to [out=down, in=down, looseness=2] (-1.5, 3.6)
to (-1.5, 5);
\draw (1.5, 5)
to (1.5, 1.6)
to [out=down, in=down, looseness=2](0.5, 1.6)
to (0.5, 3.6)
to [out=up, in=up, looseness=2](-0.5, 3.6)
to [out=down, in=down, looseness=2] (-1.5, 3.6)
to (-1.5, 5);
\draw [fill=\fillcomp, fill opacity=0.8, draw=none] (-1.5, -5)
to (-1.5, -3.6)
to [out=up, in=up, looseness=2](-0.5, -3.6)
to [out=down, in=down, looseness=2](0.5, -3.6)
to (0.5, -1.6)
to [out=up, in=up, looseness=2](1.5, -1.6)
to (1.5, -5);
\draw (-1.5, -5)
to (-1.5, -3.6)
to [out=up, in=up, looseness=2](-0.5, -3.6)
to [out=down, in=down, looseness=2](0.5, -3.6)
to (0.5, -1.6)
to [out=up, in=up, looseness=2](1.5, -1.6)
to (1.5, -5);
\draw [fill=\fillC, fill opacity=0.8, draw=none] (3.5, -5)
to (2.5,-5)
to (2.5,0)
to [out=up, in=right](1, 1)
to [out=right, in=down](2.5,2)
to [out=up, in=right](1,3)
to [out=right, in=down](2.5,4)
to (2.5,5)
to (3.5, 5);
\draw (2.5,-5)
to (2.5,0)
to [out=up, in=right](1, 1)
to [out=right, in=down](2.5,2)
to [out=up, in=right](1,3)
to [out=right, in=down](2.5,4)
to (2.5,5);
\draw [fill=\fillC, fill opacity=0.8, draw=none] (-3.5, -5)
to (-2.5,-5)
to (-2.5,-4)
to [out=up, in=left](-1,-3)
to [out=left, in=down](-2.5, -2)
to [out=up, in=left](-1,-1)
to [out=left, in=down](-2.5,0)
to (-2.5, 5)
to (-3.5, 5);
\draw  (-2.5,-5)
to (-2.5,-4)
to [out=up, in=left](-1,-3)
to [out=left, in=down](-2.5, -2)
to [out=up, in=left](-1,-1)
to [out=left, in=down](-2.5,0)
to (-2.5, 5);
\draw (1.5,3) to (-1,3);
\draw (1,-1) to (-1.5,-1);
\node [Vertex, scale=\vertexsize, vertex colour=white] at (1,1) {};
\node [Vertex, scale=\vertexsize, vertex colour=white] at (-1,-3) {};
\node [Vertex, scale=\vertexsize, vertex colour=black] at (1,-1) {};
\node [Vertex, scale=\vertexsize, vertex colour=black] at (-1,3) {};
\end{tikzpicture}
\end{aligned}
\quad=\quad
\frac{\Pd}{n}\,
\begin{aligned}
\begin{tikzpicture}[thick, scale=0.4, xscale=-1, yscale=0.7]
\draw [thick, fill=\fillcomp, fill opacity=0.8] (0.5, 5)
    to (0.5, 3) 
    to [out=down, in=down, looseness=1.5] (-0.5,3)
    to (-0.5, 5);
\draw [thick, fill=\fillcomp, fill opacity=0.8] (0.5, -5)
    to (0.5, -3) 
    to [out=up, in=up, looseness=1.5] (-0.5,-3)
    to (-0.5, -5);    
\draw [fill=\fillC, fill opacity=0.8, draw=none] (2.5,-5) to (1.5,-5)
to (1.5, 5)
to (2.5, 5);
\draw [fill=\fillC, fill opacity=0.8, draw=none] (-2.5,-5) 
to (-1.5,-5)
to (-1.5, 5)
to (-2.5, 5);
\draw (-1.5,-5) to (-1.5, 5);
\draw (1.5,-5) to (1.5, 5);
\end{tikzpicture}
\end{aligned}
\end{equation}
\end{lemma}
\begin{proof}
The result is proved using Lemma~\ref{lemma:Complementarity through unitarity} and by performing topological manipulations.
\end{proof}

\section{Quantum key distribution}
\label{sec:qkd}

In this Section we give 2\-categorical equations defining quantum key distribution (QKD), in both its BB84~\cite{BB84} and E91~\cite{E91} forms. In Theorem~\ref{thm:bb84e91equivalent} we show that these forms are equivalent. Our main result is Theorem~\ref{thm:equivalent}, in which we demonstrate that these quantum key distribution equations are equivalent to Definition~\ref{def:controlledcomplementarity} of a complementary family of measurements.

\subsection{Abstract definitions}

\begin{definition}[BB84 QKD]
\label{BB84QKD}
A controlled family of measurements satisfies \textit{BB84 quantum key distribution} if there exists a unitary 2\-cell $\psi$ satisfying the following equation:
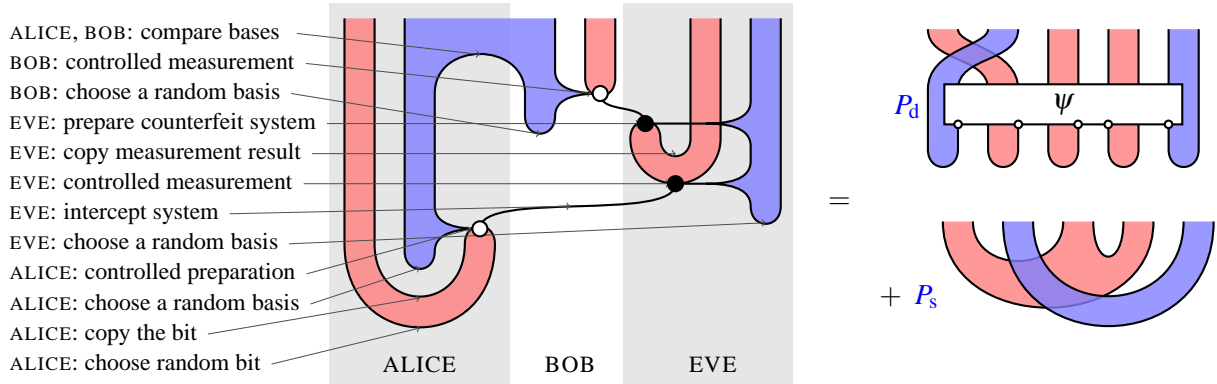
\begin{align*}
\begin{aligned}
\begin{tikzpicture} [scale=0.4,thick,yscale=1]
\draw [fill=black!10, draw=none] (-9,-7.25) rectangle +(6,12.75);
\draw [fill=black!10, draw=none] (0.75,-7.25) rectangle +(5.65,12.75);
\node (A) [Vertex, scale=\vertexsize, vertex colour=white] at (-4,-2) {};
\node (B) [Vertex, scale=\vertexsize] at (2.5,-0.5) {};
\node (C) [Vertex, scale=\vertexsize, vertex colour=white] at (0, 2.5) {};
\node (D) [Vertex, scale=\vertexsize] at (1.5, 1.5) {};
\draw [out=up, in=down, looseness=0.5] (-4,-2) to (2.5,-0.5);
\draw [out=up, in=down] (1.5, 1.5) to (0, 2.5);
\fill [fill=\fillcomp, draw, fill opacity=0.8]  (0.5,5)
to (0.5, 3)
to [out=down, in=down, looseness=2] (-0.5, 3)
to (-0.5,5);
\fill [fill=\fillcomp, draw, fill opacity=0.8]  (3,5) 
to (3,1)
to [out=down, in=down, looseness=2](2,1)
to (2, 1)
to [out=up, in=up, looseness=2] (1,1)
to [out=down, in=left](2.5,-0.5)
to [out=right, in=down](4,1)
to (4,5);
\fill [fill=\fillcomp, fill opacity=0.8, draw] (-7.5, 5)
to (-7.5, -2.5)
to [out=down, in=down, looseness=2](-4.5, -2.5)
to [out=up, in=up, looseness=2](-3.5, -2.5)
to [out=down, in=down, looseness=1.9](-8.5, -2.5)
to (-8.5, 0)
to (-8.5, 5);
\fill [fill=\fillC, fill opacity=0.8, draw] (-6.5, 5)
to (-6.5, -2.75)
to [out=down, in=down, looseness=2](-5.5, -2.75)
to [out=up, in=left] (-4,-2)
to [out=left, in=down] (-5.5, -1.25)
to (-5.5, 2.5)
to [out=up, in=up, looseness=1.5] (-2.5, 2.5)
to (-2.5, 1.75)
to [out=down, in=down, looseness=2] (-1.5,1.75)
to (-1.5,1.75)
to [out=up, in=left] (0, 2.5)
to [out=left, in=down] (-1.5, 3.25)
to [out=up, in=down](-1.5, 5);
\fill [fill=\fillC, fill opacity=0.8, draw] (5,5)
to (5,2.25)
to [out=down, in=right](3.5, 1.5)
to [out=right, in=up](5, 0.75)
to (5, 0.25)
to [out=down, in=right](3.5,-0.5)
to [out=right, in=up](5,-1.25)
to [out=down, in=down, looseness=2](6,-1.25)
to (6,5);
\draw (3.5, -0.5) to (2.5, -0.5);
\draw (3.5, 1.5) to (1.5, 1.5);
\node [] at (-6,-6.5){\sc alice};
\node [] at (-1,-6.5){\sc bob};
\node [] at (3.75,-6.5){\sc eve};
\foreach \x/\ya/\yb/\text in
  {-6/-6.5/-5.3/{\sc alice}: choose random bit,
   -6/-5.5/-4.3/{\sc alice}: copy the bit,
   -6/-4.5/-3.3/{\sc alice}: choose a random basis,
   -4.3/-3.5/-2.0/{\sc alice}: controlled preparation,
   5.5/-2.5/-1.8/{\sc eve}: choose a random basis,
   -1/-1.5/-1.25/{\sc eve}: intercept system,
   2.3/-0.5/-0.5/{\sc eve}: controlled measurement,
   2.5/0.5/0.5/{\sc eve}: copy measurement result,
   1.3/1.5/1.5/{\sc eve}: prepare counterfeit system,
   -2/2.5/1.15/{\sc bob}: choose a random basis,
   -0.3/3.5/2.5/{\sc bob}: controlled measurement,
   -4/4.5/3.8/{\sc alice, bob}: compare bases   }
{
  \node (z) at (-20,\ya) [anchor=west, font=\footnotesize] {\text\vphantom{p|}};
  \draw [black!70, ->, ultra thin] (z.east) to (\x,\yb);
}
\end{tikzpicture}
\end{aligned}
\quad=\hspace{0.3cm}
\begin{matrix}
\Pd\,
\begin{aligned}
\begin{tikzpicture}[scale=0.4,thick]
\draw [fill=\fillcomp, fill opacity=0.8] (-4,5)
to [out=down, in=up](-2, 3)
to (-2,1)
to [out=down, in=down, looseness=2](-1,1)
to (-1,3)
to [out=up, in=down, out looseness=1.2, in looseness=0.5] (-3, 5);
\draw [fill=\fillC, fill opacity=0.8] (-2, 5)
to [out=down, in=up, in looseness=1.2, out looseness=0.5](-4, 3)
to (-4,1)
to [out=down, in=down, looseness=2](-3,1)
to (-3,3)
to [out=up, in=down](-1,5);
\draw [fill=\fillcomp, fill opacity=0.8] (1, 5)
to (1,1)
to [out=down, in=down, looseness=2](0,1)
to (0,5);
\draw [fill=\fillcomp, fill opacity=0.8] (2, 5)
to (2,1)
to [out=down, in=down, looseness=2](3,1)
to (3,5);
\draw [fill=\fillC, fill opacity=0.8] (4, 5)
to (4,1)
to [out=down, in=down, looseness=2](5,1)
to (5,5);
\node (p)[minimum width=90pt, draw, fill=white, fill opacity=1] at (0.5, 2.5) {$\psi$};
\node (c) [Vertex, scale=0.5, vertex colour=white] at ([xshift=-99.5pt, yshift=1pt]p.south) {};
\node (c) [Vertex, scale=0.5, vertex colour=white] at ([xshift=14pt, yshift=1pt]p.south) {};
\node (c) [Vertex, scale=0.5, vertex colour=white] at ([xshift=99.5pt, yshift=1pt]p.south) {};
\node (c) [Vertex, scale=0.5, vertex colour=white] at ([xshift=42.5pt, yshift=1pt]p.south){};
\node (c) [Vertex, scale=0.5, vertex colour=white] at ([xshift=-42.5pt, yshift=1pt]p.south) {};
\end{tikzpicture}
\end{aligned}\\  \\
+ \,\,\,\, \Ps
\begin{aligned}
\begin{tikzpicture}[scale=0.4,thick]
\draw [fill=\fillcomp, draw=none, fill opacity=0.8]
(0, 0) 
to [out=down, in=down, looseness=1.4](7,0)
to (6, 0)
to [out=down, in=down, looseness=4](5, 0)
to (4, 0)
to [out=down, in=down, looseness=1.5](1, 0);
\draw (4, 0) to [out=down, in=down, looseness=1.5](1, 0);
\draw (6, 0) to [out=down, in=down, looseness=4](5, 0);
\draw (0, 0) to [out=down, in=down, looseness=1.4](7,0);
\draw [fill=\fillC, draw=none, fill opacity=0.8] (2, 0) 
to [out=down, in=down, looseness=1.7](9, 0)
to (8,0)
to [out=down, in=down, looseness=1.7](3, 0);
\draw (8, 0) to [out=down, in=down, looseness=1.7](3, 0);
\draw (2, 0) to [out=down, in=down, looseness=1.7](9, 0);
\end{tikzpicture}
\end{aligned}
\end{matrix}
\end{align*}
\end{definition}

\noindent
Each of the diagrams on the right-hand side corresponds to the desired behaviour depending on whether Eve guessed the basis correctly. If Eve guesses incorrectly, then the $\Pd$ term says that both basis choices and all 3 measurement results are classically uncorrelated. If Eve guesses correctly, the then $\Ps$ term says that Eve shares Alice and Bob's basis, and that all 3 agents share the same classical data.

A different equation can be obtained from consideration of the E91 QKD protocol.
\begin{definition}[E91 QKD]
\label{E91QKD}
A controlled family of measurements satisfies \textit{E91 quantum key distribution} if there exists a unitary 2\-cell $\psi$ satisfying the following equation:
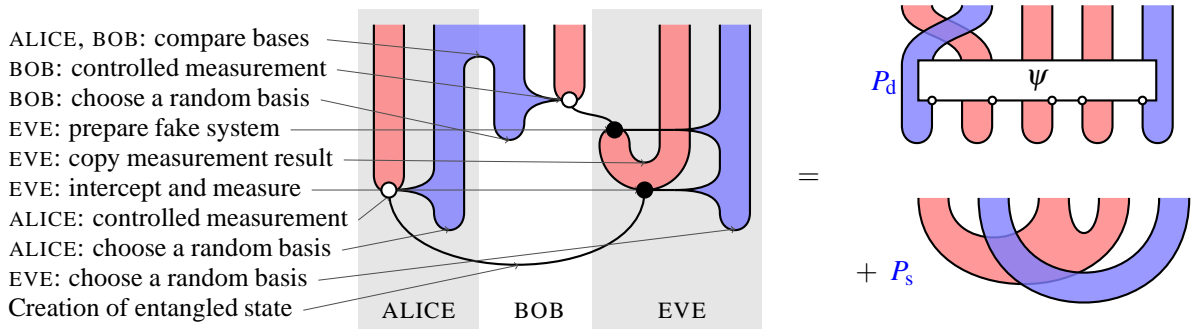
\vspace{-10pt}
\begin{align*}
\begin{aligned}
\begin{tikzpicture} [scale=0.4,thick,yscale=1]
\draw [fill=black!10, draw=none] (-7,-5.25) rectangle +(4,10.75);
\draw [fill=black!10, draw=none] (0.75,-5.25) rectangle +(5.65,10.75);
\node (A) [Vertex, scale=\vertexsize, vertex colour=white] at (-6, -0.5) {};
\node (B) [Vertex, scale=\vertexsize] at (2.5,-0.5) {};
\node (C) [Vertex, scale=\vertexsize, vertex colour=white] at (0, 2.5) {};
\node (D) [Vertex, scale=\vertexsize] at (1.5, 1.5) {};
\draw [out=down, in=down, looseness=1] (-6, -0.5) to (2.5,-0.5);
\draw [out=up, in=down] (1.5, 1.5) to (0, 2.5);
\fill [fill=\fillcomp, draw, fill opacity=0.8]  (0.5,5)
to (0.5, 3)
to [out=down, in=down, looseness=2] (-0.5, 3)
to (-0.5,5);
\fill [fill=\fillcomp, draw, fill opacity=0.8]  (3,5) 
to (3,1)
to [out=down, in=down, looseness=2](2,1)
to (2, 1)
to [out=up, in=up, looseness=2] (1,1)
to [out=down, in=left](2.5,-0.5)
to [out=right, in=down](4,1)
to (4,5);
\fill [fill=\fillcomp, fill opacity=0.8, draw] (-5.5, 5)
to (-5.5, 0)
to [out=down, in=down, looseness=2](-6.5, 0)
to (-6.5, 5);
\fill [fill=\fillC, fill opacity=0.8, draw] (-4.5, 5)
to (-4.5, 0.25)
to [out=down, in=right](-6, -0.5)
to [out=right, in=up](-4.5, -1.25)
to [out=down, in=down, looseness=2](-3.5, -1.25)
to (-3.5, 3.5)
to [out=up, in=up, looseness=1.5] (-2.5, 3.5)
to (-2.5, 1.75)
to [out=down, in=down, looseness=2] (-1.5,1.75)
to (-1.5,1.75)
to [out=up, in=left] (0, 2.5)
to [out=left, in=down] (-1.5, 3.25)
to [out=up, in=down](-1.5, 5);
\fill [fill=\fillC, fill opacity=0.8, draw] (5,5)
to (5,2.25)
to [out=down, in=right](3.5, 1.5)
to [out=right, in=up](5, 0.75)
to (5, 0.25)
to [out=down, in=right](3.5,-0.5)
to [out=right, in=up](5,-1.25)
to [out=down, in=down, looseness=2](6,-1.25)
to (6,5);
\draw (3.5, -0.5) to (2.5, -0.5);
\draw (3.5, 1.5) to (1.5, 1.5);
\node [] at (-5,-4.5){\sc alice};
\node [] at (-1,-4.5){\sc bob};
\node [] at (3.75,-4.5){\sc eve};
\foreach \x/\ya/\yb/\text in
  {-1.7/-4.5/-3/Creation of entangled state,
   5.5/-3.5/-1.8/{\sc eve}: choose a random basis,
   -4/-2.5/-1.8/{\sc alice}: choose a random basis,
   -6/-1.5/-0.7/{\sc alice}: controlled measurement,
   2.3/-0.5/-0.5/{\sc eve}: intercept and measure,
   2.5/0.5/0.4/{\sc eve}: copy measurement result,
   1.3/1.5/1.5/{\sc eve}: prepare fake system,
   -2/2.5/1.15/{\sc bob}: choose a random basis,
   -0.3/3.5/2.5/{\sc bob}: controlled measurement,
   -3/4.5/4/{\sc alice, bob}: compare bases   }
{
  \node (z) at (-19,\ya) [anchor=west, font=\small] {\text\vphantom{p|}};
  \draw [black!70, ->, ultra thin] (z.east) to (\x,\yb);
}
\end{tikzpicture}
\end{aligned}
\quad=\quad
\begin{matrix}
\Pd\begin{aligned}
\begin{tikzpicture}[scale=0.4,thick]
\draw [fill=\fillcomp, fill opacity=0.8] (-4,5)
to [out=down, in=up](-2, 3)
to (-2,1)
to [out=down, in=down, looseness=2](-1,1)
to (-1,3)
to [out=up, in=down, out looseness=1.2, in looseness=0.5] (-3, 5);
\draw [fill=\fillC, fill opacity=0.8] (-2, 5)
to [out=down, in=up, in looseness=1.2, out looseness=0.5](-4, 3)
to (-4,1)
to [out=down, in=down, looseness=2](-3,1)
to (-3,3)
to [out=up, in=down](-1,5);
\draw [fill=\fillcomp, fill opacity=0.8] (1, 5)
to (1,1)
to [out=down, in=down, looseness=2](0,1)
to (0,5);
\draw [fill=\fillcomp, fill opacity=0.8] (2, 5)
to (2,1)
to [out=down, in=down, looseness=2](3,1)
to (3,5);
\draw [fill=\fillC, fill opacity=0.8] (4, 5)
to (4,1)
to [out=down, in=down, looseness=2](5,1)
to (5,5);
\node (p)[minimum width=90pt, draw, fill=white, fill opacity=1] at (0.5, 2.5) {$\psi$};
\node (c) [Vertex, scale=0.5, vertex colour=white] at ([xshift=-99.5pt, yshift=1pt]p.south) {};
\node (c) [Vertex, scale=0.5, vertex colour=white] at ([xshift=14pt, yshift=1pt]p.south) {};
\node (c) [Vertex, scale=0.5, vertex colour=white] at ([xshift=99.5pt, yshift=1pt]p.south) {};
\node (c) [Vertex, scale=0.5, vertex colour=white] at ([xshift=42.5pt, yshift=1pt]p.south){};
\node (c) [Vertex, scale=0.5, vertex colour=white] at ([xshift=-42.5pt, yshift=1pt]p.south) {};
\end{tikzpicture}
\end{aligned}\\  \\
+ \,\,\,\, \Ps
\begin{aligned}
\begin{tikzpicture}[scale=0.4,thick]
\draw [fill=\fillcomp, draw=none, fill opacity=0.8]
(0, 0) 
to [out=down, in=down, looseness=1.4](7,0)
to (6, 0)
to [out=down, in=down, looseness=4](5, 0)
to (4, 0)
to [out=down, in=down, looseness=1.5](1, 0);
\draw (4, 0) to [out=down, in=down, looseness=1.5](1, 0);
\draw (6, 0) to [out=down, in=down, looseness=4](5, 0);
\draw (0, 0) to [out=down, in=down, looseness=1.4](7,0);
\draw [fill=\fillC, draw=none, fill opacity=0.8] (2, 0) 
to [out=down, in=down, looseness=1.7](9, 0)
to (8,0)
to [out=down, in=down, looseness=1.7](3, 0);
\draw (8, 0) to [out=down, in=down, looseness=1.7](3, 0);
\draw (2, 0) to [out=down, in=down, looseness=1.7](9, 0);
\end{tikzpicture}
\end{aligned}
\end{matrix}
\end{align*}
\end{definition}

\begin{theorem}
\label{thm:bb84e91equivalent}
The equations for BB84 and E91 QKD are equivalent.
\end{theorem}
\begin{proof}
Elementary topological manipulation.
\end{proof}

\begin{lemma}[Eve's successful interference]
\label{Eve's successful interference}
On the support of projector $\Ps$, the quantum key distribution specification is satisfied for any controlled family of measurements.
\end{lemma}
\begin{proof}
See Appendix.
\end{proof}

\subsection{Quantum key distribution from a complementary family}

\def\ys{1}
\begin{lemma}
\label{lemma:Complementarity implies QKD}
If a controlled family of measurements is complementary, then it satisfies the quantum key distribution specification with:
\begin{align}
\Pd\,
\begin{aligned}
\begin{tikzpicture}[scale=0.35,thick,yscale=\ys]
\draw [fill=\fillB, draw=none, fill opacity=0.8] (-2,-2.5)
to (-0,-2.5)
to (-0,3.5)
to (-2, 3.5);
\draw (-0,3.5) to (-0, -2.5);
\draw [fill=\fillcomp, draw=none, fill opacity=0.8] (-3,-2.75)
to (-1,-2.75)
to (-1,3.25)
to (-3, 3.25);
\draw (-1,3.25) to (-1, -2.75);
\draw [fill=\fillcomp, draw=none, fill opacity=0.8] (3,-2.75)
to (1,-2.75)
to (1,3.25)
to (3, 3.25);
\draw (1,3.25) to (1, -2.75);
\draw [fill=\fillC, draw=none, fill opacity=0.8] (4,-3)
to (2,-3)
to (2,3)
to (4, 3);
\draw (2,3) to (2, -3);
\draw [fill=\fillC, draw=none, fill opacity=0.8] (-4,-3)
to (-2,-3)
to (-2,3)
to (-4, 3);
\draw (-2,3) to (-2, -3);
\node (p)[minimum width=65pt, minimum height=19pt,draw, fill=white, fill opacity=1, scale=0.8] at (0, 0) {$\psi$};
\node (c) [Vertex, scale=0.5, vertex colour=white] at ([xshift=-28.5pt, yshift=1pt]p.south) {};
\node (c) [Vertex, scale=0.5, vertex colour=white] at ([xshift=28.5pt, yshift=1pt]p.south) {};
\node (c) [Vertex, scale=0.5, vertex colour=white] at ([xshift=57pt, yshift=1pt]p.south){};
\node (c) [Vertex, scale=0.5, vertex colour=white] at ([xshift=0pt, yshift=1pt]p.south) {};
\node (c) [Vertex, scale=0.5, vertex colour=white] at ([xshift=-57pt, yshift=1pt]p.south) {};
\end{tikzpicture}
\end{aligned}
\quad=\quad
\Pd\,
\begin{aligned}
\begin{tikzpicture}[scale=0.35,thick,yscale=\ys]
\draw [fill=\fillB, draw=none, fill opacity=0.8] (-2,-2.5)
to (-0,-2.5)
to (-0,3.5)
to (-2, 3.5);
\draw (-0,3.5) to (-0, -2.5);
\draw [fill=\fillcomp, draw=none, fill opacity=0.8] (-3,-2.75)
to (-1,-2.75)
to (-1,3.25)
to (-3, 3.25);
\draw (-1,3.25) to (-1, -2.75);
\draw [fill=\fillcomp, draw=none, fill opacity=0.8] (3,-2.75)
to (1,-2.75)
to (1,3.25)
to (3, 3.25);
\draw (1,3.25) to (1, -2.75);
\draw [fill=\fillC, draw=none, fill opacity=0.8] (4,-3)
to (2,-3)
to (2,3)
to (4, 3);
\draw (2,3) to (2, -3);
\draw [fill=\fillC, draw=none, fill opacity=0.8] (-4,-3)
to (-2,-3)
to (-2,3)
to (-4, 3);
\draw (-2,3) to (-2, -3);
\node (p) [minimum width=65pt, minimum height=19pt,draw, fill=white, fill opacity=1, scale=0.8] at (0, 1.2) {$\phi$};
\node (q)[minimum width=65pt, minimum height=19pt, draw, fill=white, fill opacity=1, scale=0.8] at (0, -1.2) {$\phi^{\dagger}$};
\node (c) [Vertex, scale=0.5, vertex colour=white] at ([xshift=28.5pt, yshift=1pt]p.south) {};
\node (c) [Vertex, scale=0.5, vertex colour=white] at ([xshift=57pt, yshift=1pt]p.south) {};
\node (c) [Vertex, scale=0.5, vertex colour=white] at ([xshift=0pt, yshift=1pt]p.south) {};
\node (c) [Vertex, scale=0.5, vertex colour=white] at ([xshift=-57pt, yshift=1pt]p.south) {};
\node (c) [Vertex, scale=0.5, vertex colour=white] at ([xshift=28.5pt, yshift=1pt]q.south) {};
\node (c) [Vertex, scale=0.5, vertex colour=white] at ([xshift=57pt, yshift=1pt]q.south) {};
\node (c) [Vertex, scale=0.5, vertex colour=white] at ([xshift=-28.5pt, yshift=1pt]q.south) {};
\node (c) [Vertex, scale=0.5, vertex colour=white] at ([xshift=-57pt, yshift=1pt]q.south) {};
\end{tikzpicture}
\end{aligned}
\end{align}
\end{lemma}
\begin{proof}
See Appendix.
\end{proof} 

\subsection{A complementary family from quantum key distribution}

\begin{lemma}
\label{lemma:QKDid}
If a controlled family of measurements allows quantum key distribution with a phase $\psi$, then:
\begin{equation}
\label{eq: QKDid}
\alpha^{\dagger}\circ\alpha\quad=\quad
\Pd\,
\begin{aligned}
\begin{tikzpicture}[scale=0.35, thick,yscale=\ys]
\node (b) [Vertex, scale=\vertexsize,  vertex colour=white] at (1.5, 2.5) {};
\node (c) [Vertex, scale=\vertexsize] at (1.5, 3.50) {};
\node (d) [Vertex, scale=\vertexsize] at (1.5, 6.5) {};
\node (e) [Vertex, scale=\vertexsize, vertex colour=white] at (1.5, 7.5) {};
\draw [fill=\fillcomp, fill opacity=0.8, draw=none] 
    (2,1)
    to(2,2)
    to [out=up, in=up, looseness=2](1,2)
    to [out=down, in=down, looseness=2](0,2)
    to (0,8)
    to [out=up, in=up, looseness=2](1,8)
    to [out=down, in=down, looseness=2](2,8)
    to (2,9)
    to (-2,9)
    to (-2,1);   
\draw (2,1)
    to(2,2)
    to [out=up, in=up, looseness=2](1,2)
    to [out=down, in=down, looseness=2](0,2)
    to (0,8)
    to [out=up, in=up, looseness=2](1,8)
    to [out=down, in=down, looseness=2](2,8)
    to (2,9);   
\draw (b.center) to [out=up, in=down] (c.center); 
\draw (d.center) to [out=up, in=down] (e.center);   
\draw [fill=\fillcomp, fill opacity=0.8, draw=none] (5,9)
    to (3, 9)
    to (3,6)
    to [out=down, in=down, looseness=2](2,6)
    to [out=up, in=up, looseness=2](1,6)
    to (1,5) 
    to (1,4)
    to [out=down, in=down, looseness=2](2,4) 
    to [out=up, in=up, looseness=2](3,4)
    to (3,1)
    to (5,1)  ;
\draw(3, 9)
    to (3,6)
    to [out=down, in=down, looseness=2](2,6)
    to [out=up, in=up, looseness=2](1,6)
    to (1,5) 
    to (1,4)
    to [out=down, in=down, looseness=2](2,4) 
    to [out=up, in=up, looseness=2](3,4)
    to (3,1) ;
\draw [fill=\fillC, draw=none, fill opacity=0.8] (-2.5,8.75) to (-1,8.75) 
        to [out=down, in=left](0.25,7.5) 
        to [out=left, in=up] (-1, 6)
        to (-1.0,4.0)       
        to [out=down, in=left] (0.25, 2.5)
        to [out=left, in=up](-1.0, 1)
        to (-1.0, 0.75)
        to (-2.5,0.75);   
\draw (-1,8.75) 
        to [out=down, in=left](0.25,7.5) 
        to [out=left, in=up] (-1, 6)
        to (-1.0,4.0)       
        to [out=down, in=left] (0.25, 2.5)
        to [out=left, in=up](-1.0, 1)
        to (-1.0, 1) to (-1, 0.75);    
\draw (0.25, 2.5) to (1.5, 2.5);
\draw (0.25, 7.5) to (1.5, 7.5);                
\draw [fill=\fillC, draw=none, fill opacity=0.8] (5.5, 8.75) to (4,8.75)
    to (4,8) 
    to [out=down, in=right](2.75,6.5)
    to [out=right, in=up](4,5)
    to [out=down, in=right] (2.75, 3.5)
    to [out=right, in=up] (4.0, 2)
    to (4.0, 0.75)
    to (5.5, 0.75);
\draw(4,8.75)
    to (4,8) 
    to [out=down, in=right](2.75,6.5)
    to [out=right, in=up](4,5)
    to [out=down, in=right] (2.75, 3.5)
    to [out=right, in=up] (4.0, 2)
    to (4.0, 0.75);          
\draw (2.75,3.5) to (1.5,3.5);
\draw (2.75,6.5) to (1.5,6.5);
\end{tikzpicture}
\end{aligned}
\quad=\quad
\Pd\,
\begin{aligned}
\begin{tikzpicture}[scale=0.35,thick,yscale=\ys]
\draw [fill=\fillcomp, draw=none, fill opacity=0.8] (3.25,-3.75)
to (1.25,-3.75)
to (1.25,4.25)
to (3.25, 4.25);
\draw (1.25,4.25) to (1.25, -3.75);
\draw [fill=\fillcomp, draw=none, fill opacity=0.8] (-3.25,-3.75)
to (-1.25,-3.75)
to (-1.25,1.25)
to [out=up, in=up, looseness=2](-0.375,1.25)
to (-0.375,-1)
to [out=down, in=down, looseness=2](0.5,-1)
to (0.5,4.25)
to (-2,4.25)
to (-3.25, 4.25);
\draw (-1.25,-3.75) 
to (-1.25,1.25)
to [out=up, in=up, looseness=2](-0.375,1.25)
to (-0.375,-1)
to [out=down, in=down, looseness=2](0.5,-1)
to (0.5,4.25);
\draw [fill=\fillC, draw=none, fill opacity=0.8] (4,-4)
to (2,-4)
to (2,4)
to (4, 4);
\draw (2,4) to (2, -4);
\draw [fill=\fillC, draw=none, fill opacity=0.8] (-4,-4)
to (-2,-4)
to (-2,4)
to (-4, 4);
\draw (-2,4) to (-2, -4);
\node (p) [minimum width=65pt, draw, fill=white, fill opacity=1, scale=0.8] at (0, 0) {$\psi$};
\node (c) [Vertex, scale=0.5, vertex colour=white] at ([xshift=35.5pt, yshift=1pt]p.south) {};
\node (c) [Vertex, scale=0.5, vertex colour=white] at ([xshift=57pt, yshift=1pt]p.south) {};
\node (c) [Vertex, scale=0.5, vertex colour=white] at ([xshift=-57pt, yshift=1pt]p.south) {};
\node (c) [Vertex, scale=0.5, vertex colour=white] at ([xshift=-35.5pt, yshift=1pt]p.south) {};
\node (c) [Vertex, scale=0.5, vertex colour=white] at ([xshift=14pt, yshift=1pt]p.south) {};
\end{tikzpicture}
\end{aligned}
\quad=\quad
\Pd\,
\begin{aligned}
\begin{tikzpicture}[scale=0.35,thick,yscale=\ys]
\draw [fill=\fillcomp, draw=none, fill opacity=0.8] (-2,-3.75)
to (0,-3.75)
to (0,4.25)
to (-2, 4.25);
\draw (0,4.25) to (0, -3.75);
\draw [fill=\fillcomp, draw=none, fill opacity=0.8] (3,-3.75)
to (1,-3.75)
to (1,4.25)
to (3, 4.25);
\draw (1,4.25) to (1, -3.75);
\draw [fill=\fillC, draw=none, fill opacity=0.8] (-3,-4)
to (-1,-4)
to (-1,4)
to (-3, 4);
\draw (-1,4) to (-1, -4);
\draw [fill=\fillC, draw=none, fill opacity=0.8] (4,-4)
to (2,-4)
to (2,4)
to (4, 4);
\draw (2,4) to (2, -4);
\end{tikzpicture}
\end{aligned}
\end{equation}
\end{lemma}

\begin{proof}
See Appendix.
\end{proof}

\begin{lemma}
\label{lemma:QKD implies complementarity}
If a controlled family of measurements allows quantum key distribution, then the family is complementary.
\end{lemma}

\begin{proof}
By Lemma~\ref{lemma:QKDid} the following map is unitary:
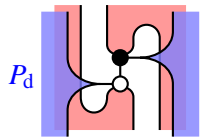
\begin{align*}
\Pd\,
\begin{aligned}
\begin{tikzpicture}[scale=0.35, thick]
\node (b) [Vertex, scale=\vertexsize, vertex colour=white] at (1.5, 2.25) {};
\node (c) [Vertex, scale=\vertexsize] at (1.5, 3.25) {};
\draw [fill=\fillcomp, fill opacity=0.8, draw=none] (-1,5.25)
    to (-1,0.5)
    to  (2,0.5)
    to(2,1.75)
    to [out=up, in=\seangle] (b.center)
    to [out=\swangle, in=up] +(-0.5,-0.5)
    to [out=down, in=down, looseness=2] +(-1,0)
    to (0.0,5.25);
    \draw  (2,0.5)
    to(2,1.75)
    to [out=up, in=\seangle] (b.center)
    to [out=\swangle, in=up] +(-0.5,-0.5)
    to [out=down, in=down, looseness=2] +(-1,0)
    to (0.0,5.25);
\draw (b.center)
    to [out=up, in=down] (c.center);   
\draw [fill=\fillcomp, fill opacity=0.8, draw=none] (1,5.25) to (1,3.75)
    to [out=down, in=\nwangle] (c.center)
    to [out=\neangle, in=down] +(0.5,0.5)
    to (2,4)
    to [out=up, in=left](2.5,4.5)
    to [out=right, in=up](3,4)
    to (3,0.5)
    to(4,0.5)
    to (4,5.25);
 \draw  (1,5.25) to (1,3.75)
    to [out=down, in=\nwangle] (c.center)
    to [out=\neangle, in=down] +(0.5,0.5)
    to (2,4)
    to [out=up, in=left](2.5,4.5)
    to [out=right, in=up](3,4)
    to (3,0.5)   ;
\draw [fill=\fillC, draw=none, fill opacity=0.8] (-1.5,5)
        to (-0.5,5)
        to (-0.5, 3.25)
        to [out=down, in=left] (0.75,2.25)
        to [out=left, in=up](-0.5, 1.25)
        to (-0.5, 0.25)
        to (-1.5,0.25);  
\draw  (-0.5,5)
        to (-0.5, 3.25)
        to [out=down, in=left] (0.75,2.25)
        to [out=left, in=up](-0.5, 1.25)
        to (-0.5, 0.25);              
\draw  (0.75,2.25) to (1.5,2.25);   
\draw [fill=\fillC, draw=none, fill opacity=0.8] (4.5, 5)
    to  (3.5, 5)
    to (3.5, 4.25)
    to [out=down, in=right] (2.25,3.25)
    to [out=right, in=up] (3.5, 2.25)
    to (3.5, 0.25)
    to (4.5, 0.25);
\draw (3.5, 5)
    to (3.5, 4.25)
    to [out=down, in=right] (2.25,3.25)
    to [out=right, in=up] (3.5, 2.25)
    to (3.5, 0.25);
    \draw (2.25,3.25) to (1.5,3.25);
\end{tikzpicture}
\end{aligned}
\end{align*}
By Lemma~\ref{lemma:Complementarity through unitarity}, we can conclude that the controlled family of measurements is complementary.
\end{proof}

\begin{theorem}
\label{thm:equivalent}
A controlled family of measurements satisfies quantum key distribution if and only if it is complementary.
\end{theorem}
\begin{proof}
Immediate by Lemmas~\ref{lemma:Complementarity implies QKD} and \ref{lemma:QKD implies complementarity}.
\end{proof}

\begin{lemma}
If a controlled family of measurements allows quantum key distribution with phase $\psi$, then we can decompose $\psi$ in the following way:
\begin{align}
\Pd
\begin{aligned}
\begin{tikzpicture}[scale=0.35,thick,yscale=\ys]
\draw [fill=\fillB, draw=none, fill opacity=0.8] (-2,-2.5)
to (-0,-2.5)
to (-0,3.5)
to (-2, 3.5);
\draw (-0,3.5) to (-0, -2.5);
\draw [fill=\fillcomp, draw=none, fill opacity=0.8] (-3,-2.75)
to (-1,-2.75)
to (-1,3.25)
to (-3, 3.25);
\draw (-1,3.25) to (-1, -2.75);
\draw [fill=\fillcomp, draw=none, fill opacity=0.8] (3,-2.75)
to (1,-2.75)
to (1,3.25)
to (3, 3.25);
\draw (1,3.25) to (1, -2.75);
\draw [fill=\fillC, draw=none, fill opacity=0.8] (4,-3)
to (2,-3)
to (2,3)
to (4, 3);
\draw (2,3) to (2, -3);
\draw [fill=\fillC, draw=none, fill opacity=0.8] (-4,-3)
to (-2,-3)
to (-2,3)
to (-4, 3);
\draw (-2,3) to (-2, -3);
\node (p)[minimum width=65pt, minimum height=19pt,draw, fill=white, fill opacity=1, scale=0.8] at (0, 0) {$\psi$};
\node (c) [Vertex, scale=0.5, vertex colour=white] at ([xshift=-28.5pt, yshift=1pt]p.south) {};
\node (c) [Vertex, scale=0.5, vertex colour=white] at ([xshift=28.5pt, yshift=1pt]p.south) {};
\node (c) [Vertex, scale=0.5, vertex colour=white] at ([xshift=57pt, yshift=1pt]p.south){};
\node (c) [Vertex, scale=0.5, vertex colour=white] at ([xshift=0pt, yshift=1pt]p.south) {};
\node (c) [Vertex, scale=0.5, vertex colour=white] at ([xshift=-57pt, yshift=1pt]p.south) {};
\end{tikzpicture}
\end{aligned}
\quad=\quad
\Pd
\begin{aligned}
\begin{tikzpicture}[scale=0.35,thick,yscale=\ys]
\draw [fill=\fillB, draw=none, fill opacity=0.8] (-2,-2.5)
to (-0,-2.5)
to (-0,3.5)
to (-2, 3.5);
\draw (-0,3.5) to (-0, -2.5);
\draw [fill=\fillcomp, draw=none, fill opacity=0.8] (-3,-2.75)
to (-1,-2.75)
to (-1,3.25)
to (-3, 3.25);
\draw (-1,3.25) to (-1, -2.75);
\draw [fill=\fillcomp, draw=none, fill opacity=0.8] (3,-2.75)
to (1,-2.75)
to (1,3.25)
to (3, 3.25);
\draw (1,3.25) to (1, -2.75);
\draw [fill=\fillC, draw=none, fill opacity=0.8] (4,-3)
to (2,-3)
to (2,3)
to (4, 3);
\draw (2,3) to (2, -3);
\draw [fill=\fillC, draw=none, fill opacity=0.8] (-4,-3)
to (-2,-3)
to (-2,3)
to (-4, 3);
\draw (-2,3) to (-2, -3);
\node (p) [minimum width=65pt, minimum height=19pt,draw, fill=white, fill opacity=1, scale=0.8] at (0, 1.2) {$\phi$};
\node (q)[minimum width=65pt, minimum height=19pt, draw, fill=white, fill opacity=1, scale=0.8] at (0, -1.2) {$\phi^{\dagger}$};
\node (c) [Vertex, scale=0.5, vertex colour=white] at ([xshift=28.5pt, yshift=1pt]p.south) {};
\node (c) [Vertex, scale=0.5, vertex colour=white] at ([xshift=57pt, yshift=1pt]p.south) {};
\node (c) [Vertex, scale=0.5, vertex colour=white] at ([xshift=0pt, yshift=1pt]p.south) {};
\node (c) [Vertex, scale=0.5, vertex colour=white] at ([xshift=-57pt, yshift=1pt]p.south) {};
\node (c) [Vertex, scale=0.5, vertex colour=white] at ([xshift=28.5pt, yshift=1pt]q.south) {};
\node (c) [Vertex, scale=0.5, vertex colour=white] at ([xshift=57pt, yshift=1pt]q.south) {};
\node (c) [Vertex, scale=0.5, vertex colour=white] at ([xshift=-28.5pt, yshift=1pt]q.south) {};
\node (c) [Vertex, scale=0.5, vertex colour=white] at ([xshift=-57pt, yshift=1pt]q.south) {};
\end{tikzpicture}
\end{aligned}
\end{align}
\end{lemma}
\begin{proof}
Immediate by Lemmas~\ref{lemma:Complementarity implies QKD} and \ref{lemma:QKD implies complementarity}.
\end{proof}

\section{The Mean King problem}
\label{sec:meanking}

The Mean King problem is defined as follows~\cite{Aharonov1987,Gothic}. There are two agents, Alice and the King, who take part in the following procedure.
\begin{enumerate}
\item Alice hands a quantum state to the King.
\item The King measures the state in one of $n$ mutually unbiased bases, keeping both the basis and the outcome secret, and returns the state to Alice. 
\item Alice performs any quantum measurement she wishes.
\item The King reveals his measurement basis to Alice.
\item Using only classical processes, Alice must calculate the King's earlier measurement outcome.
\end{enumerate}
With some thought, it becomes clear that for Alice to have the best chance of succeeding, she should retain an entangled partner to the system initially passed to the King. Alice should then apply a predetermined measurement procedure to the entangled state that will reveal the King's measurement result every time, regardless of his basis choice. In other words, she should perform some nondegenerate PVM $\mu$ on both systems in step 3, and prepare a lookup table $f$ that tells her, depending on the King's measurement basis choice and her own measurement result, what King's result was.

The key results of this Section is a graphical definition of a solution of the Mean King problem, and a graphical proof of the correctness of Klappenecker and R\"otteler's solution~\cite{Gothic} to the Mean King problem.

\subsection{Abstract definition}

We begin with an abstract definition of the Mean King problem. Recall that in categorical quantum mechanics, a \emph{classical function} is defined as a morphism between classical data which satisfies the comonoid homomorphism property, and that regions with topological boundary carry a canonical comonoid structure.
\def\mdot{\ensuremath{\begin{tikzpicture}\node [Vertex, scale=1.0] at (0,0) {};\end{tikzpicture}}}
\begin{definition}[Mean King scheme]
\label{Def:Mean king problem scheme}
Given a complementary family of measurements $\mdot$, a bipartite measurement $\mu$, and a classical function $f$, a \textit{Mean King scheme} $\mathrm{MK}_{\mdot,\mu,f}$ is defined as the following composite:

\vspace{-30pt}
\begin{align}
\begin{aligned}
\begin{tikzpicture} [scale=0.35,thick,yscale=0.9]
\draw [white] (-0.5, 6)
to [out=down, in=up, out looseness=1.5, in looseness=0.8] (3.5,2.5)
to [out=down, in=right](2, 1.5)
to [out=right, in=up](3.5,0.5)
to (3.5, 0)
to [out=down, in=right](2,-1)
to [out=right, in=up](3.5,-2)
to [out=down, in=down, looseness=2](4.5,-2)
to (4.5,8)
to (3.5, 8)
to (3.5, 3)
to [out=down, in=down, in looseness=1.3](0.5, 6);
\fill [fill=\fillcomp, fill opacity=0.8, draw] (0.75, 8) 
to (0.75, 3)
to (0, 3)
to (0, 8);
\fill [fill=\fillcomp, fill opacity=0.8, draw] (1.625, 8) 
to (1.625, 3)
to (2.375, 3)
to (2.375, 8);
\fill [fill=\fillC, fill opacity=0.8, draw] (3.25, 8) 
to (3.25, 3)
to (4, 3)
to (4, 8);
\node [draw, fill=white, minimum width=90pt, minimum height=25pt,scale=0.6, fill opacity=1]at (2, 3) {$\mathrm{MK}_{\mdot,\mu,f}$};
\end{tikzpicture}
\end{aligned}
\quad:=\quad
\begin{aligned}
\begin{tikzpicture} [scale=0.35,thick,yscale=0.9]
\draw [fill=black!10, draw=none] (1.4,-3.7) rectangle +(5.65,12);
\node (B) [Vertex, scale=\vertexsize] at (3,-0.5) {};
\node (D) [Vertex, scale=\vertexsize] at (2, 1.5) {};
\draw (-2,3) to [out=\swangle, in=\swangle, in looseness=2, out looseness=1] (3,-0.5);
\draw [out=\nwangle, in=\seangle] (2, 1.5) to (-2, 3);
\fill [fill=\fillcomp, fill opacity=0.8, draw](-1.35,6)
to (-1.35, 4)
to [out=down, in=down, looseness=2] (-2.65, 4)
to (-2.65,6);
\fill [fill=\fillcomp, fill opacity=0.8, draw] (3.5, 8)
to (3.5,5) 
to (3.5,1)
to [out=down, in=down, looseness=2](2.5,1)
to (2.5, 1)
to [out=up, in=up, looseness=2] (1.5,1)
to [out=down, in=down, looseness=1.7](4.5,1)
to (4.5,5)
to (4.5, 8);
\fill [fill=\fillcomp, fill opacity=0.8, draw] (-1.5, 8) 
to (-1.5, 6)
to (-0.5, 6)
to (-0.5, 8);
\fill [fill=\fillC, fill opacity=0.8, draw=none] (-0.5, 6)
to [out=down, in=up, out looseness=1.5, in looseness=0.8] (5.5,2.5)
to [out=down, in=right](4, 1.5)
to [out=right, in=up](5.5,0.5)
to [out=down, in=right](4,-0.5)
to [out=right, in=up](5.5,-1.5)
to [out=down, in=down, looseness=2](6.5,-1.5)
to (6.5,8)
to (5.5, 8)
to (5.5, 5)
to [out=down, in=down, in looseness=1.2, out looseness=0.4](0.5, 6);
\draw(-0.5, 6)
to [out=down, in=up, out looseness=1.5, in looseness=0.8] (5.5,2.5)
to [out=down, in=right](4, 1.5)
to [out=right, in=up](5.5,0.5)
to [out=down, in=right](4,-0.5)
to [out=right, in=up](5.5,-1.5)
to [out=down, in=down, looseness=2](6.5,-1.5)
to (6.5,8)
(5.5, 8)
to (5.5, 5)
to [out=down, in=down, in looseness=1.2, out looseness=0.4](0.5, 6);
\draw (4, -0.5) to (3, -0.5);
\draw (4, 1.5) to (2, 1.5);
\node (A) [scale=0.6, draw, circle, fill=white, fill opacity=1] at (-2,3) {$\mu$};
\node [draw, fill=white, minimum width=70pt, minimum height=25pt , scale=0.6, fill opacity=1]at (-1, 6.6) {$f$};
\node [] at (-1.5,-3){\sc alice};
\node [] at (4,-3){\sc king};
\end{tikzpicture}
\end{aligned}
\end{align}
\end{definition}

\begin{definition}[Mean King solution]
\label{Def:Mean king problem specification}
A Mean King scheme $\mathrm{MK}_{\mdot,\mu,f}$ \textit{solves the Mean King problem} if the following equation holds:
\begin{equation}
\label{eq:Mean king problem specification}
\begin{aligned}
\begin{tikzpicture} [scale=0.4,thick]
\fill [fill=\fillcomp, fill opacity=0.8, draw] (0.75, 8) 
to (0.75, 5)
to (0, 5)
to (0, 8);
\fill [fill=\fillcomp, fill opacity=0.8, draw] (1.625, 8) 
to (1.625, 5)
to (2.375, 5)
to (2.375, 8);
\fill [fill=\fillC, fill opacity=0.8, draw] (3.25, 8) 
to (3.25, 5)
to (4, 5)
to (4, 8);
\node [draw, fill=white, minimum width=90pt, minimum height=25pt,scale=0.6, fill opacity=1]at (2, 5) {$\text{MK}_{\mdot,\mu,f}$};
\end{tikzpicture}
\end{aligned}
\quad=\quad
\begin{aligned}
\begin{tikzpicture} [scale=0.4,thick]
\fill [fill=\fillcomp, fill opacity=0.8, draw=none] (0, 8)
to (0, 5)
to (0.75, 5)
to (0.75, 6.5) 
to [out=up, in=up, looseness=2](1.625, 6.5) 
to (1.625, 5)
to (2.375, 5)
to (2.375, 8)
to (1.625, 8)
to [out=down, in=down, looseness=2](0.75, 8);
\draw (0, 8)
to (0, 5)
to (0.75, 5)
to (0.75, 6.5) 
to [out=up, in=up, looseness=2](1.625, 6.5) 
to (1.625, 5)
to (2.375, 5)
to (2.375, 8)
 (1.625, 8)
to [out=down, in=down, looseness=2](0.75, 8);
\fill [fill=\fillC, fill opacity=0.8, draw] (3.25, 8) 
to (3.25, 5)
to (4, 5)
to (4, 8);
\node [draw, fill=white, minimum width=90pt, minimum height=25pt,scale=0.6, fill opacity=1]at (2, 5) {$\text{MK}_{\mdot,\mu,f}$};
\end{tikzpicture}
\end{aligned}
\end{equation}
\end{definition}

\noindent
This says exactly that, after carrying out the procedure, Alice and the King carry the same measurement result information. It will be satisfied precisely if the Mean King scheme $\text{MK}_{\mdot,\mu,f}$ is correct.

\subsection{Solving the Mean King problem}

Our solution to the Mean King problem is presented entirely graphically. It is based on a solution due to Klappenecker and R\" otteler~\cite{Gothic}. Giving this solution graphically is an interesting exercise in the graphical formalism, and demonstrates that it is capable of reasoning about sophisticated schemes. Our presentation is roughly comparable in complexity to Klappenecker and R\"otteler's. One advantage of our presentation is that it is perhaps clearer exactly how complementarity is being used.

We begin by giving a scheme to construct a bipartite state from a classical function.
\begin{definition}
\label{def: Alice's basis}
Given a classical function $f_i:n \to m$, and an $n$-fold controlled family of measurements on $\C^n$, the associated bipartite state $\ket{\mu _f} \in \C^n \otimes \C^n$ is defined as follows:
\begin{align*}
\mu_{f}\quad:=\quad\frac{1}{\sqrt{n}}
\begin{aligned}
\begin{tikzpicture} [scale=0.5,thick, scale=0.8]
\draw [fill=\fillC, fill opacity=0.8] (0,0) 
to [out=right, in=down](2.9,3.2)
to [out=up, in=up, looseness=2](2,3.2)
to [out=down, in=right](1.3, 2.9)
to [out=right, in=up](1.9, 2.6)
to [out=down, in=right](0.5,0.7)
to [out=left, in=down](0.25, 1.5)
to [out=down, in=up](-0.25, 1.5)
to [out=down, in=right](-0.5, 0.7)
to [out=left, in=down](-1.9, 2.6)
to [out=up, in=left](-1.3, 2.9)
to [out=left, in=down](-2,3.2)
to  [out=up, in=up, looseness=2](-2.9,3.2)
to [out=down, in=left](0,0);
\draw [fill=\fillcomp, fill opacity=0.8](0.25, 1.5)
to [out=up,in=down](1.5,2.5)
to [out=up, in=up, looseness=2](1, 2.5)
to [out=down, in=down, looseness=0.6](-1,2.5)
to [out=up, in=up, looseness=2](-1.5,2.5)
to [out=down, in=up](-0.25, 1.5);
\draw [thick](1.3,3) to (1.3,4);  
\draw [thick](-1.3,3) to (-1.3,4);  
\node [draw, scale=0.5, minimum width=30, minimum height=14,fill=white, fill opacity=1] at (0, 1.3) {$f_i$};
\node [Vertex, scale=\vertexsize, vertex colour=white] at (1.3, 2.9) {};
\node [Vertex, scale=\vertexsize, vertex colour=white]at (-1.3, 2.9) {};
\end{tikzpicture}
\end{aligned}
\quad-\quad
\begin{aligned}
\begin{tikzpicture}[thick]
\draw [white]
(0,0.5) to (0,2);
\draw (0,2) 
to [out=down, in=down, looseness=2] (2,2);
\end{tikzpicture}
\end{aligned}
\end{align*}
\end{definition}
\begin{definition}[Collisions]
Given classical functions $f,g:n \to m$, let $f \diamond g := |\{a|f(a)=g(a)\}|$ be the number of \textit{collisions} between them.
\end{definition}

\begin{lemma}
\label{lem:collision}
Let $f, g: [n+1] \to [n]$ be functions. Then $n \langle\mu_f|\mu_g\rangle + 1 = f \diamond g$. 
\end{lemma}

\begin{proof}
A straightforward graphical proof is possible, which we omit.
\end{proof}
\begin{lemma}
\label{lemma:orthonormal basis}
Given a family of $n^2$ classical functions $f_i: [n+1] \to [n]$ with $i \neq j \,\Rightarrow\, f_i \diamond f_j = 1$, the states $\ket{\mu _{f_i}}$ form an orthonormal basis.
\end{lemma}
\begin{proof}
Rearranging the result of Lemma~\ref{lem:collision}, we see that $\langle \mu _f | \mu _g\rangle = ((f \diamond g) - 1)/n$, and the conclusion follows.
\end{proof}

\begin{lemma}
\label{lemma: prime power}
For any prime power $n=p^k$, the following structures exist:
\begin{enumerate}
\item a family of $n^2$ functions $f_i : [n+1] \to [n]$, such that for $i \neq j$ we have $f_i \diamond f_j=1$;
\item a  family of $n+1$ mutually complementary bases.
\end{enumerate}
\end{lemma}
\begin{proof}
See~\cite[Section 2]{Gothic}.
\end{proof}

\begin{lemma}
\label{lemma: auxilliary MKP complementarity result}
For a complementary family of controlled measurements and a classical function $g$, the following holds:
\begin{align}
\begin{aligned}
\begin{tikzpicture} [scale=0.4,thick,yscale=0.8]
\draw [use as bounding box, draw=none] (-2,1) rectangle (2,6.5);
\draw [fill=\fillC, fill opacity=0.8] (0,0.5) 
to [out=right, in=down, looseness=1.3](2,3.2)
to [out=up, in=up, looseness=2](1.5,3.2)
to [out=down, in=right](0.8, 2.9)
to [out=right, in=up](1.4, 2.6)
to [out=down, in=right](0.5,1)
to [out=left, in=down](0.25, 1.5)
to (-0.25, 1.5)
to [out=down, in=right](-0.5, 1)
to [out=left, in=down](-1.4, 2.6)
to [out=up, in=left](-0.8, 2.9)
to [out=left, in=down](-1.5,3.2)
to  [out=up, in=up, looseness=2](-2,3.2)
to [out=down, in=left, looseness=1.3](0,0.5);
\draw [fill=\fillcomp, fill opacity=0.8] (0.25, 1.5)
to [out=up,in=down](1,2.5)
to [out=up, in=up, looseness=2](0.5, 2.5)
to [out=down, in=down, looseness=1.1](-0.5,2.5)
to [out=up, in=up, looseness=2](-1,2.5)
to [out=down, in=up](-0.25, 1.5);
\draw [thick](0.8,3) to (0.8,5);  
\draw [thick](-0.8,3) to (-0.8,5);  
\node [draw, scale=0.4, minimum width=30, minimum height=17,fill=white, fill opacity=1] at (0, 1.4) {$g$};
\node [Vertex, scale=\vertexsize, vertex colour=white] at (0.8, 2.9) {};
\node [Vertex, scale=\vertexsize, vertex colour=white]at (-0.8, 2.9) {};
\draw [fill=\fillC, yscale=-1, fill opacity=0.8] (0,-7.5) 
to [out=right, in=down, looseness=1.1](2,-4.8)
to [out=up, in=up, looseness=2](1.5,-4.8)
to [out=down, in=right](0.8, -5.1)
to [out=right, in=up](1.4, -5.4)
to [out=down, in=right](0, -7.)
to [out=left, in=down](-1.4, -5.4)
to [out=up, in=left](-0.8, -5.1)
to [out=left, in=down](-1.5,-4.8)
to  [out=up, in=up, looseness=2](-2,-4.8)
to [out=down, in=left, looseness=1.1](0,-7.5);
\draw [fill=\fillcomp, yscale=-1, fill opacity=0.8](0, -6.7)
to [out=right,in=down](1.1,-5.5)
to [out=up, in=up, looseness=2](0.5, -5.5)
to [out=down, in=down, looseness=2](-0.5,-5.5)
to [out=up, in=up, looseness=2](-1.1,-5.5)
to [out=down, in=left](0, -6.7);
\node [Vertex, scale=\vertexsize, vertex colour=white] at (0.8, 5.1) {};
\node [Vertex, scale=\vertexsize, vertex colour=white]at (-0.8, 5.1) {};
\node [scale=\labelsize] at (0, 7.2) {$a$};
\node [scale=\labelsize] at (0, 6.35) {$b$};
\end{tikzpicture}
\end{aligned}
\quad= \quad
\begin{aligned}
\begin{tikzpicture} [thick, scale=0.5]
\draw [fill=\fillC, fill opacity=0.8] (0,1)
to (0,0)
to [out=down,in=down, looseness=2](1,0)
to (1,1)
to (0,1);
\draw [fill=\fillcomp, fill opacity=0.8] (0,2)
to (0,1)
to (1,1)
to (1,2)
to [out=up, in=up, looseness=2](0,2);
\node [draw, minimum width=40, 
minimum height=27,fill=white, scale=0.5, fill opacity=1]
 at (0.5, 1) {$g$}; 
   \node [scale=\labelsize]
 at (0.5, 0) {$a$};
   \node [scale=\labelsize]
 at (0.5, 2) {$b$};
\end{tikzpicture}
\end{aligned}
\quad+\quad
1
\end{align}
\end{lemma}
\begin{proof}
See Appendix.
\end{proof}

\begin{theorem} [Solution to the Mean King problem]
\label{lemma: MKP correctness}
For a family of functions $f_i: [n+1] \to [n]$ such that $\ket{\mu _{f_i}}$ form a basis, and a family $\mdot$ of $n+1$ complementary bases of $\mathbb{C} ^n$, the following assignments give a Mean King solution:
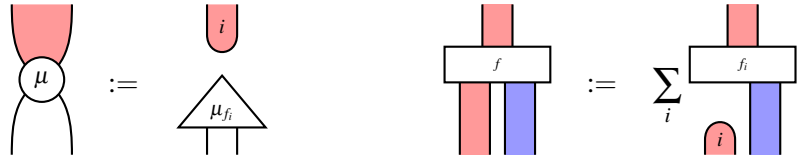
\begin{align*}
\begin{aligned}
\begin{tikzpicture} [thick, scale=0.4]
\draw [white](0, -3) to (0, 2);
\draw [fill=\fillcomp,  fill opacity=0.8] (-1, 2) 
to [out=down, in=\nwangle](0, -0.5) 
to [out=\neangle, in=down](1, 2);
\draw (-1, -3) 
to [out=up, in=\swangle](0, -0.5) 
to [out=\seangle, in=up](1, -3);
\node [draw, circle, fill=white, fill opacity=1, scale=0.8] at (0, -0.5) {$\mu$}; 
\end{tikzpicture}
\end{aligned}
\quad:=\quad
\begin{aligned}
\begin{tikzpicture} [thick, scale=0.4]
\draw [fill=\fillcomp, fill opacity=0.8] (-0.5, 2)
to (-0.5, 1)
to [out=down, in=down, looseness=2](0.5, 1)
to (0.5, 2);
\draw (0.5, -3) to (0.5, -2);
\draw (-0.5, -3) to (-0.5, -2);
\node [scale=\labelsize] at (0, 1.25) {$i$};
\node [draw, triangleup, fill=white, fill opacity=1, scale=0.7] at (0, -1.6) {$\mu_{f_i}$}; 
\end{tikzpicture}
\end{aligned}
\qquad\qquad\qquad
\begin{aligned}
\begin{tikzpicture} [thick, scale=0.4]
\draw [fill=\fillC, fill opacity=0.8] (2.75, -3)
to (2.75, 0)
to (1.75, 0)
to (1.75, -3);
\draw [fill=\fillcomp, fill opacity=0.8] (0.25, -3)
to (0.25, 0)
to (1.25, 0)
to (1.25, -3);
\draw [fill=\fillcomp, fill opacity=0.8] (1,2)
to (1,0)
to (2,0)
to (2, 2);
\node [draw, minimum width=80, minimum height=27,fill=white, scale=0.5, fill opacity=1] at (1.5, 0) {$f$}; 
\end{tikzpicture}
\end{aligned}
\quad:=\quad
\sum_i
\begin{aligned}
\begin{tikzpicture} [thick, scale=0.4]
\draw [fill=\fillC, fill opacity=0.8] (2.75, -3)
to (2.75, 0)
to (1.75, 0)
to (1.75, -3);
\draw [fill=\fillcomp, fill opacity=0.8] (0.25, -3)
to (0.25, -2.5)
to [out=up, in=up, looseness=2](1.25, -2.5)
to (1.25, -3);
\draw [fill=\fillcomp, fill opacity=0.8] (1,2)
to (1,0)
to (2,0)
to (2, 2);
\node [scale=\labelsize] at (0.75, -2.5) {$i$};
\node [draw, minimum width=80, minimum height=27,fill=white, scale=0.5, fill opacity=1] at (1.5, 0) {$f_i$}; 
\end{tikzpicture}
\end{aligned}
\end{align*}
\end{theorem}
\begin{proof}
See Appendix. 
\end{proof}

\bibliographystyle{eptcs}
\bibliography{references}

\begin{thebibliography}{10}
\providecommand{\bibitemdeclare}[2]{}
\providecommand{\surnamestart}{}
\providecommand{\surnameend}{}
\providecommand{\urlprefix}{Available at }
\providecommand{\url}[1]{\texttt{#1}}
\providecommand{\href}[2]{\texttt{#2}}
\providecommand{\urlalt}[2]{\href{#1}{#2}}
\providecommand{\doi}[1]{doi:\urlalt{http://dx.doi.org/#1}{#1}}
\providecommand{\bibinfo}[2]{#2}

\bibitemdeclare{article}{Abramsky:2004ac}
\bibitem{Abramsky:2004ac}
\bibinfo{author}{Samson \surnamestart Abramsky\surnameend} \&
  \bibinfo{author}{Bob \surnamestart Coecke\surnameend} (\bibinfo{year}{2004}):
  \emph{\bibinfo{title}{A Categorical Semantics of Quantum Protocols}}.
\newblock {\sl \bibinfo{journal}{Proceedings of the 19th Annual IEEE Symposium
  on Logic in Computer Science}}, pp. \bibinfo{pages}{415--425},
  \doi{10.1109/LICS.2004.1}.
\newblock
  \bibinfo{note}{\href{http://arxiv.org/abs/quant-ph/0402130}{arXiv:quant-ph/0%
402130}}.

\bibitemdeclare{inbook}{ac08-cqm}
\bibitem{ac08-cqm}
\bibinfo{author}{Samson \surnamestart Abramsky\surnameend} \&
  \bibinfo{author}{Bob \surnamestart Coecke\surnameend} (\bibinfo{year}{2008}):
  \emph{\bibinfo{title}{Handbook of Quantum Logic and Quantum Structures}},
  chapter \bibinfo{chapter}{Categorical Quantum Mechanics}, pp.
  \bibinfo{pages}{261--323}.
\newblock \bibinfo{volume}{2}, \bibinfo{publisher}{Elsevier}.

\bibitemdeclare{article}{b97-hda2}
\bibitem{b97-hda2}
\bibinfo{author}{John~C \surnamestart Baez\surnameend} (\bibinfo{year}{1997}):
  \emph{\bibinfo{title}{Higher-Dimensional Algebra {II}: 2-{H}ilbert Spaces}}.
\newblock {\sl \bibinfo{journal}{Advances in Mathematics}}
  \bibinfo{volume}{127}, pp. \bibinfo{pages}{125--189},
  \doi{10.1006/aima.1997.1617}.
\newblock
  \bibinfo{note}{\href{http://arxiv.org/abs/q-alg/9609018}{arXiv:q-alg/9609018%
}}.

\bibitemdeclare{article}{BarVicary}
\bibitem{BarVicary}
\bibinfo{author}{Krzysztof \surnamestart Bar\surnameend} \&
  \bibinfo{author}{Jamie \surnamestart Vicary\surnameend}
  (\bibinfo{year}{2014}): \emph{\bibinfo{title}{Groupoid Semantics for Thermal
  Computing}}.
\newblock
  \bibinfo{note}{\href{http://arxiv.org/abs/1401.3280}{arXiv:1401.3280}}.

\bibitemdeclare{article}{BB84}
\bibitem{BB84}
\bibinfo{author}{C.H. \surnamestart Bennett\surnameend} \&
  \bibinfo{author}{G.~\surnamestart Brassard\surnameend}
  (\bibinfo{year}{1985}): \emph{\bibinfo{title}{Quantum public key
  distribution}}.
\newblock {\sl \bibinfo{journal}{IBM Technical Disclosure Bulletin}}
  \bibinfo{volume}{28}, pp. \bibinfo{pages}{3153--3163}.

\bibitemdeclare{article}{Duncan:2009}
\bibitem{Duncan:2009}
\bibinfo{author}{Bob \surnamestart Coecke\surnameend} \& \bibinfo{author}{Ross
  \surnamestart Duncan\surnameend} (\bibinfo{year}{2011}):
  \emph{\bibinfo{title}{Interacting Quantum Observables: Categorical Algebra
  and Diagrammatics}}.
\newblock {\sl \bibinfo{journal}{New Journal of Physics}}
  \bibinfo{volume}{13}(\bibinfo{number}{043016}),
  \doi{10.1088/1367-2630/13/4/043016}.
\newblock
  \bibinfo{note}{\href{http://arxiv.org/abs/0906.4725}{arXiv:0906.4725}}.

\bibitemdeclare{article}{coeckeperdrix}
\bibitem{coeckeperdrix}
\bibinfo{author}{Bob \surnamestart Coecke\surnameend} \& \bibinfo{author}{Simon
  \surnamestart Perdrix\surnameend} (\bibinfo{year}{2012}):
  \emph{\bibinfo{title}{Environment and Classical Channels in Categorical
  Quantum Mechanics}}.
\newblock {\sl \bibinfo{journal}{Logical Methods in Computer Science}}
  \bibinfo{volume}{8}(\bibinfo{number}{4}), pp. \bibinfo{pages}{1--24},
  \doi{10.2168/LMCS-8(4:14)2012}.
\newblock
  \bibinfo{note}{\href{http://arxiv.org/abs/1004.1598}{arXiv:1004.1598}}.

\bibitemdeclare{article}{E91}
\bibitem{E91}
\bibinfo{author}{Artur \surnamestart Ekert\surnameend} (\bibinfo{year}{1991}):
  \emph{\bibinfo{title}{Quantum cryptography based on {B}ell's theorem}}.
\newblock {\sl \bibinfo{journal}{Physical Review Letters}}
  \bibinfo{volume}{67}, p. \bibinfo{pages}{661},
  \doi{dx.doi.org/10.1103/PhysRevLett.67.661}.

\bibitemdeclare{article}{Gothic}
\bibitem{Gothic}
\bibinfo{author}{Andreas \surnamestart Klappenecker\surnameend} \&
  \bibinfo{author}{Martin \surnamestart R\"otteler\surnameend}
  (\bibinfo{year}{2005}): \emph{\bibinfo{title}{New Tales of the Mean King}}.
\newblock
  \bibinfo{note}{\href{http://arxiv.org/abs/quant-ph/0502138}{arXiv:quant-ph/0%
502138}}.

\bibitemdeclare{mastersthesis}{r11-2vect}
\bibitem{r11-2vect}
\bibinfo{author}{Dan \surnamestart Roberts\surnameend} (\bibinfo{year}{2011}):
  \emph{\bibinfo{title}{Representing Modular Tensor Categories: A Computer
  Algebra System for Topological Quantum Computing}}.
\newblock Master's thesis, \bibinfo{school}{Department of Computer Science,
  University of Oxford}.

\bibitemdeclare{unpublished}{2vect}
\bibitem{2vect}
\bibinfo{author}{Dan \surnamestart Roberts\surnameend} \&
  \bibinfo{author}{Jamie \surnamestart Vicary\surnameend}:
  \emph{\bibinfo{title}{The {TwoVect} Package}}.
\newblock \bibinfo{note}{A computer algebra system for higher linear algebra.
  \href{http://ncatlab.org/nlab/show/TwoVect}{http://ncatlab.org/nlab/show/Two%
Vect}}.

\bibitemdeclare{inproceedings}{StayVicary}
\bibitem{StayVicary}
\bibinfo{author}{Mike \surnamestart Stay\surnameend} \& \bibinfo{author}{Jamie
  \surnamestart Vicary\surnameend} (\bibinfo{year}{2013}):
  \emph{\bibinfo{title}{Bicategorical Semantics of Nondeterministic
  Computation}}.
\newblock In: {\sl \bibinfo{booktitle}{Proceedings of the 29th Conference on
  the Mathematical Foundations of Programming Semantics}},
  \bibinfo{volume}{298}, pp. \bibinfo{pages}{367--382},
  \doi{10.1016/j.entcs.2013.09.022}.
\newblock
  \bibinfo{note}{\href{http://arxiv.org/abs/1301.3393}{arXiv:1301.3393}}.

\bibitemdeclare{article}{Aharonov1987}
\bibitem{Aharonov1987}
\bibinfo{author}{Lev \surnamestart Vaidman\surnameend}, \bibinfo{author}{Yakir
  \surnamestart Aharonov\surnameend} \& \bibinfo{author}{David~Z. \surnamestart
  Albert\surnameend} (\bibinfo{year}{1987}): \emph{\bibinfo{title}{How to
  Ascertain the Values of $\sigma_x$, $\sigma_y$ and $\sigma_z$ of a
  Spin-$\frac{1}{2}$ Particle}}.
\newblock {\sl \bibinfo{journal}{Physics Review Letters}} \bibinfo{volume}{58},
  pp. \bibinfo{pages}{1385--1387},
  \doi{dx.doi.org/10.1103/PhysRevLett.58.1385}.

\bibitemdeclare{article}{Vicary:2012hqt}
\bibitem{Vicary:2012hqt}
\bibinfo{author}{Jamie \surnamestart Vicary\surnameend} (\bibinfo{year}{2012}):
  \emph{\bibinfo{title}{Higher Semantics of Quantum Protocols}}.
\newblock {\sl \bibinfo{journal}{Proceedings of the 27th Annual IEEE Symposium
  on Logic in Computer Science}}, pp. \bibinfo{pages}{606--615},
  \doi{10.1109/LICS.2012.70}.
\newblock \bibinfo{note}{Extended version at
  \href{http://arxiv.org/abs/1207.4563}{arXiv:1207.4563}}.

\end{thebibliography}

\newpage
\appendix
\section*{Appendix}

{\bf Proof of Lemma~\ref{lemma:Complementarity through unitarity}.}
We consider the following chain of equivalences:
\def\quad{\hspace{0.3cm}}
\begin{align*}
&\left[
\Pd
\begin{aligned}
\begin{tikzpicture}[scale=0.4, thick]
\node (a) [Vertex, scale=\vertexsize, vertex colour=white]at (0.5,-0.25) {};
\node (b) [Vertex, scale=\vertexsize, vertex colour=white]at (1.5, 2.) {};
\node (c) [Vertex, scale=\vertexsize] at (1.5, 3.00) {};
\draw [fill=\fillcomp, fill opacity=0.8, draw=none] (-1,4)
    to (-1,1.75)
    to [out=down, in=\nwangle, out looseness=1.5] (a.center)
    to [out=\neangle, in=down, in looseness=1.5] (2,1.5)
    to [out=up, in=\seangle] (b.center)
    to [out=\swangle, in=up] +(-0.5,-0.5)
    to [out=down, in=down, looseness=2] +(-1,0)
    to (0.0,4);
\draw (-1,4)
    to (-1,1.75)
    to [out=down, in=\nwangle, out looseness=1.5] (a.center)
    to [out=\neangle, in=down, in looseness=1.5] (2,1.5)
    to [out=up, in=\seangle] (b.center)
    to [out=\swangle, in=up] +(-0.5,-0.5)
    to [out=down, in=down, looseness=2] +(-1,0)
    to (0.0,4);
\draw (b.center)
    to [out=up, in=down] (c.center);  
\draw [fill=\fillcomp, fill opacity=0.8, draw=none] (1,4) to (1,3.5)
    to [out=down, in=\nwangle] (c.center)
    to [out=\neangle, in=down] +(0.5,0.5)
    to (2,4);
\draw (1,4) to (1,3.5)
    to [out=down, in=\nwangle] (c.center)
    to [out=\neangle, in=down] +(0.5,0.5)
    to (2,4);
\draw [fill=\fillC, fill opacity=0.8, draw=none] (-2.5,4.0)
        to (-2.0,4.0)
        to (-2, 3)
        to [out=down, in=left] (-1,2)
        to [out=left, in=up] (-2, 1)
        to (-2.0, 0.25)
        to [out=down, in=left] (-1, -0.25)
        to [out=left, in=up] (-2.0, -0.75)
        to (-2.5,-0.75);        
\draw(-2.0,4.0)
        to (-2, 3)
        to [out=down, in=left] (-1,2)
        to [out=left, in=up] (-2, 1)
        to (-2.0, 0.25)
        to [out=down, in=left] (-1, -0.25)
        to [out=left, in=up] (-2.0, -0.75);         
\draw [fill=\fillC, fill opacity=0.8, draw=none] (3.5, 4.0)
    to  (3.0, 4.0)
    to [out=down, in=right] (c.center)
    to [out=right, in=up] (3.0, 2.0)
    to (3.0, -0.75)
    to (3.5, -0.75);
\draw(3.0, 4.0)
    to [out=down, in=right] (c.center)
    to [out=right, in=up] (3.0, 2.0)
    to (3.0, -0.75);
\draw(-1,-0.25) to (0.5, -0.25);
\draw(-1,2) to (1.5, 2);
\draw (0.5,-0.75 -| a.center) to (a.center);
\end{tikzpicture}
\end{aligned}
\quad=\quad
\frac{\Pd}{\sqrt{n}}
\begin{aligned}
\begin{tikzpicture}[scale=0.4, thick]
\node (a) [Vertex, scale=\vertexsize, vertex colour=white] at (0.0, 0.1) {};
\draw [fill=\fillcomp, fill opacity=0.8, draw=none] (-0.5,3.5)
    to (-0.5,0.5)
    to [out=down,  in=left] (a.center) 
    to [out=right,  in=down](0.5,0.5)
    to (0.5,3.5);
\draw (-0.5,3.5)
    to (-0.5,0.5)
    to [out=down,  in=left] (a.center) 
    to [out=right,  in=down](0.5,0.5)
    to (0.5,3.5);
\draw [fill=\fillcomp, fill opacity=0.8, draw=none] (1.5,3.5)
    to (1.5,0.5)
    to [out=down,  in=left] (2,0)
    to [out=right, in=down](2.5,0.5)
    to (2.5,3.5);    
\draw (1.5,3.5)
    to (1.5,0.5)
    to [out=down,  in=left] (2,0)
    to [out=right, in=down](2.5,0.5)
    to (2.5,3.5);
\draw (0.0,-1.0) to (a.center);
\draw [fill=\fillC, fill opacity=0.8, draw=none] (-2.0,3.5)
to (-1.5,3.5)
    to (-1.5,1.5)
    to [out=down,  in=left] (-0.5,0)
    to [out=left, in=up] (-1.5, -1)
    to (-2.0, -1);
\draw (-1.5,3.5)
    to (-1.5,1.5)
    to [out=down,  in=left] (-0.5,0)
    to [out=left, in=up] (-1.5, -1);
\draw [fill=\fillC, fill opacity=0.8, draw=none] (4.0,3.5)
    to (3.5,3.5)
    to (3.5,-1)
    to (4, -1);
\draw (3.5,3.5)to (3.5,-1);
\draw (-0.5, 0) to (0,0);
\node (p) [minimum width=80pt, draw, fill=white, , scale=0.8, fill opacity=1] at (1, 2) {$\phi$};
\node [Vertex, vertex colour=white, scale=0.5] at ([xshift=-71pt, yshift=1pt]p.south){};   
\node [Vertex, vertex colour=white, scale=0.5] at ([xshift=-14.5pt, yshift=1pt]p.south){};   
\node [Vertex, vertex colour=white, scale=0.5] at ([xshift=14.5pt, yshift=1pt]p.south){};   
\node [Vertex, vertex colour=white, scale=0.5] at ([xshift=71pt, yshift=1pt]p.south){};   
\end{tikzpicture}
\end{aligned}
\right]
\quad
\Leftrightarrow
\quad
\left[
\Pd
\begin{aligned}
\begin{tikzpicture}[scale=0.4, thick]
\node (b) [Vertex, scale=\vertexsize, vertex colour=white]at (1.5, 2) {};
\node (c) [Vertex, scale=\vertexsize] at (1.5, 3.00) {};
\draw [fill=\fillcomp, fill opacity=0.8, draw=none] (-1,4)
    to (-1,-0.75)
    to  (2,-0.75)
    to(2,1.5)
    to [out=up, in=\seangle] (b.center)
    to [out=\swangle, in=up] +(-0.5,-0.5)
    to [out=down, in=down, looseness=2] +(-1,0)
    to (0.0,4);
    to (2,-0.75);
\draw (-1,4) to (-1,-0.75);
\draw (2,-0.75)
    to(2,1.5)
    to [out=up, in=\seangle] (b.center)
    to [out=\swangle, in=up] +(-0.5,-0.5)
    to [out=down, in=down, looseness=2] +(-1,0)
    to (0.0,4);
    to (2,-0.75);
\draw (b.center)
    to [out=up, in=down] (c.center);
\draw [fill=\fillcomp, fill opacity=0.8] (1,4) to (1,3.5)
    to [out=down, in=\nwangle] (c.center)
    to [out=\neangle, in=down] +(0.5,0.5)
    to (2,4);
\draw [fill=\fillC, fill opacity=0.8, draw=none] (-2.5,4.0)
to (-2,4)
        to (-2.0,3.0)
        to [out=down, in=left] (-1,2)
        to [out=left, in=up](-2.0, 1)
        to (-2.0, -0.75)
        to (-2.5,-0.75); 
\draw (-2,4)
        to (-2.0,3.0)
        to [out=down, in=left] (-1,2)
        to [out=left, in=up](-2.0, 1)
        to (-2.0, -0.75);         
\draw (-1,2)     to (1.5, 2)   ;
\draw [fill=\fillC, fill opacity=0.8, draw=none] (3.5, 4.0)
    to  (3.0, 4.0)
    to [out=down, in=right] (2,3)
    to [out=right, in=up] (3.0, 2.0)
    to (3.0, -0.75)
    to (3.5, -0.75);
\draw (3.0, 4.0)
    to [out=down, in=right] (2,3)
    to [out=right, in=up] (3.0, 2.0)
    to (3.0, -0.75);   
\draw (2,3) to (1.5, 3);
\end{tikzpicture}
\end{aligned}
\quad=\quad
\frac{\Pd}{\sqrt{n}}
\begin{aligned}
\begin{tikzpicture}[scale=0.4, thick]
\draw [fill=\fillcomp, fill opacity=0.8, draw=none] (-0.5,3.5)
    to (-0.5,-1)
    to (0.5,-1)
    to (0.5,3.5); 
\draw (-0.5,3.5)
    to (-0.5,-1);
\draw (0.5,-1)
    to (0.5,3.5);   
\draw [fill=\fillcomp, fill opacity=0.8] (1.5,3.5)
    to (1.5,0.5)
    to [out=down,  in=left] (2,0)
    to [out=right, in=down](2.5,0.5)
    to (2.5,3.5);    
\draw [fill=\fillC, fill opacity=0.8, draw=none] (-2.0,3.5)
to (-1.5,3.5)
    to (-1.5, -1)
    to (-2.0, -1);
\draw [fill=\fillC, fill opacity=0.8, draw=none] (4.0,3.5)
    to (3.5,3.5)
    to (3.5,-1)
    to (4, -1);
\draw (3.5,3.5) to (3.5, -1);
\draw (-1.5,3.5) to (-1.5, -1);
\node (p) [minimum width=80pt, draw, fill=white, scale=0.8, fill opacity=1] at (1, 1.7) {$\phi$};   
\node [Vertex, vertex colour=white, scale=0.5] at ([xshift=-71pt, yshift=1pt]p.south){};   
\node [Vertex, vertex colour=white, scale=0.5] at ([xshift=-14.5pt, yshift=1pt]p.south){};   
\node [Vertex, vertex colour=white, scale=0.5] at ([xshift=14.5pt, yshift=1pt]p.south){};   
\node [Vertex, vertex colour=white, scale=0.5] at ([xshift=71pt, yshift=1pt]p.south){};   
\end{tikzpicture}
\end{aligned}
\right]
\\
&\hspace{2cm}
\Leftrightarrow
\quad
\left[
\Pd
\begin{aligned}
\begin{tikzpicture}[scale=0.4, thick]
\node (b) [Vertex, scale=\vertexsize, vertex colour=white] at (1.5, 2.25) {};
\node (c) [Vertex, scale=\vertexsize] at (1.5, 3.25) {};
\draw [fill=\fillcomp, fill opacity=0.8, draw=none] (-1,5.25)
    to (-1,0.5)
    to  (2,0.5)
    to(2,1.75)
    to [out=up, in=\seangle] (b.center)
    to [out=\swangle, in=up] +(-0.5,-0.5)
    to [out=down, in=down, looseness=2] +(-1,0)
    to (0.0,5.25);
    \draw  (2,0.5)
    to(2,1.75)
    to [out=up, in=\seangle] (b.center)
    to [out=\swangle, in=up] +(-0.5,-0.5)
    to [out=down, in=down, looseness=2] +(-1,0)
    to (0.0,5.25);
\draw (b.center)
    to [out=up, in=down] (c.center);   
\draw [fill=\fillcomp, fill opacity=0.8, draw=none] (1,5.25) to (1,3.75)
    to [out=down, in=\nwangle] (c.center)
    to [out=\neangle, in=down] +(0.5,0.5)
    to (2,4)
    to [out=up, in=left](2.5,4.5)
    to [out=right, in=up](3,4)
    to (3,0.5)
    to(4,0.5)
    to (4,5.25);
 \draw  (1,5.25) to (1,3.75)
    to [out=down, in=\nwangle] (c.center)
    to [out=\neangle, in=down] +(0.5,0.5)
    to (2,4)
    to [out=up, in=left](2.5,4.5)
    to [out=right, in=up](3,4)
    to (3,0.5)   ;
\draw [fill=\fillC, draw=none, fill opacity=0.8] (-1.5,5)
        to (-0.5,5)
        to (-0.5, 3.25)
        to [out=down, in=left] (0.75,2.25)
        to [out=left, in=up](-0.5, 1.25)
        to (-0.5, 0.25)
        to (-1.5,0.25);  
\draw  (-0.5,5)
        to (-0.5, 3.25)
        to [out=down, in=left] (0.75,2.25)
        to [out=left, in=up](-0.5, 1.25)
        to (-0.5, 0.25);              
\draw  (0.75,2.25) to (1.5,2.25);   
\draw [fill=\fillC, draw=none, fill opacity=0.8] (4.5, 5)
    to  (3.5, 5)
    to (3.5, 4.25)
    to [out=down, in=right] (2.25,3.25)
    to [out=right, in=up] (3.5, 2.25)
    to (3.5, 0.25)
    to (4.5, 0.25);
\draw (3.5, 5)
    to (3.5, 4.25)
    to [out=down, in=right] (2.25,3.25)
    to [out=right, in=up] (3.5, 2.25)
    to (3.5, 0.25);
    \draw (2.25,3.25) to (1.5,3.25);
\end{tikzpicture}
\end{aligned}
\quad=\quad
\frac{\Pd}{\sqrt{n}}
\begin{aligned}
\begin{tikzpicture}[scale=0.4,thick]
\draw [fill=\fillcomp, draw=none, fill opacity=0.8] (-3.5,-0.75)
to (-1.5,-0.75)
to (-1.5,4)
to (-3.5, 4);
\draw (-1.5,4) to (-1.5, -0.75);
\draw [fill=\fillcomp, draw=none, fill opacity=0.8] (1.5,-0.75)
to (-0.5,-0.75)
to (-0.5,4)
to (1.5, 4);
\draw (-0.5,4) to (-0.5, -0.75);
\draw [fill=\fillC, draw=none, fill opacity=0.8] (-4,-1)
to (-2,-1)
to (-2,3.75)
to (-4, 3.75);
\draw (-2,3.75) to (-2, -1);
\draw [fill=\fillC, draw=none, fill opacity=0.8] (2,-1)
to (0,-1)
to (0,3.75)
to (2, 3.75);
\draw (0,3.75) to (0, -1);
\node [minimum width=60pt, draw, fill=white, fill opacity=1, scale=0.8] at (-1,1.5) {$\phi$};
\node [Vertex, vertex colour=white, scale=0.5] at (-1.5,0.9){};   
\node [Vertex, vertex colour=white, scale=0.5] at (-2,0.9){};   
\node [Vertex, vertex colour=white, scale=0.5] at (0,0.9){};   
\node [Vertex, vertex colour=white, scale=0.5] at (-0.5,0.9){}; 
\end{tikzpicture}
\end{aligned}
\right]
\quad
\end{align*}
For the first equivalence we compose at the bottom with the inverse of the controlled measurement vertex; for the second we perform a topological manipulation. Since $\phi$ is an arbitrary unitary 2\-cell, it is clear that the last condition is exactly that given in the statement of the lemma. $\hfill \square$

\noindent
{\bf Proof of Lemma~\ref{Eve's successful interference}.}
We investigate this scenario by applying the projector $\Ps$ on both sides of the specification. In this case the right-hand side only retains the $\Ps$ component, and the left hand side simplifies as follows:
\begin{align*}
\Ps\,
\begin{aligned}
\begin{tikzpicture} [scale=0.3,thick]
\draw [use as bounding box, draw=none] (-6.5,6) rectangle +(11,-11.3);
\node (A) [Vertex, scale=\vertexsize, vertex colour=white] at (-2,-2) {};
\node (B) [Vertex, scale=\vertexsize] at (1,-0.5) {};
\node (C) [Vertex, scale=\vertexsize, vertex colour=white] at (0, 2.5) {};
\node (D) [Vertex, scale=\vertexsize] at (0, 1.5) {};
\draw [out=up, in=down] (-2,-2) to (1,-0.5);
\draw [out=up, in=down] (0, 1.5) to (0, 2.5);
\fill [fill=\fillcomp, draw, fill opacity=0.8] (0.5, 6)
to (0.5,5)
to (0.5, 3)
to [out=down, in=down, looseness=2] (-0.5, 3)
to (-0.5,5)
to (-0.5, 6);
\fill [fill=\fillcomp, draw, fill opacity=0.8] (1.5, 6)
to (1.5,5) 
to (1.5,1)
to [out=down, in=down, looseness=2](0.5,1)
to (0.5, 1)
to [out=up, in=up, looseness=2] (-0.5,1)
to [out=down, in=left](1,-0.5)
to [out=right, in=down](2.5,1)
to (2.5,5)
to (2.5, 6);
\fill [fill=\fillcomp, fill opacity=0.8, draw] (-5.5, 6)
to (-5.5, 5)
to [out=down, in=up](-5.5, 0)
to (-5.5, -2.5)
to [out=down, in=down, looseness=2](-2.5, -2.5)
to [out=up, in=up, looseness=2](-1.5, -2.5)
to [out=down, in=down, looseness=1.9](-6.5, -2.5)
to (-6.5, 0)
to [out=up, in=down, out looseness=1.3, in looseness=0.9](-6.5, 5)
to (-6.5, 6);
\fill [fill=\fillC, draw=none, fill opacity=0.8] (-4.5, 6)
to (-4.5, -2.75)
to [out=down, in=down, looseness=2](-3.5, -2.75)
to [out=up, in=left] (-2,-2)
to [out=left, in=down] (-3.5, -1.25)
to (-3.5, 3)
to [out=up, in=up, looseness=2](-2.5, 3)
to (-2.5, 0.5)
to [out=down, in=down, looseness=2] (-1.5,0.5)
to (-1.5,1.75)
to [out=up, in=left] (0, 2.5)
to [out=left, in=down] (-1.5, 3.25)
to [out=up, in=down](-1.5, 4)
to [out=up, in=up, looseness=0.5](3.5,4)
to (3.5,2.25)
to [out=down, in=right](2, 1.5)
to [out=right, in=up](3.5, 0.75)
to (3.5, 0.25)
to [out=down, in=right](2,-0.5)
to [out=right, in=up](3.5,-1.25)
to [out=down, in=down, looseness=2](4.5,-1.25)
to (4.5,5)
to (4.5, 6)
to (3.5, 6)
to [out=down, in=down, looseness=0.5](-1.5, 6);
\draw  (-4.5, 6)
to (-4.5, -2.75)
to [out=down, in=down, looseness=2](-3.5, -2.75)
to [out=up, in=left] (-2,-2)
to [out=left, in=down] (-3.5, -1.25)
to (-3.5, 3)
to [out=up, in=up, looseness=2](-2.5, 3)
to (-2.5, 0.5)
to [out=down, in=down, looseness=2] (-1.5,0.5)
to (-1.5,1.75)
to [out=up, in=left] (0, 2.5)
to [out=left, in=down] (-1.5, 3.25)
to [out=up, in=down](-1.5, 4)
to [out=up, in=up, looseness=0.5](3.5,4)
to (3.5,2.25)
to [out=down, in=right](2, 1.5)
to [out=right, in=up](3.5, 0.75)
to (3.5, 0.25)
to [out=down, in=right](2,-0.5)
to [out=right, in=up](3.5,-1.25)
to [out=down, in=down, looseness=2](4.5,-1.25)
to (4.5,5)
to (4.5, 6)
(3.5, 6)
to [out=down, in=down, looseness=0.5](-1.5, 6);
\draw (2, -0.5) to (1, -0.5);
\draw (2, 1.5) to (0, 1.5);
\end{tikzpicture}
\end{aligned}
\quad=\quad
\Ps
\begin{aligned}
\begin{tikzpicture} [scale=0.3,thick]
\draw [use as bounding box, draw=none] (-4.5,6) rectangle +(9,-11.3);
\node (A) [Vertex, scale=\vertexsize, vertex colour=white] at (0,-2) {};
\node (B) [Vertex, scale=\vertexsize, vertex colour=white] at (1,-0.5) {};
\node (C) [Vertex, scale=\vertexsize, vertex colour=white] at (0, 3) {};
\node (D) [Vertex, scale=\vertexsize, vertex colour=white] at (0, 1.5) {};
\draw [out=up, in=down] (0,-2) to (1,-0.5);
\draw [out=up, in=down] (0, 1.5) to (0, 3);
\fill [fill=\fillcomp, draw, fill opacity=0.8] (0.5, 6)
to (0.5,5)
to (0.5, 3.5)
to [out=down, in=down, looseness=2] (-0.5, 3.5)
to (-0.5,5)
to (-0.5, 6);
\fill [fill=\fillcomp, draw, fill opacity=0.8] (1.5, 6)
to (1.5,5) 
to (1.5,1)
to [out=down, in=down, looseness=2](0.5,1)
to (0.5, 1)
to [out=up, in=up, looseness=2] (-0.5,1)
to [out=down, in=left](1,-0.5)
to [out=right, in=down](2.5,1)
to (2.5,5)
to (2.5, 6);
\fill [fill=\fillcomp, fill opacity=0.8, draw] (-3.5, 6)
to (-3.5, 5)
to [out=down, in=up](-3.5, 0)
to (-3.5, -2.5)
to [out=down, in=down, looseness=2](-0.5, -2.5)
to [out=up, in=up, looseness=2](0.5, -2.5)
to [out=down, in=down, looseness=1.9](-4.5, -2.5)
to (-4.5, 0)
to [out=up, in=down, out looseness=1.3, in looseness=0.9](-4.5, 5)
to (-4.5, 6);
\fill [fill=\fillC, draw=none, fill opacity=0.8] (-2.5, 6)
to (-2.5, 4)
to (-2.5, -2.75)
to [out=down, in=down, looseness=2](-1.5, -2.75)
to [out=up, in=left] (0,-2)
to [out=left, in=down] (-1.5, -1.25)
to [out=up, in=left](0, -0.5)
to [out=left, in=down](-1.5, 0.25)
to (-1.5, 0.75)
to [out=up, in=left](0,1.5)
to [out=left, in=down](-1.5, 2.25)
to [out=up, in=left] (0, 3)
to [out=left, in=down] (-1.5, 3.75)
to [out=up, in=down](-1.5, 4)
to [out=up, in=up, looseness=0.3](3.5,4)
to [out=down, in=down, looseness=2] (4.5,4)
to (4.5, 6)
to (3.5, 6)
to [out=down, in=down, looseness=0.3](-1.5, 6);
\draw(-2.5, 6)
to (-2.5, 4)
to (-2.5, -2.75)
to [out=down, in=down, looseness=2](-1.5, -2.75)
to [out=up, in=left] (0,-2)
to [out=left, in=down] (-1.5, -1.25)
to [out=up, in=left](0, -0.5)
to [out=left, in=down](-1.5, 0.25)
to (-1.5, 0.75)
to [out=up, in=left](0,1.5)
to [out=left, in=down](-1.5, 2.25)
to [out=up, in=left] (0, 3)
to [out=left, in=down] (-1.5, 3.75)
to [out=up, in=down](-1.5, 4)
to [out=up, in=up, looseness=0.3](3.5,4)
to [out=down, in=down, looseness=2] (4.5,4)
to (4.5, 6) 
(3.5, 6)
to [out=down, in=down, looseness=0.3](-1.5, 6);
\draw (0, -0.5) to (1, -0.5);
\end{tikzpicture}
\end{aligned}
\quad=\quad
\Ps\,
\begin{aligned}
\begin{tikzpicture} [scale=0.3,thick]
\draw [use as bounding box, draw=none] (-4.5,6) rectangle +(9,-11.3);
\fill [fill=\fillcomp, fill opacity=0.8, draw=none] (-3.5, 6)
to (-3.5, -2.5)
to [out=down, in=down, looseness=2] (-0.5, -2.5)
to (-0.5, 6)
to (0.5, 6)
to (0.5, 1)
to [out=down, in=down, looseness=2](1.5, 1)
to (1.5, 6)
to (2.5, 6)
to (2.5, 1)
to [out=down, in=up, looseness=](0.5, -2.5)
to [out=down, in=down, looseness=2] (-4.5, -2.5)
to (-4.5, 6);
\draw (-3.5, 6)
to (-3.5, -2.5)
to [out=down, in=down, looseness=2] (-0.5, -2.5)
to (-0.5, 6)
 (0.5, 6)
to (0.5, 1)
to [out=down, in=down, looseness=2](1.5, 1)
to (1.5, 6)
(2.5, 6)
to (2.5, 1)
to [out=down, in=up, looseness=](0.5, -2.5)
to [out=down, in=down, looseness=2] (-4.5, -2.5)
to (-4.5, 6);
\fill [fill=\fillC, draw=none, fill opacity=0.8] (-2.5, 6)
to (-2.5, 4)
to (-2.5, -2.75)
to [out=down, in=down, looseness=2](-1.5, -2.75)
to (-1.5, 0.75)
to [out=up, in=down](-1.5, 2.5)
to [out=up, in=up, looseness=0.8](3.5,2.5)
to [out=down, in=down, looseness=2] (4.5,2.5)
to (4.5, 6)
to (3.5, 6)
to [out=down, in=down, looseness=0.8](-1.5, 6);
\draw  (-2.5, 6)
to (-2.5, 4)
to (-2.5, -2.75)
to [out=down, in=down, looseness=2](-1.5, -2.75)
to (-1.5, 0.75)
to [out=up, in=down](-1.5, 2.5)
to [out=up, in=up, looseness=0.8](3.5,2.5)
to [out=down, in=down, looseness=2] (4.5,2.5)
to (4.5, 6)
(3.5, 6)
to [out=down, in=down, looseness=0.8](-1.5, 6);
\end{tikzpicture}
\end{aligned}
\end{align*}
Vertex colour changes are justified by changing the side from which the operations are controlled in accordance with Definition~\ref{def:flipcontrol}. By this, we can conclude that after application of $\Ps$ the QKD specification becomes a tautology. $\hfill \square$

\noindent
{\bf Proof of Lemma~\ref{lemma:Complementarity implies QKD}.}
Suppose the controlled complementarity condition~\ref{def:controlledcomplementarity} is satisfied. Then we make the following argument:
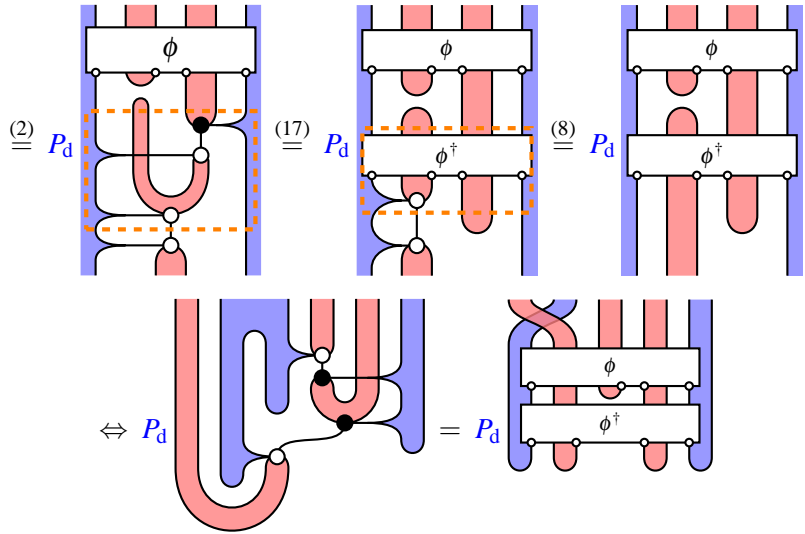
\def\quad{\hspace{0.1cm}}
\begin{gather*}
\Pd\,
\begin{aligned}
\begin{tikzpicture}[scale=0.4, xscale=-1,thick]
\node (a) [Vertex, scale=\vertexsize]at (0.5,-0.25) {};
\node (b) [Vertex, scale=\vertexsize]at (1.5, 2.) {};
\node (c) [Vertex, scale=\vertexsize, vertex colour=white] at (1.5, 3.00) {};
\node (d) [Vertex, scale=\vertexsize, vertex colour=white] at (0.5, -2) {};
\draw [fill=\fillcomp, fill opacity=0.8, draw=none] (-1,5.5)
    to (-1,1.75)
    to [out=down, in=\nwangle, out looseness=1.5] (a.center)
    to [out=\neangle, in=down, in looseness=1.5] (2,1.5)
    to [out=up, in=\seangle] (b.center)
    to [out=\swangle, in=up] +(-0.5,-0.5)
    to [out=down, in=down, looseness=2] +(-1,0)
    to (0.0,5.5);
\draw (-1,5.5)
    to (-1,1.75)
    to [out=down, in=\nwangle, out looseness=1.5] (a.center)
    to [out=\neangle, in=down, in looseness=1.5] (2,1.5)
    to [out=up, in=\seangle] (b.center)
    to [out=\swangle, in=up] +(-0.5,-0.5)
    to [out=down, in=down, looseness=2] +(-1,0)
    to (0.0,5.5);
\draw (b.center)
    to [out=up, in=down] (c.center);  
\draw [fill=\fillcomp, fill opacity=0.8, draw=none] (1,5.5) to (1,3.5)
    to [out=down, in=\nwangle] (c.center)
    to [out=\neangle, in=down] +(0.5,0.5)
    to (2,5.5);
\draw (1,5.5) to (1,3.5)
    to [out=down, in=\nwangle] (c.center)
    to [out=\neangle, in=down] +(0.5,0.5)
    to (2,5.5);
\draw [fill=\fillcomp, fill opacity=0.8, draw=none] (0,-3.5) to (0,-2.5)
    to [out=up, in=\swangle] (d.center)
    to [out=\seangle, in=up] +(0.5,-0.5)
    to (1,-3.5); 
\draw (0,-3.5) to (0,-2.5)
    to [out=up, in=\swangle] (d.center)
    to [out=\seangle, in=up] +(0.5,-0.5)
    to (1,-3.5);       
\draw [fill=\fillC, fill opacity=0.8, draw=none] (-2.5,5.5)
        to (-2.0,5.5)
        to (-2, 3)
        to [out=down, in=left] (-1,2)
        to [out=left, in=up] (-2, 1)
        to (-2.0, 0.75)
        to [out=down, in=left] (-1, -0.25)
        to [out=left, in=up] (-2.0, -1.25)
        to (-2, -3.5)
        to (-2.5,-3.5);        
\draw (-2, 5.5) to (-2.0,4.0)
        to (-2, 3)
        to [out=down, in=left] (-1,2)
        to [out=left, in=up] (-2, 1)
        to (-2.0, 0.75)
        to [out=down, in=left] (-1, -0.25)
        to [out=left, in=up] (-2.0, -1.25)
        to (-2, -3.5);         
\draw [fill=\fillC, fill opacity=0.8, draw=none] (3.5, 5.5)
        to (3, 5.5)
        to  (3.0, 4.0)
        to [out=down, in=right] (c.center)
        to [out=right, in=up] (3, 2.0)
        to (3, -1)
        to [out=down, in=right] (1.5, -2)
        to [out=right, in=up](3.0, -3)
        to (3,-3.5)
        to (3.5, -3.5);
\draw(3.0, 5.5)
    to (3,4)
    to [out=down, in=right] (c.center)
    to [out=right, in=up] (3.0, 2.0)
    to (3, -1)
    to [out=down, in=right] (1.5, -2)
    to [out=right, in=up] (3.0, -3)
    to (3, -3.5);
\draw(-1,-0.25) to (0.5, -0.25);
\draw(1.5,-2) to (0.5, -2);
\draw(-1,2) to (1.5, 2);
\draw (0.5,-2 -| a.center) to (a.center);
\end{tikzpicture}
\end{aligned}
\quad
\stackrel{\eqref{eq:flippedcomplementarity}}{=}
\quad
\Pd\,
\begin{aligned}
\begin{tikzpicture}[scale=0.4, thick, xscale=-1]
\node (a) [Vertex, scale=\vertexsize] at (0.0, 0) {};
\node (5) [Vertex, scale=\vertexsize, vertex colour=white] at (0.0, -3.5) {};
\draw [fill=\fillcomp, fill opacity=0.8, draw=none] (-0.5,3.5)
    to (-0.5,0.5)
    to [out=down,  in=left] (a.center) 
    to [out=right,  in=down](0.5,0.5)
    to (0.5,3.5);
\draw(-0.5,3.5)
    to (-0.5,0.5)
    to [out=down,  in=left] (a.center) 
    to [out=right,  in=down](0.5,0.5)
    to (0.5,3.5);
\draw [fill=\fillcomp, fill opacity=0.8, draw=none] (1.5,3.5)
    to (1.5,0.5)
    to [out=down,  in=left] (2,0)
    to [out=right, in=down](2.5,0.5)
    to (2.5,3.5);  
\draw (1.5,3.5)
    to (1.5,0.5)
    to [out=down,  in=left] (2,0)
    to [out=right, in=down](2.5,0.5)
    to (2.5,3.5);         
\draw (0.0,-3.5) to (a.center);
\draw [fill=\fillC, fill opacity=0.8, draw=none] (-2.0,3.5)
to (-1.5,3.5)
    to (-1.5,1)
    to [out=down,  in=left] (a.center)
    to [out=left, in=up] (-1.5, -1)
    to (-1.5, -2.5)
    to (-1.5, -5.5)
    to (-2.0, -5.5);
\draw (-1.5,3.5)
    to (-1.5,1)
    to [out=down,  in=left] (a.center)
    to [out=left, in=up] (-1.5, -1)
    to (-1.5, -2.5)
    to (-1.5, -5.5);    
\draw [fill=\fillC, fill opacity=0.8, draw=none] (4.0,3.5)
    to (3.5,3.5)
    to (3.5, -2.5)
    to [out=down, in=right](2, -3.5)
    to [out=right, in=up](3.5, -4.5)
    to (3.5,-5.5)
    to (4, -5.5);
\draw [fill=\fillcomp, fill opacity=0.8, draw=none](-0.5,-5.5) 
    to (-0.5,-4.1)
    to [out=up, in=up, looseness=2] (0.5,-4.1)
    to (0.5,-5.5); 
\draw (-0.5,-5.5) 
    to (-0.5,-4.1)
    to [out=up, in=up, looseness=2] (0.5,-4.1)
    to (0.5,-5.5); 
\draw (3.5,3.5)
    to (3.5, -2.5)
    to [out=down, in=right](2, -3.5)
    to [out=right, in=up](3.5, -4.5) 
    to (3.5,-5.5);    
\draw(2,-3.5) to (0, -3.5);
\node (p) [minimum width=64pt, draw, fill=white, minimum height=0.8*19pt, fill opacity=1] at (1, 2) {$\phi$};   
\node (c) [Vertex, scale=0.5, vertex colour=white] at ([xshift=-14.5pt, yshift=1pt]p.south) {};
\node (c) [Vertex, scale=0.5, vertex colour=white] at ([xshift=-71pt, yshift=1pt]p.south) {};
\node (c) [Vertex, scale=0.5, vertex colour=white] at ([xshift=71pt, yshift=1pt]p.south) {};
\node (c) [Vertex, scale=0.5, vertex colour=white] at ([xshift=14.5pt, yshift=1pt]p.south) {};
\end{tikzpicture}
\end{aligned}
\quad
\stackrel{\eqref{eq:controlledmeasurementunitarity}}{=}
\quad
\Pd\,
\begin{aligned}
\begin{tikzpicture}[scale=0.4, thick, xscale=-1]
\node (a) [Vertex, scale=\vertexsize] at (0.0, 0) {};
\node (b) [Vertex, scale=\vertexsize, vertex colour=white] at (0.0, -1.5) {};
\node (c) [Vertex, scale=\vertexsize, vertex colour=white] at (0.0, -3) {};
\node (c) [Vertex, scale=\vertexsize, vertex colour=white] at (0.0, -4.5) {};
\draw [fill=\fillcomp, fill opacity=0.8, draw=none] (-0.5,3.5)
    to (-0.5,0.5)
    to [out=down,  in=left] (a.center) 
    to [out=right,  in=down](0.5,0.5)
    to (0.5,3.5);
\draw (-0.5,3.5)
    to (-0.5,0.5)
    to [out=down,  in=left] (a.center) 
    to [out=right,  in=down](0.5,0.5)
    to (0.5,3.5);
\draw [fill=\fillcomp, fill opacity=0.8, draw=none] (1.5,3.5)
    to (1.5,0.5)
    to [out=down,  in=left] (2,0)
    to [out=right, in=down](2.5,0.5)
    to (2.5,3.5);    
\draw (1.5,3.5)
    to (1.5,0.5)
    to [out=down,  in=left] (2,0)
    to [out=right, in=down](2.5,0.5)
    to (2.5,3.5);  
\draw (0.0,-2.0) to (a.center);
\draw (0.0,-4.5) to (0,-3);
\draw [fill=\fillcomp, fill opacity=0.8](0.25, -1.7)
to (0.25,-2.8)
to [out=down, in=down](-0.25,-2.8)
to (-0.25, -1.7)
to [out=up, in=up](0.25, -1.7);
\draw [fill=\fillC, fill opacity=0.8, draw=none] (-2.0,3.5)
to (-1.5,3.5)
    to (-1.5,0.75)
    to [out=down,  in=left] (-0.5,0)
    to [out=left, in=up] (-1.5, -0.75)to
    (-1.5, -5.5)
    to (-2.0, -5.5);
\draw (-1.5,3.5)
    to (-1.5,0.75)
    to [out=down,  in=left] (-0.5,0)
    to [out=left, in=up] (-1.5, -0.75)to
    (-1.5, -5.5);
\draw (-0.5,0) to (0, 0);   
\draw [fill=\fillC, fill opacity=0.8, draw=none] (4.0,3.5)
    to (3.5,3.5)
    to (3.5,-0.75)
    to [out=down, in=right](2.5,-1.5)
    to [out=right, in=up](3.5, -2.25)
    to [out=down, in=right](2.5,-3)    
    to [out=right, in=up](3.5, -3.75)
    to [out=down, in=right]((2.5, -4.5)
    to [out=right, in=up](3.5, -5.25)
    to (3.5, -5.5)
    to (4,-5.5); 
\draw(3.5,3.5)
    to (3.5,-0.75)
    to [out=down, in=right](2.5,-1.5)
    to [out=right, in=up](3.5, -2.25)
    to [out=down, in=right](2.5,-3)    
    to [out=right, in=up](3.5, -3.75)
    to [out=down, in=right]((2.5, -4.5)
    to [out=right, in=up](3.5, -5.25)
    to (3.5, -5.5); 
\draw [fill=\fillcomp, fill opacity=0.8, draw=none](-0.5,-5.5) 
    to (-0.5,-5.1)
    to [out=up, in=up, looseness=2] (0.5,-5.1)
    to (0.5,-5.5); 
\draw (-0.5,-5.5) 
    to (-0.5,-5.1)
    to [out=up, in=up, looseness=2] (0.5,-5.1)
    to (0.5,-5.5);    
\draw (2.5, -4.5) to (0, -4.5);    
\draw (2.5, -3) to (0, -3);
\draw (2.5, -1.5) to (0, -1.5);
\node (p)[minimum width=64pt, draw, fill=white, minimum height=0.8*19pt, fill opacity=1] at (1, 2) {$\phi$};
\node (c) [Vertex, scale=0.5, vertex colour=white] at ([xshift=-14.5pt, yshift=1pt]p.south) {};
\node (c) [Vertex, scale=0.5, vertex colour=white] at ([xshift=-71pt, yshift=1pt]p.south) {};
\node (c) [Vertex, scale=0.5, vertex colour=white] at ([xshift=71pt, yshift=1pt]p.south) {};
\node (c) [Vertex, scale=0.5, vertex colour=white] at ([xshift=14.5pt, yshift=1pt]p.south) {};
\end{tikzpicture}
\end{aligned}
\\
\stackrel{\eqref{eq:top1}}{=}
\quad
\Pd\,
\begin{aligned}
\begin{tikzpicture}[scale=0.4, thick, xscale=-1]
\node (a) [Vertex, scale=\vertexsize] at (0.0, -1) {};
\node (b) [Vertex, scale=\vertexsize, vertex colour=white] at (0, -2) {};
\node (c) [Vertex, scale=\vertexsize, vertex colour=white] at (1, -4) {};
\node (b) [Vertex, scale=\vertexsize, vertex colour=white] at (1, -5) {};
\draw [fill=\fillcomp, fill opacity=0.8, draw=none] (-0.5,3)
    to (-0.5,-0.5)
    to [out=down,  in=down, looseness=2](0.5,-0.5)
    to (0.5,3);
\draw (-0.5,3)
    to (-0.5,-0.5)
    to [out=down,  in=down, looseness=2](0.5,-0.5)
    to (0.5,3);
\draw [fill=\fillcomp, fill opacity=0.8, draw=none] (1.5,3)
    to (1.5,0.9)
    to [out=down,  in=down, looseness=2] (2.5,0.9)
    to (2.5,3);      
\draw (1.5,3)
    to (1.5,0.9)
    to [out=down,  in=down, looseness=2] (2.5,0.9)
    to (2.5,3);   
\draw [fill=\fillcomp, fill opacity=0.8, draw=none] (1.5,-6)
    to (1.5,-5.6)
    to [out=up,  in=up, looseness=2] (0.5,-5.6)
    to (0.5,-6);
\draw (1.5,-6)
    to (1.5,-5.6)
    to [out=up,  in=up, looseness=2] (0.5,-5.6)
    to (0.5,-6);  
\draw (0,-2) to (a.center);
\draw (1,-4) to (1,-5);
\draw [fill=\fillcomp, fill opacity=0.8](0.25, -2.2)
to (0.25,-2.5) 
to [out=down, in=down, looseness=2](1.75,-2.5) 
to (1.75,-0.4)
to [out=up, in=up, looseness=2](2.25,-0.4)
to (2.25,-2.5)
to [out=down, in=down, looseness=2](-0.25,-2.5)
to (-0.25, -2.2)
to [out=up, in=up](0.25, -2.2);
\draw [fill=\fillC, fill opacity=0.8, draw=none] (-2.0,3)
to (-1.5,3)
    to (-1.5,1.5)
    to (-1.5, -0.5)
    to [out=down,  in=left] (-0.5,-1)
    to [out=left, in=up] (-1.5, -1.5)
    to (-1.5, -6)
    to (-2.0, -6);
\draw  (-1.5,3)
    to (-1.5,1.5)
    to (-1.5, -0.5)
    to [out=down,  in=left] (-0.5,-1)
    to [out=left, in=up] (-1.5, -1.5)
    to (-1.5, -6);
\draw [fill=\fillC, fill opacity=0.8, draw=none] (4.0,3)
    to (3.5,3)
    to (3.5,-1.5)
    to [out=down, in=right](2.5,-2)
    to [out=right, in=up](3.5, -2.5)
    to (3.5, -3.5)
    to [out=down, in=right](2.5, -4)
    to [out=right, in=up](3.5, -4.5)
    to [out=down, in=right](2.5,-5)  
    to [out=right, in=up](3.5, -5.5)
    to (3.5, -6)
    to (4,-6);
\draw(3.5,3)
    to (3.5,-1.5)
    to [out=down, in=right](2.5,-2)
    to [out=right, in=up](3.5, -2.5)
    to (3.5, -3.5)
    to [out=down, in=right](2.5, -4)
    to [out=right, in=up](3.5, -4.5)
    to [out=down, in=right](2.5,-5)  
    to [out=right, in=up](3.5, -5.5)
    to (3.5, -6);
\draw (2.5, -5) to (1, -5);
\draw (2.5, -2) to (0, -2);
\draw (2.5, -4) to (1, -4);
\draw (-0.5, -1) to (0, -1);
\node (p) [minimum width=64pt, draw, fill=white, minimum height=0.8*19pt, fill opacity=1] at (1, 1.5) {$\phi$};   
\node (c) [Vertex, scale=0.5, vertex colour=white] at ([xshift=-14.5pt, yshift=1pt]p.south) {};
\node (c) [Vertex, scale=0.5, vertex colour=white] at ([xshift=-71pt, yshift=1pt]p.south) {};
\node (c) [Vertex, scale=0.5, vertex colour=white] at ([xshift=71pt, yshift=1pt]p.south) {};
\node (c) [Vertex, scale=0.5, vertex colour=white] at ([xshift=14.5pt, yshift=1pt]p.south) {};
\node [minimum height=0.8*70pt,minimum width=0.8*100pt, draw,orange, dashed, ultra thick, scale=0.8] at (1, -2.5) {};   
\end{tikzpicture}
\end{aligned}
\quad
\stackrel{\eqref{eq:complementaryfamily}}{=}
\quad
\Pd\,
\begin{aligned}
\begin{tikzpicture}[scale=0.4, thick, xscale=-1]
\draw [fill=\fillcomp, fill opacity=0.8, draw=none] (-0.5,3.5)
    to (-0.5,-3.5)
    to [out=down,  in=down, looseness=2](0.5,-3.5)
    to (0.5,3.5);
\draw(-0.5,3.5)
    to (-0.5,-3.5)
    to [out=down,  in=down, looseness=2](0.5,-3.5)
    to (0.5,3.5);
\draw [fill=\fillcomp, fill opacity=0.8, draw=none] (1.5,3.5)
    to (1.5,1)
    to [out=down,  in=left] (2,0.5)
    to [out=right, in=down](2.5,1)
    to (2.5,3.5);
\draw  (1.5,3.5)
    to (1.5,1)
    to [out=down,  in=left] (2,0.5)
    to [out=right, in=down](2.5,1)
    to (2.5,3.5);    
\draw [fill=\fillcomp, fill opacity=0.8](1.5,-2) 
to (1.5,-0.5)
to [out=up, in=up, looseness=2](2.5,-0.5)
to (2.5,-2.5)
to [out=down, in=down, looseness=2](1.5, -2.5)
to (1.5, -2);
\draw [fill=\fillC, fill opacity=0.8, draw=none] (-2.0,3.5)
to (-1.5,3.5)
    to (-1.5,1.5)
    to (-1.5, -2)to
    (-1.5, -5.5)
    to (-2.0, -5.5);
\draw  (-1.5,3.5) to (-1.5, -5.5);
\draw [fill=\fillC, fill opacity=0.8, draw=none] (4.0,3.5)
    to (3.5,3.5)
    to (3.5, -2.25)
    to [out=down, in=right](2.5, -3) 
    to [out=right, in=up](3.5, -3.75)
    to [out=down, in=right](2.5, -4.5)
    to [out=right, in=up](3.5, -5.25)
    to (3.5, -5.5)
    to (4,-5.5);  
\draw (3.5,3.5)
    to (3.5, -2.25)
    to [out=down, in=right](2.5, -3) 
    to [out=right, in=up](3.5, -3.75)
    to [out=down, in=right](2.5, -4.5)
    to [out=right, in=up](3.5, -5.25)
    to (3.5, -5.5); 
    \draw (2.5, -4.5) to (2, -4.5); 
    \draw (2.5, -3) to (2, -3);
\draw [fill=\fillcomp, fill opacity=0.8, draw=none] (1.5,-5.5)
    to (1.5,-5)
    to [out=up,  in=up, looseness=2] (2.5,-5)
    to (2.5,-5.5);
\draw (1.5,-5.5)
    to (1.5,-5)
    to [out=up,  in=up, looseness=2] (2.5,-5)
    to (2.5,-5.5);    
\node (c) [Vertex, scale=\vertexsize, vertex colour=white] at (2, -3) {};
\node (c) [Vertex, scale=\vertexsize, vertex colour=white] at (2, -4.5) {};
\draw (2, -3) to (2, -4.5);
\node (p)[minimum width=80pt, draw, fill=white, minimum height=19pt, fill opacity=1, scale=0.8] at (1, 2) {$\phi$};   
\node (q)[minimum width=80pt, draw, fill=white, minimum height=19pt, fill opacity=1, scale=0.8] at (1, -1.5) {$\phi^{\dagger}$}; 
\node (c) [Vertex, scale=0.5, vertex colour=white] at ([xshift=-14.5pt, yshift=1pt]p.south) {};
\node (c) [Vertex, scale=0.5, vertex colour=white] at ([xshift=-71pt, yshift=1pt]p.south) {};
\node (c) [Vertex, scale=0.5, vertex colour=white] at ([xshift=71pt, yshift=1pt]p.south) {};
\node (c) [Vertex, scale=0.5, vertex colour=white] at ([xshift=14.5pt, yshift=1pt]p.south) {};
\node (c) [Vertex, scale=0.5, vertex colour=white] at ([xshift=-14.5pt, yshift=1pt]q.south) {};
\node (c) [Vertex, scale=0.5, vertex colour=white] at ([xshift=-71pt, yshift=1pt]q.south) {};
\node (c) [Vertex, scale=0.5, vertex colour=white] at ([xshift=71pt, yshift=1pt]q.south) {};
\node (c) [Vertex, scale=0.5, vertex colour=white] at ([xshift=14.5pt, yshift=1pt]q.south) {};
\node [minimum height=0.8*40pt,minimum width=0.8*80pt, draw,orange, dashed, ultra thick, scale=1] at (1, -2) {}; 
\end{tikzpicture}
\end{aligned}
\quad\stackrel{\eqref{eq:controlledmeasurementunitarity}}{=}\quad
\Pd\,
\begin{aligned}
\begin{tikzpicture}[scale=0.4, thick, xscale=-1]
\draw [fill=\fillcomp, fill opacity=0.8, draw=none] (-0.5,3.5)
    to (-0.5,-3.5)
    to [out=down,  in=down, looseness=2](0.5,-3.5)
    to (0.5,3.5);
\draw(-0.5,3.5)
    to (-0.5,-3.5)
    to [out=down,  in=down, looseness=2](0.5,-3.5)
    to (0.5,3.5);
\draw [fill=\fillcomp, fill opacity=0.8, draw=none] (1.5,3.5)
    to (1.5,1)
    to [out=down,  in=left] (2,0.5)
    to [out=right, in=down](2.5,1)
    to (2.5,3.5);
\draw  (1.5,3.5)
    to (1.5,1)
    to [out=down,  in=left] (2,0.5)
    to [out=right, in=down](2.5,1)
    to (2.5,3.5);    
\draw [fill=\fillC, fill opacity=0.8, draw=none] (-2.0,3.5)
to (-1.5,3.5)
    to (-1.5,1.5)
    to (-1.5, -2)to
    (-1.5, -5.5)
    to (-2.0, -5.5);
\draw  (-1.5,3.5) to (-1.5, -5.5);
\draw [fill=\fillC, fill opacity=0.8, draw=none] (4.0,3.5)
    to (3.5,3.5)
    to (3.5, -2.25)
    to (3.5, -5.5)
    to (4,-5.5);  
\draw (3.5,3.5)
    to (3.5, -2.25)
   to (3.5, -5.5); 
\draw [fill=\fillcomp, fill opacity=0.8, draw=none] (1.5,-5.5)
    to (1.5,-0.5)
    to [out=up,  in=up, looseness=2] (2.5,-0.5)
    to (2.5,-5.5);
\draw (1.5,-5.5)
    to (1.5,-0.5)
    to [out=up,  in=up, looseness=2] (2.5,-0.5)
    to (2.5,-5.5);    
\node (p) [minimum width=80pt, draw, fill=white, minimum height=19pt, fill opacity=1, scale=0.8] at (1, 2) {$\phi$};   
\node (q) [minimum width=80pt, draw, fill=white, minimum height=19pt, fill opacity=1, scale=0.8] at (1, -1.5) {$\phi^{\dagger}$}; 
\node (c) [Vertex, scale=0.5, vertex colour=white] at ([xshift=-14.5pt, yshift=1pt]p.south) {};
\node (c) [Vertex, scale=0.5, vertex colour=white] at ([xshift=-71pt, yshift=1pt]p.south) {};
\node (c) [Vertex, scale=0.5, vertex colour=white] at ([xshift=71pt, yshift=1pt]p.south) {};
\node (c) [Vertex, scale=0.5, vertex colour=white] at ([xshift=14.5pt, yshift=1pt]p.south) {};
\node (c) [Vertex, scale=0.5, vertex colour=white] at ([xshift=-14.5pt, yshift=1pt]q.south) {};
\node (c) [Vertex, scale=0.5, vertex colour=white] at ([xshift=-71pt, yshift=1pt]q.south) {};
\node (c) [Vertex, scale=0.5, vertex colour=white] at ([xshift=71pt, yshift=1pt]q.south) {};
\node (c) [Vertex, scale=0.5, vertex colour=white] at ([xshift=14.5pt, yshift=1pt]q.south) {};
\end{tikzpicture}
\end{aligned}
\\
\Leftrightarrow\quad
\Pd\,
\begin{aligned}
\begin{tikzpicture} [scale=0.3,thick]
\node (A) [Vertex, scale=\vertexsize, vertex colour=white] at (-2,-2) {};
\node (B) [Vertex, scale=\vertexsize] at (1,-0.5) {};
\node (C) [Vertex, scale=\vertexsize, vertex colour=white] at (0, 2.5) {};
\node (D) [Vertex, scale=\vertexsize] at (0, 1.5) {};
\draw [out=up, in=down] (-2,-2) to (1,-0.5);
\draw [out=up, in=down] (0, 1.5) to (0, 2.5);
\fill [fill=\fillcomp, draw, fill opacity=0.8]  (0.5,5)
to (0.5, 3)
to [out=down, in=down, looseness=2] (-0.5, 3)
to (-0.5,5);
\fill [fill=\fillcomp, draw, fill opacity=0.8]  (1.5,5) 
to (1.5,1)
to [out=down, in=down, looseness=2](0.5,1)
to (0.5, 1)
to [out=up, in=up, looseness=2] (-0.5,1)
to [out=down, in=left](1,-0.5)
to [out=right, in=down](2.5,1)
to (2.5,5);
\fill [fill=\fillcomp, fill opacity=0.8, draw] (-5.5, 5)
to [out=down, in=up](-5.5, 0)
to (-5.5, -2.5)
to [out=down, in=down, looseness=2](-2.5, -2.5)
to [out=up, in=up, looseness=2](-1.5, -2.5)
to [out=down, in=down, looseness=1.9](-6.5, -2.5)
to (-6.5, 0)
to [out=up, in=down, out looseness=1.3, in looseness=0.9](-6.5, 5);
\fill [fill=\fillC, fill opacity=0.8, draw] (-4.5, 5)
to (-4.5, -2.75)
to [out=down, in=down, looseness=2](-3.5, -2.75)
to [out=up, in=left] (-2,-2)
to [out=left, in=down] (-3.5, -1.25)
to (-3.5, 3)
to [out=up, in=up, looseness=2](-2.5, 3)
to (-2.5, 0.5)
to [out=down, in=down, looseness=2] (-1.5,0.5)
to (-1.5,1.75)
to [out=up, in=left] (0, 2.5)
to [out=left, in=down] (-1.5, 3.25)
to [out=up, in=down](-1.5, 5);
\fill [fill=\fillC, fill opacity=0.8, draw] (3.5,5)
to (3.5,2.25)
to [out=down, in=right](2, 1.5)
to [out=right, in=up](3.5, 0.75)
to (3.5, 0.25)
to [out=down, in=right](2,-0.5)
to [out=right, in=up](3.5,-1.25)
to [out=down, in=down, looseness=2](4.5,-1.25)
to (4.5,5);
\draw (2, -0.5) to (1, -0.5);
\draw (2, 1.5) to (0, 1.5);
\end{tikzpicture}
\end{aligned}
\quad=\quad
\Pd\,
\begin{aligned}
\begin{tikzpicture}[scale=0.3,thick]
\fill [white] (1.5, 6)
to (1.5, 6)
to [out=down, in=up](-2.5, 0)
to (-2.5, -1.5)
to [out=down, in=down, looseness=2](0.5, -1.5)
to [out=up, in=up, looseness=2](1.5, -1.5)
to [out=down, in=down, looseness=1.9](-3.5, -1.5)
to (-3.5, 0)
to [out=up, in=down, out looseness=1.3, in looseness=0.9](0.5, 6)
to (0.5, 6);
\fill [fill=\fillC, fill opacity=0.8, draw] (-2,6)
to [out=down, in=up, in looseness=1.2, out looseness=0.5](-4, 4)
to (-4,-1)
to [out=down, in=down, looseness=2](-3,-1)
to (-3,4)
to [out=up, in=down](-1,6);
\fill [fill=\fillcomp, fill opacity=0.8, draw] (-4,6)
to [out=down, in=up](-2, 4)
to (-2,-1)
to [out=down, in=down, looseness=2](-1,-1)
to (-1,4)
to [out=up, in=down, out looseness=1.2, in looseness=0.5] (-3, 6);
\fill [fill=\fillcomp, fill opacity=0.8, draw] (1, 6)
to (1,2.2)
to [out=down, in=down, looseness=2](0,2.2)
to (0,6);
\fill [fill=\fillcomp, fill opacity=0.8, draw] (2, 6)
to (2,-1)
to [out=down, in=down, looseness=2](3,-1)
to (3,6);
\fill [fill=\fillC, fill opacity=0.8, draw] (4, 6)
to (4,-1)
to [out=down, in=down, looseness=2](5,-1)
to (5,6);
\node (p) [minimum width=90pt, draw, fill=white, minimum height=19pt, fill opacity=1, scale=0.75] at (0.5, 3) {$\phi$};
\node (q) [minimum width=90pt, draw, fill=white, minimum height=19pt, fill opacity=1, scale=0.75] at (0.5, 0.5) {$\phi^{\dagger}$};
\node (c) [Vertex, scale=0.5, vertex colour=white] at ([xshift=-43pt, yshift=1pt]q.south) {};
\node (c) [Vertex, scale=0.5, vertex colour=white] at ([xshift=-99.5pt, yshift=1pt]p.south) {};
\node (c) [Vertex, scale=0.5, vertex colour=white] at ([xshift=99.5pt, yshift=1pt]q.south) {};
\node (c) [Vertex, scale=0.5, vertex colour=white] at ([xshift=43pt, yshift=1pt]q.south) {};
\node (c) [Vertex, scale=0.5, vertex colour=white] at ([xshift=14pt, yshift=1pt]p.south) {};
\node (c) [Vertex, scale=0.5, vertex colour=white] at ([xshift=-99.5pt, yshift=1pt]q.south) {};
\node (c) [Vertex, scale=0.5, vertex colour=white] at ([xshift=99.5pt, yshift=1pt]p.south) {};
\node (c) [Vertex, scale=0.5, vertex colour=white] at ([xshift=43pt, yshift=1pt]p.south) {};
\end{tikzpicture}
\end{aligned}
\end{gather*}
The final equality follows from the first chain of equalities by topological deformation. This final equality is equivalent to the statement of BB84 quantum key distribution as given in Definition~\ref{BB84QKD}, since by Lemma~\ref{Eve's successful interference} the $\Ps$ component is trivially satisfied.
$\hfill \square$

\noindent
{\bf Proof of Lemma~\ref{lemma:QKDid}.}
If Eve picks the wrong basis but does not influence the communication between Alice and Bob, their key information is still the same. We post-select on this scenario by applying a projector $\Pd$ to pools of classical information corresponding to Alice's, Bob's and Eve's basis information and by applying the comparison operation to pools corresponding to Bob's and Alice's key information:
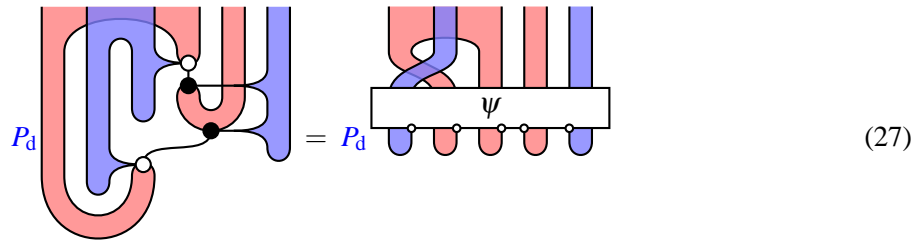
\begin{align}
\Pd
\begin{aligned}
\begin{tikzpicture} [scale=0.3,thick]
\node (A) [Vertex, scale=\vertexsize, vertex colour=white] at (-2,-2) {};
\node (B) [Vertex, scale=\vertexsize] at (1,-0.5) {};
\node (C) [Vertex, scale=\vertexsize, vertex colour=white] at (0, 2.5) {};
\node (D) [Vertex, scale=\vertexsize] at (0, 1.5) {};
\draw [out=up, in=down] (-2,-2) to (1,-0.5);
\draw [out=up, in=down] (0, 1.5) to (0, 2.5);
\fill [fill=\fillcomp, draw, fill opacity=0.8]  (1.5,5) 
to (1.5,1)
to [out=down, in=down, looseness=2](0.5,1)
to (0.5, 1)
to [out=up, in=up, looseness=2] (-0.5,1)
to [out=down, in=left](1,-0.5)
to [out=right, in=down](2.5,1)
to (2.5,5);
\fill [fill=\fillcomp, fill opacity=0.8, draw=none](0.5,5)
to (0.5, 3)
to [out=down, in=down, looseness=2] (-0.5, 3)
to [out=up, in=up, looseness=1](-5.5, 3)
to [out=down, in=up](-5.5, 0)
to (-5.5, -2.5)
to [out=down, in=down, looseness=2](-2.5, -2.5)
to [out=up, in=up, looseness=2](-1.5, -2.5)
to [out=down, in=down, looseness=1.9](-6.5, -2.5)
to (-6.5, 0)
to [out=up, in=down, out looseness=1.3, in looseness=0.9](-6.5, 5);
\draw(0.5,5)
to (0.5, 3)
to [out=down, in=down, looseness=2] (-0.5, 3)
to [out=up, in=up, looseness=1](-5.5, 3)
to [out=down, in=up](-5.5, 0)
to (-5.5, -2.5)
to [out=down, in=down, looseness=2](-2.5, -2.5)
to [out=up, in=up, looseness=2](-1.5, -2.5)
to [out=down, in=down, looseness=1.9](-6.5, -2.5)
to (-6.5, 0)
to [out=up, in=down, out looseness=1.3, in looseness=0.9](-6.5, 5);
\fill [fill=\fillC, fill opacity=0.8, draw] (-4.5, 5)
to (-4.5, -2.75)
to [out=down, in=down, looseness=2](-3.5, -2.75)
to [out=up, in=left] (-2,-2)
to [out=left, in=down] (-3.5, -1.25)
to (-3.5, 3)
to [out=up, in=up, looseness=2](-2.5, 3)
to (-2.5, 0.5)
to [out=down, in=down, looseness=2] (-1.5,0.5)
to (-1.5,1.75)
to [out=up, in=left] (0, 2.5)
to [out=left, in=down] (-1.5, 3.25)
to [out=up, in=down](-1.5, 5);
\fill [fill=\fillC, fill opacity=0.8, draw] (3.5,5)
to (3.5,2.25)
to [out=down, in=right](2, 1.5)
to [out=right, in=up](3.5, 0.75)
to (3.5, 0.25)
to [out=down, in=right](2,-0.5)
to [out=right, in=up](3.5,-1.25)
to [out=down, in=down, looseness=2](4.5,-1.25)
to (4.5,5);
\draw (2, -0.5) to (1, -0.5);
\draw (2, 1.5) to (0, 1.5);
\end{tikzpicture}
\end{aligned}
\quad=\quad
\Pd
\begin{aligned}
\begin{tikzpicture}[scale=0.3,thick]
\draw [white] (1.5, 7)
to (1.5, 7)
to [out=down, in=up](-2.5, 0)
to (-2.5, -0.5)
to [out=down, in=down, looseness=2](-2.5, -0.5)
to [out=up, in=up, looseness=2](1.5, -0.5)
to [out=down, in=down, looseness=1.9](-3.5, -0.5)
to (-3.5, 2)
to [out=up, in=down, out looseness=1.3, in looseness=0.9](0.5, 7)
to (0.5, 7);
\draw [fill=\fillcomp, fill opacity=0.8, draw=none] (-4,7)
to (-4,5)
to [out=down, in=up](-2, 3)
to (-2,1)
to [out=down, in=down, looseness=2](-1,1)
to (-1,3)
to [out=up, in=down, out looseness=1.2, in looseness=0.5] (-3, 5)
to [out=up, in=up, looseness=0.7](0, 5)
to (0,1)
to [out=down, in=down, looseness=2](1,1)
to (1,7);
\draw (-4,7)
to (-4,5)
to [out=down, in=up](-2, 3)
to (-2,1)
to [out=down, in=down, looseness=2](-1,1)
to (-1,3)
to [out=up, in=down, out looseness=1.2, in looseness=0.5] (-3, 5)
to [out=up, in=up, looseness=0.7](0, 5)
to (0,1)
to [out=down, in=down, looseness=2](1,1)
to (1,7)
(0,7);
\draw [fill=\fillC, fill opacity=0.8] (-2,7)
to (-2, 5)
to [out=down, in=up, in looseness=1.2, out looseness=0.5](-4, 3)
to (-4,1)
to [out=down, in=down, looseness=2](-3,1)
to (-3,3)
to [out=up, in=down](-1,5)
to (-1,7);
\draw [fill=\fillcomp, fill opacity=0.8] (2, 7)
to (2,1)
to [out=down, in=down, looseness=2](3,1)
to (3,7);
\draw [fill=\fillC, fill opacity=0.8] (4, 7)
to (4,1)
to [out=down, in=down, looseness=2](5,1)
to (5,7);
\node (p) [minimum width=90pt, draw, fill=white, fill opacity=1] at (0.5, 2.5) {$\psi$};
\node (c) [Vertex, scale=0.5, vertex colour=white] at ([xshift=-99.5pt, yshift=1pt]p.south) {};
\node (c) [Vertex, scale=0.5, vertex colour=white] at ([xshift=14pt, yshift=1pt]p.south) {};
\node (c) [Vertex, scale=0.5, vertex colour=white] at ([xshift=99.5pt, yshift=1pt]p.south) {};
\node (c) [Vertex, scale=0.5, vertex colour=white] at ([xshift=42.5pt, yshift=1pt]p.south){};
\node (c) [Vertex, scale=0.5, vertex colour=white] at ([xshift=-42.5pt, yshift=1pt]p.south) {};
\end{tikzpicture}
\end{aligned}
\end{align}
Using topology-preserving elementary 2\-cell operations, the first equality in~\ref{eq: QKDid} is obtained. For the second equality, up to application of $\Pd$ on the outer pools of classical information, the middle 2\-cell in equation~\ref{eq: QKDid} is a unitary, since $\psi$ is a unitary. Also, $\alpha^{\dagger}\circ\alpha$ is a positive map. The only positive unitary is the identity, hence the result is established.
$\hfill \square$

\noindent
{\bf Proof of Lemma~\ref{lemma: auxilliary MKP complementarity result}.}
We use the fact that these controlled measurements form a family of complementary controlled operations. Hence equation (\ref{th:controlled}) holds, and we combine it with the classical function $g$ to obtain the left-hand side given below:
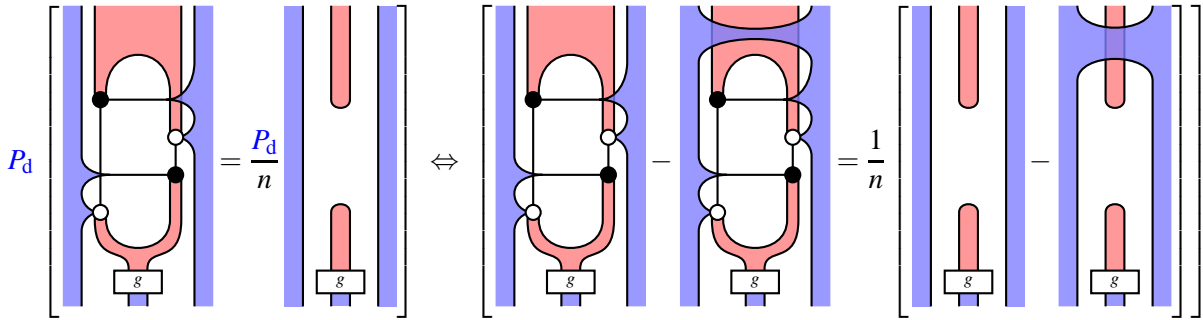
\begin{align*}
\Pd\,\left[
\begin{aligned}
\begin{tikzpicture} [thick, scale=0.25]
\draw (2,1) to (2,-1);
\draw (-2,-3) to (-2,3);
\draw (2,3) to (-2,3);
\draw (2,-1) to (-2,-1);
\draw [fill=\fillcomp, fill opacity=0.8] (2.3, 8)
to (2.3, 1.25)
to [out=down, in=down, looseness=2](1.7, 1.25)
to (1.7, 3.4)
to [out=up, in=up, looseness=2](-1.7, 3.4)
to [out=down, in=down, looseness=2] (-2.3, 3.4)
to (-2.3, 8);
\draw [fill=\fillcomp, fill opacity=0.8] 
(-0.5, -7)
to (-0.5, -6)
to [out=up, in=down, in looseness=2] (-2.3, -3.4)
to [out=up, in=up, looseness=2](-1.7, -3.4)
to [out=down, in=down, looseness=1.5](1.7, -3.4)
to (1.7, -1.25)
to [out=up, in=up, looseness=2](2.3, -1.25)
to (2.3, -3.4)
to [out=down, in=up, out looseness=2](0.5, -6)
to (0.5, -7);
\draw [fill=\fillC, fill opacity=0.8] 
(-0.5, -8)
to (-0.5, -7)
to (0.5, -7)
to (0.5, -8);
\draw [fill=\fillC, draw=none, fill opacity=0.8] (4, -8)
to (3,-8)
to (3,0)
to [out=up, in=right](2, 1)
to [out=right, in=down](3,2)
to [out=up, in=right](1.5,3)
to [out=right, in=down](3,5)
to (3,8)
to (4, 8);
\draw (3,-8)
to (3,0)
to [out=up, in=right](2, 1)
to [out=right, in=down](3,2)
to [out=up, in=right](1.5,3)
to [out=right, in=down](3,5)
to (3,8);
\draw [fill=\fillC, draw=none, fill opacity=0.8] (-4, -8)
to (-3,-8)
to (-3, -4.5)
to (-3,-4)
to [out=up, in=left](-2,-3)
to [out=left, in=down](-3, -2)
to [out=up, in=left](-1.5,-1)
to [out=left, in=down](-3,0)
to (-3, 0)
to (-3, 8)
to (-4, 8);
\draw  (-3,-8)
to (-3, -4.5)
to (-3,-4)
to [out=up, in=left](-2,-3)
to [out=left, in=down](-3, -2)
to [out=up, in=left](-1.5,-1)
to [out=left, in=down](-3,0)
to (-3, 0)
to (-3, 8);
\node [Vertex, scale=\vertexsize, vertex colour=white] at (2,1) {};
\node [Vertex, scale=\vertexsize, vertex colour=white] at (-2,-3) {};
\node [Vertex, scale=\vertexsize, vertex colour=black] at (2,-1) {};
\node [Vertex, scale=\vertexsize, vertex colour=black] at (-2,3) {};
\node (e) [fill=white, draw, minimum width=30pt, scale=0.6, fill opacity=1] at (0,-6.7) {$g$};
\end{tikzpicture}
\end{aligned}
=
\frac{\Pd}{n}
\begin{aligned}
\begin{tikzpicture}[thick, scale=0.25]
\draw [thick, fill=\fillcomp, fill opacity=0.8] (0.5, 8)
    to (0.5, 3) 
    to [out=down, in=down, looseness=1.5] (-0.5,3)
    to (-0.5, 8);
\draw [thick, fill=\fillcomp, fill opacity=0.8] (0.5, -7)
    to (0.5, -3) 
    to [out=up, in=up, looseness=1.5] (-0.5,-3)
    to (-0.5, -7);    
\draw [fill=\fillC, fill opacity=0.8, draw=none] (3, -8) to (2,-8)
to (2, 8)
to (3, 8);
\draw [fill=\fillC, fill opacity=0.8, draw=none] (-3,-8) 
to (-2,-8)
to (-2, 8)
to (-3, 8);
\draw [fill=\fillC, fill opacity=0.8] 
(-0.5, -8)
to (-0.5, -7)
to (0.5, -7)
to (0.5, -8);
\draw (-2,-8) to (-2, 8);
\draw (2,-8) to (2, 8);
\node (e) [fill=white, draw, minimum width=30pt, scale=0.6, fill opacity=1] at (0,-6.7) {$g$};
\end{tikzpicture}
\end{aligned}
\right]
\quad
\Leftrightarrow
\quad
\left[
\begin{aligned}
\begin{tikzpicture} [thick, scale=0.25]
\draw (2,1) to (2,-1);
\draw (-2,-3) to (-2,3);
\draw (2,3) to (-2,3);
\draw (2,-1) to (-2,-1);
\draw [fill=\fillcomp, fill opacity=0.8] (2.3, 8)
to (2.3, 1.25)
to [out=down, in=down, looseness=2](1.7, 1.25)
to (1.7, 3.4)
to [out=up, in=up, looseness=2](-1.7, 3.4)
to [out=down, in=down, looseness=2] (-2.3, 3.4)
to (-2.3, 8);
\draw [fill=\fillcomp, fill opacity=0.8] 
(-0.5, -7)
to (-0.5, -6)
to [out=up, in=down, in looseness=2](-2.3, -3.4)
to [out=up, in=up, looseness=2](-1.7, -3.4)
to [out=down, in=down, looseness=1.5](1.7, -3.4)
to (1.7, -1.25)
to [out=up, in=up, looseness=2](2.3, -1.25)
to (2.3, -3.4)
to [out=down, in=up, out looseness=2](0.5, -6)
to (0.5, -7);
\draw [fill=\fillC, fill opacity=0.8] 
(-0.5, -8)
to (-0.5, -7)
to (0.5, -7)
to (0.5, -8);
\draw [fill=\fillC, draw=none, fill opacity=0.8] (4, -8)
to (3,-8)
to (3,0)
to [out=up, in=right](2, 1)
to [out=right, in=down](3,2)
to [out=up, in=right](1.5,3)
to [out=right, in=down](3,5)
to (3,8)
to (4, 8);
\draw (3,-8)
to (3,0)
to [out=up, in=right](2, 1)
to [out=right, in=down](3,2)
to [out=up, in=right](1.5,3)
to [out=right, in=down](3,5)
to (3,8);
\draw [fill=\fillC, draw=none, fill opacity=0.8] (-4, -8)
to (-3,-8)
to (-3, -4.5)
to (-3,-4)
to [out=up, in=left](-2,-3)
to [out=left, in=down](-3, -2)
to [out=up, in=left](-1.5,-1)
to [out=left, in=down](-3,0)
to (-3, 0)
to (-3, 8)
to (-4, 8);
\draw  (-3,-8)
to (-3, -4.5)
to (-3,-4)
to [out=up, in=left](-2,-3)
to [out=left, in=down](-3, -2)
to [out=up, in=left](-1.5,-1)
to [out=left, in=down](-3,0)
to (-3, 0)
to (-3, 8);
\node [Vertex, scale=\vertexsize, vertex colour=white] at (2,1) {};
\node [Vertex, scale=\vertexsize, vertex colour=white] at (-2,-3) {};
\node [Vertex, scale=\vertexsize, vertex colour=black] at (2,-1) {};
\node [Vertex, scale=\vertexsize, vertex colour=black] at (-2,3) {};
\node (e) [fill=white, draw, minimum width=30pt, scale=0.6, fill opacity=1] at (0,-6.7) {$g$};
\end{tikzpicture}
\end{aligned}
-
\begin{aligned}
\begin{tikzpicture} [thick, scale=0.25]
\draw (2,1) to (2,-1);
\draw (-2,-3) to (-2,3);
\draw (2,3) to (-2,3);
\draw (2,-1) to (-2,-1);
\draw [fill=\fillcomp, fill opacity=0.8] (2.3, 8)
to (2.3, 1.25)
to [out=down, in=down, looseness=2](1.7, 1.25)
to (1.7, 3.4)
to [out=up, in=up, looseness=2](-1.7, 3.4)
to [out=down, in=down, looseness=2] (-2.3, 3.4)
to (-2.3, 8);
\draw [fill=\fillcomp, fill opacity=0.8] 
(-0.5, -7)
to (-0.5, -6)
to [out=up, in=down, in looseness=2](-2.3, -3.4)
to [out=up, in=up, looseness=2](-1.7, -3.4)
to [out=down, in=down, looseness=1.5](1.7, -3.4)
to (1.7, -1.25)
to [out=up, in=up, looseness=2](2.3, -1.25)
to (2.3, -3.4)
to [out=down, in=up, out looseness=2](0.5, -6)
to (0.5, -7);
\draw [fill=\fillC, fill opacity=0.8] 
(-0.5, -8)
to (-0.5, -7)
to (0.5, -7)
to (0.5, -8);
\draw [fill=\fillC, draw=none, fill opacity=0.8] (4, -8)
to (3,-8)
to (3,0)
to [out=up, in=right](2, 1)
to [out=right, in=down](3,2)
to [out=up, in=right](1.5,3)
to [out=right, in=down](3,5)
to (3,5)
to [out=up,in=up, looseness=0.7](-3, 5)
to (-3, 0)
to [out=down, in=left](-1.5,-1)
to [out=left, in=up](-3, -2)
to [out=down, in=left](-2,-3)
to [out=left, in=up](-3,-4)
to (-3, -4.5)
to (-3,-8)
to (-4, -8)
to (-4, 8)
to (-3, 8)
to [out=down, in=down, looseness=0.7](3, 8)
to (4, 8);
\draw (3,-8)
to (3,0)
to [out=up, in=right](2, 1)
to [out=right, in=down](3,2)
to [out=up, in=right](1.5,3)
to [out=right, in=down](3,5)
to (3,5)
to [out=up,in=up, looseness=0.7](-3, 5)
to (-3, 0)
to [out=down, in=left](-1.5,-1)
to [out=left, in=up](-3, -2)
to [out=down, in=left](-2,-3)
to [out=left, in=up](-3,-4)
to (-3, -4.5)
to (-3,-8);
\draw (-3, 8) to [out=down, in=down, looseness=0.7](3, 8);
\node [Vertex, scale=\vertexsize, vertex colour=white] at (2,1) {};
\node [Vertex, scale=\vertexsize, vertex colour=white] at (-2,-3) {};
\node [Vertex, scale=\vertexsize, vertex colour=black] at (2,-1) {};
\node [Vertex, scale=\vertexsize, vertex colour=black] at (-2,3) {};
\node (e) [fill=white, draw, minimum width=30pt, scale=0.6, fill opacity=1] at (0,-6.7) {$g$};
\end{tikzpicture}
\end{aligned}
=
\frac{1}{n}
\left[
\begin{aligned}
\begin{tikzpicture}[thick, scale=0.25]
\draw [thick, fill=\fillcomp, fill opacity=0.8] (0.5, 8)
    to (0.5, 3) 
    to [out=down, in=down, looseness=1.5] (-0.5,3)
    to (-0.5, 8);
\draw [thick, fill=\fillcomp, fill opacity=0.8] (0.5, -7)
    to (0.5, -3) 
    to [out=up, in=up, looseness=1.5] (-0.5,-3)
    to (-0.5, -7);
\draw [thick, fill=\fillC, fill opacity=0.8] (0.5, -8)
    to (0.5, -7) 
    to  (-0.5,-7)
    to (-0.5, -8);        
\draw [fill=\fillC, fill opacity=0.8, draw=none] (3, -8) to (2,-8)
to (2, 8)
to (3, 8);
\draw [fill=\fillC, fill opacity=0.8, draw=none] (-3,-8) 
to (-2,-8)
to (-2, 8)
to (-3, 8);
\draw (-2,-8) to (-2, 8);
\draw (2,-8) to (2, 8);
\node (e) [fill=white, draw, minimum width=30pt, scale=0.6, fill opacity=1] at (0,-6.7) {$g$};
\end{tikzpicture}
\end{aligned}
-
\begin{aligned}
\begin{tikzpicture}[thick, scale=0.25]
\draw [thick, fill=\fillcomp, fill opacity=0.8] (0.5, 8)
    to (0.5, 3) 
    to [out=down, in=down, looseness=1.5] (-0.5,3)
    to (-0.5, 8);
\draw [thick, fill=\fillcomp, fill opacity=0.8] (0.5, -7)
    to (0.5, -3) 
    to [out=up, in=up, looseness=1.5] (-0.5,-3)
    to (-0.5, -7);
\draw [thick, fill=\fillC, fill opacity=0.8] (0.5, -8)
    to (0.5, -7) 
    to  (-0.5,-7)
    to (-0.5, -8);  
\draw [fill=\fillC, draw=none, fill opacity=0.8] (-3,-8) 
to (-2,-8)
to (-2, 4)
to [out=up, in=up](2, 4)
to (2, -8)
to (3, -8)
to (3, 8)
to (2, 8)
to [out=down, in=down, looseness=1](-2, 8)
to (-3, 8);
\draw (-2,-8) to (-2, 4)
to [out=up, in=up](2, 4)
to (2, -8);
\draw (-2,8) to [out=down, in=down](2, 8);
\node (e) [fill=white, draw, minimum width=30pt, scale=0.6, fill opacity=1] at (0,-6.7) {$g$};
\end{tikzpicture}
\end{aligned}
\right] \right]
\end{align*}
The right-hand side is obtained by expanding out the action of the projectors $\Pd$. We next assign specific values $a,b$ to pools of classical information and perform elementary 2\-cell operations. We can replace black vertices with white, as long as we switch the side from which the vertex is controlled. Since pools of classical information exhibit topological behaviour, we can reposition them freely.

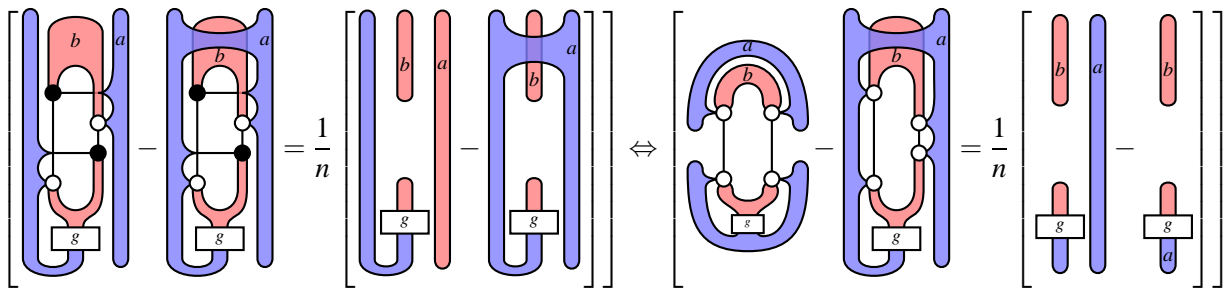
\begin{align*}
\left[
\begin{aligned}
\begin{tikzpicture} [thick, scale=0.2]
\draw (1.5,1) to (1.5,-1);
\draw (-1.5,-3) to (-1.5,3);
\draw (1.5,3) to (-1.5,3);
\draw (1.5,-1) to (-1.5,-1);
\draw [fill=\fillcomp, fill opacity=0.8] (1.8, 7)
to (1.8, 1.25)
to [out=down, in=down, looseness=2](1.2, 1.25)
to (1.2, 3.4)
to [out=up, in=up, looseness=2](-1.2, 3.4)
to [out=down, in=down, looseness=2] (-1.8, 3.4)
to (-1.8, 7)
to [out=up, in=up ,looseness=1](1.8, 7);
\draw [fill=\fillcomp, fill opacity=0.8] (0.5, -6)
to [out=up, in=down, in looseness=2](1.8, -3.4)
to (1.8, -1.25)
to [out=up, in=up, looseness=2](1.2, -1.25)
to (1.2, -3.4)
to [out=down, in=down, looseness=2](-1.2, -3.4)
to [out=up, in=up, looseness=2](-1.8, -3.4)
to [out=down, in=up, out looseness=2](-0.5, -6);
\draw [fill=\fillC, fill opacity=0.8] (3.5, -8)
to [out=down, in=down ,looseness=2](2.5,-8)
to (2.5,0)
to [out=up, in=right](1.5, 1)
to [out=right, in=down](2.5,2)
to [out=up, in=right](1.5,3)
to [out=right, in=down](2.5,5)
to (2.5,8)
to [out=up, in=up ,looseness=2](3.5, 8)
to (3.5, -8);
\draw [fill=\fillC, fill opacity=0.8] (-2.5,-8)
to (-2.5, -4.5)
to (-2.5,-4)
to [out=up, in=left](-1.5,-3)
to [out=left, in=down](-2.5, -2)
to [out=up, in=left](-1.5,-1)
to [out=left, in=down](-2.5,0)
to (-2.5, 0)
to (-2.5, 8)
to [out=up, in=up, looseness=2](-3.5, 8)
to (-3.5, -8)
to [out=down, in=down](0.5, -8)
to (0.5, -7)
to (-0.5, -7)
to  (-0.5, -8)
to [out=down, in=down](-2.5, -8);
\node [Vertex, scale=\vertexsize, vertex colour=white] at (1.5,1) {};
\node [Vertex, scale=\vertexsize, vertex colour=white] at (-1.5,-3) {};
\node [Vertex, scale=\vertexsize, vertex colour=black] at (1.5,-1) {};
\node [Vertex, scale=\vertexsize, vertex colour=black] at (-1.5,3) {};
\node (e) [fill=white, draw, minimum width=30pt, scale=0.6, fill opacity=1] at (0,-6.7) {$g$};
\node [scale=\labelsize] at (3,6.5) {$a$};
\node [scale=\labelsize] at (0,6.5) {$b$};
\end{tikzpicture}
\end{aligned}
-
\begin{aligned}
\begin{tikzpicture} [thick, scale=0.2]
\draw (1.5,1) to (1.5,-1);
\draw (-1.5,-3) to (-1.5,3);
\draw (1.5,3) to (-1.5,3);
\draw (1.5,-1) to (-1.5,-1);
\draw [fill=\fillcomp, fill opacity=0.8] (0.5, -6)
to [out=up, in=down, in looseness=2](1.8, -3.4)
to (1.8, -1.25)
to [out=up, in=up, looseness=2](1.2, -1.25)
to (1.2, -3.4)
to [out=down, in=down, looseness=2](-1.2, -3.4)
to [out=up, in=up, looseness=2](-1.8, -3.4)
to [out=down, in=up, out looseness=2](-0.5, -6);
\draw [fill=\fillcomp, fill opacity=0.8] (1.8, 7)
to (1.8, 1.25)
to [out=down, in=down, looseness=2](1.2, 1.25)
to (1.2, 3.4)
to [out=up, in=up, looseness=2](-1.2, 3.4)
to [out=down, in=down, looseness=2] (-1.8, 3.4)
to (-1.8, 7)
to [out=up, in=up ,looseness=1](1.8, 7);
\draw [fill=\fillC, fill opacity=0.8] (-2.5,-8)
to (-2.5, -4.5)
to (-2.5,-4)
to [out=up, in=left](-1.5,-3)
to [out=left, in=down](-2.5, -2)
to [out=up, in=left](-1.5,-1)
to [out=left, in=down](-2.5,0)
to (-2.5, 0)
to (-2.5, 5)
to [out=up, in=up, looseness=0.7](2.5, 5)
to (2.5, 4)
to [out=down, in=right](1.5, 3)
to [out=right, in=up](2.5, 2)
to [out=down, in=right](1.5, 1)
to [out=right, in=up](2.5, 0)
to (2.5, -8)
to [out=down, in=down, looseness=2](3.5, -8)
to (3.5, 8)
to [out=up, in=up, looseness=2](2.5, 8)
to [out=down, in=down, looseness=0.7](-2.5, 8)
to [out=up, in=up, looseness=2](-3.5, 8)
to (-3.5, -8)
to [out=down, in=down](0.5, -8)
to (0.5, -7)
to (-0.5, -7)
to  (-0.5, -8)
to [out=down, in=down](-2.5, -8);
\node [Vertex, scale=\vertexsize, vertex colour=white] at (1.5,1) {};
\node [Vertex, scale=\vertexsize, vertex colour=white] at (-1.5,-3) {};
\node [Vertex, scale=\vertexsize, vertex colour=black] at (1.5,-1) {};
\node [Vertex, scale=\vertexsize, vertex colour=black] at (-1.5,3) {};
\node (e) [fill=white, draw, minimum width=30pt, scale=0.6, fill opacity=1] at (0,-6.7) {$g$};
\node [scale=\labelsize] at (3,6.5) {$a$};
\node [scale=\labelsize] at (0,5.5) {$b$};
\end{tikzpicture}
\end{aligned}
=
\frac{1}{n}
\left[
\begin{aligned}
\begin{tikzpicture}[thick, scale=0.2]
\draw [thick, fill=\fillcomp, fill opacity=0.8] (0.5, 8)
    to (0.5, 3) 
    to [out=down, in=down, looseness=1.5] (-0.5,3)
    to (-0.5, 8)
    to [out=up, in=up, looseness=2](0.5, 8); 
\draw [fill=\fillcomp, fill opacity=0.8] (3, -8) 
to [out=down, in=down, looseness=2](2,-8)
to (2, 8)
to  [out=up, in=up, looseness=2] (3, 8)
to (3, -8);
\draw [fill=\fillC, fill opacity=0.8] (-3,-8) 
to [out=down, in=down](0.5, -8)
to (0.5, -6)
to  (-0.5, -6)
to (-0.5, -8)
to [out=down, in=down](-2,-8)
to (-2, 8)
to [out=up, in=up, looseness=2](-3, 8)
to (-3, -8);
\draw [thick, fill=\fillcomp, fill opacity=0.8] (0.5, -6)
    to (0.5, -3) 
    to [out=up, in=up, looseness=1.5] (-0.5,-3)
    to (-0.5, -6); 
\node (e) [fill=white, draw, minimum width=30pt, scale=0.6, fill opacity=1] at (0,-5.5) {$g$};
\node [scale=\labelsize] at (2.5,5) {$a$};
\node [scale=\labelsize] at (0,5) {$b$};
\end{tikzpicture}
\end{aligned}
-
\begin{aligned}
\begin{tikzpicture}[thick, scale=0.2]
\draw [thick, fill=\fillcomp, fill opacity=0.8] (0.5, 8)
    to (0.5, 3) 
    to [out=down, in=down, looseness=1.5] (-0.5,3)
    to (-0.5, 8)
    to [out=up, in=up, looseness=2](0.5, 8);  
\draw [fill=\fillC, fill opacity=0.8](-2,-8)
to (-2, 4)
to [out=up, in=up](2, 4)
to (2, -8)
to  [out=down, in=down, looseness=2](3, -8)
to (3, 8)
to [out=up, in=up, looseness=2](2, 8)
to [out=down, in=down, looseness=1](-2, 8)
to [out=up, in=up, looseness=2](-3, 8)
to (-3, -8)
to [out=down, in=down](0.5, -8)
to (0.5, -6)
to (-0.5, -6)
to (-0.5, -8)
to [out=down, in=down](-2,-8);
\draw [thick, fill=\fillcomp, fill opacity=0.8] (0.5, -6)
    to (0.5, -3) 
    to [out=up, in=up, looseness=1.5] (-0.5,-3)
    to (-0.5, -6); 
\node (e) [fill=white, draw, minimum width=30pt, scale=0.6, fill opacity=1] at (0,-5.5) {$g$};
\node [scale=\labelsize] at (2.5,6) {$a$};
\node [scale=\labelsize] at (0,4) {$b$};
\end{tikzpicture}
\end{aligned}
\right]
\right]
\Leftrightarrow
\left[
\begin{aligned}
\begin{tikzpicture} [scale=0.4,thick]
\draw [fill=\fillC, fill opacity=0.8] (0,0.5) 
to [out=right, in=down, looseness=1.3](2,3.2)
to [out=up, in=up, looseness=2](1.5,3.2)
to [out=down, in=right](0.8, 2.9)
to [out=right, in=up](1.4, 2.6)
to [out=down, in=right](0.5,1)
to [out=left, in=down](0.25, 1.5)
to (-0.25, 1.5)
to [out=down, in=right](-0.5, 1)
to [out=left, in=down](-1.4, 2.6)
to [out=up, in=left](-0.8, 2.9)
to [out=left, in=down](-1.5,3.2)
to  [out=up, in=up, looseness=2](-2,3.2)
to [out=down, in=left, looseness=1.3](0,0.5);
\draw [fill=\fillcomp, fill opacity=0.8] (0.25, 1.5)
to [out=up,in=down](1,2.5)
to [out=up, in=up, looseness=2](0.5, 2.5)
to [out=down, in=down, looseness=1.1](-0.5,2.5)
to [out=up, in=up, looseness=2](-1,2.5)
to [out=down, in=up](-0.25, 1.5);
\draw [thick](0.8,3) to (0.8,5);  
\draw [thick](-0.8,3) to (-0.8,5);  
\node [draw, scale=0.4, minimum width=30, minimum height=17,fill=white, fill opacity=1] at (0, 1.4) {$g$};
\node [Vertex, scale=\vertexsize, vertex colour=white] at (0.8, 2.9) {};
\node [Vertex, scale=\vertexsize, vertex colour=white]at (-0.8, 2.9) {};
\draw [fill=\fillC, yscale=-1, fill opacity=0.8] (0,-7.5) 
to [out=right, in=down, looseness=1.1](2,-4.8)
to [out=up, in=up, looseness=2](1.5,-4.8)
to [out=down, in=right](0.8, -5.1)
to [out=right, in=up](1.4, -5.4)
to [out=down, in=right](0, -7.)
to [out=left, in=down](-1.4, -5.4)
to [out=up, in=left](-0.8, -5.1)
to [out=left, in=down](-1.5,-4.8)
to  [out=up, in=up, looseness=2](-2,-4.8)
to [out=down, in=left, looseness=1.1](0,-7.5);
\draw [fill=\fillcomp, yscale=-1, fill opacity=0.8](0, -6.7)
to [out=right,in=down](1.1,-5.5)
to [out=up, in=up, looseness=2](0.5, -5.5)
to [out=down, in=down, looseness=2](-0.5,-5.5)
to [out=up, in=up, looseness=2](-1.1,-5.5)
to [out=down, in=left](0, -6.7);
\node [Vertex, scale=\vertexsize, vertex colour=white] at (0.8, 5.1) {};
\node [Vertex, scale=\vertexsize, vertex colour=white]at (-0.8, 5.1) {};
\node [scale=\labelsize] at (0, 7.3) {$a$};
\node [scale=\labelsize] at (0, 6.4) {$b$};
\end{tikzpicture}
\end{aligned}
-
\begin{aligned}
\begin{tikzpicture} [thick, scale=0.2]
\draw (1.5,1) to (1.5,-1);
\draw (-1.5,-3) to (-1.5,3);
\draw [fill=\fillcomp, fill opacity=0.8] (0.5, -6)
to [out=up, in=down, in looseness=2](1.8, -3.4)
to (1.8, -1.25)
to [out=up, in=up, looseness=2](1.2, -1.25)
to (1.2, -3.4)
to [out=down, in=down, looseness=2](-1.2, -3.4)
to [out=up, in=up, looseness=2](-1.8, -3.4)
to [out=down, in=up, out looseness=2](-0.5, -6);
\draw [fill=\fillcomp, fill opacity=0.8] (1.8, 7)
to (1.8, 1.25)
to [out=down, in=down, looseness=2](1.2, 1.25)
to (1.2, 3.4)
to [out=up, in=up, looseness=2](-1.2, 3.4)
to [out=down, in=down, looseness=2] (-1.8, 3.4)
to (-1.8, 7)
to [out=up, in=up ,looseness=1](1.8, 7);
\draw [fill=\fillC, fill opacity=0.8] (-2.5,-8)
to (-2.5, -4.5)
to (-2.5,-4)
to [out=up, in=left](-1.5,-3)
to [out=left, in=down](-2.5, -2)
to (-2.5, 2)
to [out=up, in=left](-1.5, 3)
to [out=left, in=down](-2.5, 4)
to (-2.5, 5)
to [out=up, in=up, looseness=0.7](2.5, 5)
to (2.5, 4)
to [out=down, in=right](1.5, 1)
to [out=right, in=up](2.5, 0)
to [out=down, in=right](1.5, -1)
to [out=right, in=up](2.5, -2)
to (2.5, -8)
to [out=down, in=down, looseness=2](3.5, -8)
to (3.5, 8)
to [out=up, in=up, looseness=2](2.5, 8)
to [out=down, in=down, looseness=0.7](-2.5, 8)
to [out=up, in=up, looseness=2](-3.5, 8)
to (-3.5, -8)
to [out=down, in=down](0.5, -8)
to (0.5, -7)
to (-0.5, -7)
to  (-0.5, -8)
to [out=down, in=down](-2.5, -8);
\node [Vertex, scale=\vertexsize, vertex colour=white] at (1.5,1) {};
\node [Vertex, scale=\vertexsize, vertex colour=white] at (-1.5,-3) {};
\node [Vertex, scale=\vertexsize, vertex colour=white] at (1.5,-1) {};
\node [Vertex, scale=\vertexsize, vertex colour=white] at (-1.5,3) {};
\node (e) [fill=white, draw, minimum width=30pt, scale=0.6, fill opacity=1] at (0,-6.7) {$g$};
\node [scale=\labelsize] at (3,6.5) {$a$};
\node [scale=\labelsize] at (0,5.5) {$b$};
\end{tikzpicture}
\end{aligned}
=
\frac{1}{n}
\left[
\begin{aligned}
\begin{tikzpicture}[thick, scale=0.2]
\draw [thick, fill=\fillcomp, fill opacity=0.8] (0.5, 8)
    to (0.5, 3) 
    to [out=down, in=down, looseness=1.5] (-0.5,3)
    to (-0.5, 8)
    to [out=up, in=up, looseness=2](0.5, 8); 
\draw [fill=\fillC, fill opacity=0.8] (3, -8) 
to [out=down, in=down, looseness=2](2,-8)
to (2, 8)
to  [out=up, in=up, looseness=2] (3, 8)
to (3, -8);
\draw [fill=\fillcomp, fill opacity=0.8](0.5, -5.5)
to (0.5, -3)
to [out=up, in=up, looseness=1.5]  (-0.5, -3)
to (-0.5, -5.5)
to [out=down, in=down, looseness=2](0.5, -5.5);
\draw [fill=\fillC, fill opacity=0.8](0.5, -8)
to (0.5, -5.5)
to  (-0.5, -5.5)
to (-0.5, -8)
to [out=down, in=down, looseness=2](0.5, -8);
\node (e) [fill=white, draw, minimum width=30pt, scale=0.6, fill opacity=1] at (0,-5.5) {$g$};
\node [scale=\labelsize] at (2.5,5) {$a$};
\node [scale=\labelsize] at (0,5) {$b$};
\end{tikzpicture}
\end{aligned}
-\begin{aligned}
\begin{tikzpicture}[thick, scale=0.2]
\draw [thick, fill=\fillcomp, fill opacity=0.8] (0.5, 8)
    to (0.5, 3) 
    to [out=down, in=down, looseness=1.5] (-0.5,3)
    to (-0.5, 8)
    to [out=up, in=up, looseness=2](0.5, 8); 
\draw [fill=\fillcomp, fill opacity=0.8](0.5, -5.5)
to (0.5, -3)
to [out=up, in=up, looseness=1.5]  (-0.5, -3)
to (-0.5, -5.5)
to [out=down, in=down, looseness=2](0.5, -5.5);
\draw [fill=\fillC, fill opacity=0.8](0.5, -8)
to (0.5, -5.5)
to  (-0.5, -5.5)
to (-0.5, -8)
to [out=down, in=down, looseness=2](0.5, -8);
\node (e) [fill=white, draw, minimum width=30pt, scale=0.6, fill opacity=1] at (0,-5.5) {$g$};
\node [scale=\labelsize] at (0,-7.5) {$a$};
\node [scale=\labelsize] at (0,5) {$b$};
\end{tikzpicture}
\end{aligned}
\right]
\right]
\end{align*}
After we cancel out measurement and encoding operations controlled by the same pools of classical information, the equation is simplified to:
\begin{align}
\begin{aligned}
\begin{tikzpicture} [scale=0.4,thick]
\draw [fill=\fillC, fill opacity=0.8] (0,0.5) 
to [out=right, in=down, looseness=1.3](2,3.2)
to [out=up, in=up, looseness=2](1.5,3.2)
to [out=down, in=right](0.8, 2.9)
to [out=right, in=up](1.4, 2.6)
to [out=down, in=right](0.5,1)
to [out=left, in=down](0.25, 1.5)
to (-0.25, 1.5)
to [out=down, in=right](-0.5, 1)
to [out=left, in=down](-1.4, 2.6)
to [out=up, in=left](-0.8, 2.9)
to [out=left, in=down](-1.5,3.2)
to  [out=up, in=up, looseness=2](-2,3.2)
to [out=down, in=left, looseness=1.3](0,0.5);
\draw [fill=\fillcomp, fill opacity=0.8] (0.25, 1.5)
to [out=up,in=down](1,2.5)
to [out=up, in=up, looseness=2](0.5, 2.5)
to [out=down, in=down, looseness=1.1](-0.5,2.5)
to [out=up, in=up, looseness=2](-1,2.5)
to [out=down, in=up](-0.25, 1.5);
\draw [thick](0.8,3) to (0.8,5);  
\draw [thick](-0.8,3) to (-0.8,5);  
\node [draw, scale=0.4, minimum width=30, minimum height=17,fill=white, fill opacity=1] at (0, 1.4) {$g$};
\node [Vertex, scale=\vertexsize, vertex colour=white] at (0.8, 2.9) {};
\node [Vertex, scale=\vertexsize, vertex colour=white]at (-0.8, 2.9) {};
\draw [fill=\fillC, yscale=-1, fill opacity=0.8] (0,-7.5) 
to [out=right, in=down, looseness=1.1](2,-4.8)
to [out=up, in=up, looseness=2](1.5,-4.8)
to [out=down, in=right](0.8, -5.1)
to [out=right, in=up](1.4, -5.4)
to [out=down, in=right](0, -7.)
to [out=left, in=down](-1.4, -5.4)
to [out=up, in=left](-0.8, -5.1)
to [out=left, in=down](-1.5,-4.8)
to  [out=up, in=up, looseness=2](-2,-4.8)
to [out=down, in=left, looseness=1.1](0,-7.5);
\draw [fill=\fillcomp, yscale=-1, fill opacity=0.8](0, -6.7)
to [out=right,in=down](1.1,-5.5)
to [out=up, in=up, looseness=2](0.5, -5.5)
to [out=down, in=down, looseness=2](-0.5,-5.5)
to [out=up, in=up, looseness=2](-1.1,-5.5)
to [out=down, in=left](0, -6.7);
\node [Vertex, scale=\vertexsize, vertex colour=white] at (0.8, 5.1) {};
\node [Vertex, scale=\vertexsize, vertex colour=white]at (-0.8, 5.1) {};
\node [scale=\labelsize] at (0, 7.2) {$a$};
\node [scale=\labelsize] at (0, 6.35) {$b$};
\end{tikzpicture}
\end{aligned}
\quad&=\quad
\begin{aligned}
\begin{tikzpicture} [thick, scale=0.5]
\draw [fill=\fillC, fill opacity=0.8] (0,1)
to (0,0)
to [out=down,in=down, looseness=2](1,0)
to (1,1)
to (0,1);
\draw [fill=\fillcomp, fill opacity=0.8] (0,2)
to (0,1)
to (1,1)
to (1,2)
to [out=up, in=up, looseness=2](0,2);
\node [draw, minimum width=40, 
minimum height=27,fill=white, scale=0.5, fill opacity=1]
 at (0.5, 1) {$g$}; 
   \node [scale=\labelsize]
 at (0.5, 0) {$a$};
   \node [scale=\labelsize]
 at (0.5, 2) {$b$};
\end{tikzpicture}
\end{aligned}
\quad+\quad\frac{1}{n}
\left(
\begin{aligned}
\begin{tikzpicture} [thick, scale=0.5]
\draw [fill=\fillC, fill opacity=0.8] (0,1)
to (0,0)
to [out=down,in=down, looseness=2](1,0)
to (1,1)
to (0,1);
\draw [fill=\fillcomp, fill opacity=0.8] (0,2)
to (0,1)
to (1,1)
to (1,2)
to [out=up, in=up, looseness=2](0,2);
\node [draw, minimum width=40, 
minimum height=27,fill=white, scale=0.5, fill opacity=1]
 at (0.5, 1) {$g$}; 
\draw [fill=\fillC, fill opacity=0.8] (2.5,1)
to [out=up, in=up, looseness=2](1.5,1)
to [out=down,in=down, looseness=2](2.5,1);
\draw [fill=\fillcomp, fill opacity=0.8] (3,1)
to [out=up, in=up, looseness=2](4,1)
to [out=down,in=down, looseness=2](3,1);
   \node [scale=\labelsize]
 at (2, 1) {$a$};
   \node [scale=\labelsize]
 at (3.5, 1) {$b$};
\end{tikzpicture}
\end{aligned}
\quad-\quad
\begin{aligned}
\begin{tikzpicture} [thick, scale=0.5]
\draw [fill=\fillC, fill opacity=0.8] (0,1)
to (0,0)
to [out=down,in=down, looseness=2](1,0)
to (1,1)
to (0,1);
\draw [fill=\fillcomp, fill opacity=0.8] (0,2)
to (0,1)
to (1,1)
to (1,2)
to [out=up, in=up, looseness=2](0,2);
\node [draw, minimum width=40, 
minimum height=27,fill=white, scale=0.5, fill opacity=1]
 at (0.5, 1) {$g$}; 
\draw [fill=\fillcomp, fill opacity=0.8] (2.5,1)
to [out=up, in=up, looseness=2](1.5,1)
to [out=down,in=down, looseness=2](2.5,1);
   \node [scale=\labelsize]
 at (0.5, 0) {$a$};
   \node [scale=\labelsize]
 at (2, 1) {$b$};
\end{tikzpicture}
\end{aligned}
\right)
\\[-20pt]
\nonumber
&= \quad
\def\quad{\hspace{0.3cm}}
\begin{aligned}
\begin{tikzpicture} [thick, scale=0.5]
\draw [fill=\fillC, fill opacity=0.8] (0,1)
to (0,0)
to [out=down,in=down, looseness=2](1,0)
to (1,1)
to (0,1);
\draw [fill=\fillcomp, fill opacity=0.8] (0,2)
to (0,1)
to (1,1)
to (1,2)
to [out=up, in=up, looseness=2](0,2);
\node [draw, minimum width=40, 
minimum height=27,fill=white, scale=0.5, fill opacity=1]
 at (0.5, 1) {$g$}; 
   \node [scale=\labelsize]
 at (0.5, 0) {$a$};
   \node [scale=\labelsize]
 at (0.5, 2) {$b$};
\end{tikzpicture}
\end{aligned}
\quad+\quad
\frac 1 n (n+1) \quad-\quad \frac 1 n
\quad= \quad
\begin{aligned}
\begin{tikzpicture} [thick, scale=0.5]
\draw [fill=\fillC, fill opacity=0.8] (0,1)
to (0,0)
to [out=down,in=down, looseness=2](1,0)
to (1,1)
to (0,1);
\draw [fill=\fillcomp, fill opacity=0.8] (0,2)
to (0,1)
to (1,1)
to (1,2)
to [out=up, in=up, looseness=2](0,2);
\node [draw, minimum width=40, 
minimum height=27,fill=white, scale=0.5, fill opacity=1]
 at (0.5, 1) {$g$}; 
   \node [scale=\labelsize]
 at (0.5, 0) {$a$};
   \node [scale=\labelsize]
 at (0.5, 2) {$b$};
\end{tikzpicture}
\end{aligned}
\quad+\quad
1
\end{align}

\noindent
{\bf Proof of Theorem~\ref{lemma: MKP correctness}.}
By Lemma~\ref{lemma: prime power} a suitable family of $n^2$ functions $f_i : [n+1] \to [n]$ and a complementary family of controlled measurements in $n+1$ bases exist. The latter by defining a controlled operation to pick one of the $n+1$ complementary bases to measure in. For each $f_i$ we define a state $\mu_{f_i}$ in accordance with Lemma~\ref{def: Alice's basis}. By Lemma~\ref{lemma:orthonormal basis} states $\ket{\mu _{f_i}}$ form an orthonormal basis $\mu$ that we use to solve the problem. 
The scheme $\mathrm{MK}_{\mdot, \mu, f}$ then simplifies to:
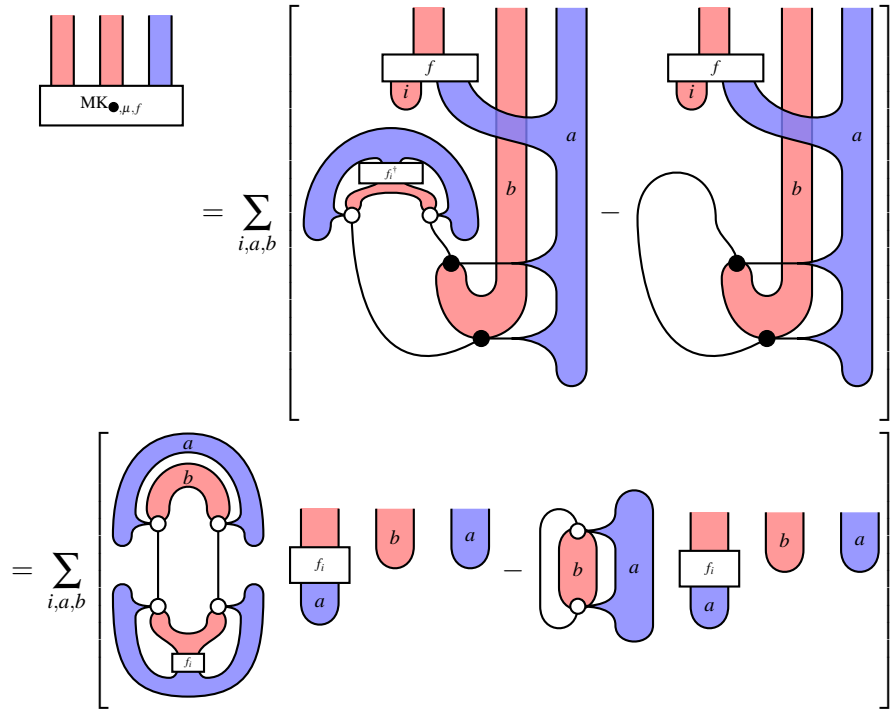
\begin{align*}
\begin{aligned}
\begin{tikzpicture} [scale=0.4,thick]
\draw [white] (-0.5, 6)
to [out=down, in=up, out looseness=1.5, in looseness=0.8] (3.5,2.5)
to [out=down, in=right](2, 1.5)
to [out=right, in=up](3.5,0.5)
to (3.5, 0)
to [out=down, in=right](2,-1)
to [out=right, in=up](3.5,-2)
to [out=down, in=down, looseness=2](4.5,-2)
to (4.5,10)
to (3.5, 10)
to (3.5, 5)
to [out=down, in=down, in looseness=1.3](0.5, 6);
\fill [fill=\fillcomp, fill opacity=0.8, draw] (0.75, 10) 
to (0.75, 7)
to (0, 7)
to (0, 10);
\fill [fill=\fillcomp, fill opacity=0.8, draw] (1.625, 10) 
to (1.625, 7)
to (2.375, 7)
to (2.375, 10);
\fill [fill=\fillC, fill opacity=0.8, draw] (3.25, 10) 
to (3.25, 7)
to (4, 7)
to (4, 10);
\node [draw, fill=white, minimum width=90pt, minimum height=25pt,scale=0.6, fill opacity=1]at (2, 7) {$\mathrm{MK}_{\mdot,\mu,f}$};
\end{tikzpicture}
\end{aligned}
\quad=\quad
\sum_{i,a,b}
\left[
\begin{aligned}
\begin{tikzpicture} [scale=0.4,thick]
\node (B) [Vertex, scale=\vertexsize] at (1,-1) {};
\node (D) [Vertex, scale=\vertexsize] at (0, 1.5) {};
\draw (-3.3,3) to [out=down, in=\swangle, in looseness=2, out looseness=0.5] (1,-1);
\draw [out=up, in=down] (0, 1.5) to (-0.7, 3);
\fill [fill=\fillcomp, fill opacity=0.8, draw] (1.5, 10)
to (1.5,5) 
to (1.5,1)
to [out=down, in=down, looseness=2](0.5,1)
to [out=up, in=up, looseness=2] (-0.5,1)
to [out=down, in=left](1,-1)
to [out=right, in=down](2.5,0.5)
to (2.5,5)
to (2.5, 10);
\fill [fill=\fillcomp, fill opacity=0.8, draw] (-1.25, 10) 
to (-1.25, 8)
to (-0.25, 8)
to (-0.25, 10);
\fill [fill=\fillC, fill opacity=0.8, draw=none] (-0.5, 8)
to [out=down, in=up, out looseness=1.5, in looseness=0.8] (3.5,4.5)
to (3.5, 2.5)
to [out=down, in=right](2, 1.5)
to [out=right, in=up](3.5,0.5)
to (3.5, 0)
to [out=down, in=right](2,-1)
to [out=right, in=up](3.5,-2)
to [out=down, in=down, looseness=2](4.5,-2)
to (4.5,10)
to (3.5, 10)
to (3.5, 7)
to [out=down, in=down, in looseness=1.3](0.5, 8);
\draw (-0.5, 8)
to [out=down, in=up, out looseness=1.5, in looseness=0.8] (3.5,4.5)
to (3.5, 2.5)
to [out=down, in=right](2, 1.5)
to [out=right, in=up](3.5,0.5)
to (3.5, 0)
to [out=down, in=right](2,-1)
to [out=right, in=up](3.5,-2)
to [out=down, in=down, looseness=2](4.5,-2)
to (4.5,10)
 (3.5, 10)
to (3.5, 7)
to [out=down, in=down, in looseness=1.3](0.5, 8);
\draw (2, -1) to (1, -1);
\draw (2, 1.5) to (0, 1.5);
\draw [fill=\fillC, fill opacity=0.8,yscale=-1] (-2,-6) 
to [out=right, in=down](0.9,-2.8)
to [out=up, in=up, looseness=2](0,-2.8)
to [out=down, in=right](-0.7, -3.1)
to [out=right, in=up](-0.1, -3.4)
to [out=down, in=right](-1.5,-5)
to [out=left, in=down](-1.75, -4.5)
to (-2.25, -4.5)
to [out=down, in=right](-2.5, -5)
to [out=left, in=down](-3.9, -3.4)
to [out=up, in=left](-3.3, -3.1)
to [out=left, in=down](-4,-2.8)
to  [out=up, in=up, looseness=2](-4.9,-2.8)
to [out=down, in=left](-2,-6);
\draw [fill=\fillcomp, fill opacity=0.8,yscale=-1](-1.75, -4.5)
to [out=up,in=down](-0.5,-3.5)
to [out=up, in=up, looseness=2](-1, -3.5)
to [out=down, in=down, looseness=0.6](-3,-3.5)
to [out=up, in=up, looseness=2](-3.5,-3.5)
to [out=down, in=up](-2.25, -4.5); 
\fill [fill=\fillcomp, fill opacity=0.8, draw](-2,8)
to (-2, 7.2)
to [out=down, in=down, looseness=2](-1, 7.2)
to (-1,8);
\node [scale=\labelsize] at (-1.5, 7.2) {$i$};
\node [scale=\labelsize] at (4, 5.7) {$a$};
\node [scale=\labelsize] at (2, 4) {$b$};
\node [draw, scale=0.4, minimum width=60, minimum height=17,fill=white, fill opacity=1] at (-2, 4.5) {${f_i}^{\dagger}$};
\node [Vertex, scale=\vertexsize, vertex colour=white] at (-0.7, 3.1) {};
\node [Vertex, scale=\vertexsize, vertex colour=white]at (-3.3, 3.1) {};
\node [draw, fill=white, minimum width=60pt, scale=0.6, fill opacity=1]at (-0.7, 8) {$f$};
\end{tikzpicture}
\end{aligned}
\quad-\quad
\begin{aligned}
\begin{tikzpicture} [scale=0.4,thick]
\node (B) [Vertex, scale=\vertexsize] at (1,-1) {};
\node (D) [Vertex, scale=\vertexsize] at (0, 1.5) {};
\draw  (0, 1.5) 
to [out=up, in=down](-0.7, 3)
to [out=up, in=up, looseness=2](-3.3,3)
to [out=down, in=\swangle, in looseness=2, out looseness=0.5] (1,-1);
\fill [fill=\fillcomp, fill opacity=0.8, draw](-2,8)
to (-2, 7.2)
to [out=down, in=down, looseness=2](-1, 7.2)
to (-1,8);
\fill [fill=\fillcomp, fill opacity=0.8, draw] (1.5, 10)
to (1.5,5) 
to (1.5,1)
to [out=down, in=down, looseness=2](0.5,1)
to [out=up, in=up, looseness=2] (-0.5,1)
to [out=down, in=left](1,-1)
to [out=right, in=down](2.5,0.5)
to (2.5,5)
to (2.5, 10);
\fill [fill=\fillcomp, fill opacity=0.8, draw] (-1.25, 10) 
to (-1.25, 8)
to (-0.25, 8)
to (-0.25, 10);
\fill [fill=\fillC, fill opacity=0.8, draw=none] (-0.5, 8)
to [out=down, in=up, out looseness=1.5, in looseness=0.8] (3.5,4.5)
to (3.5, 2.5)
to [out=down, in=right](2, 1.5)
to [out=right, in=up](3.5,0.5)
to (3.5, 0)
to [out=down, in=right](2,-1)
to [out=right, in=up](3.5,-2)
to [out=down, in=down, looseness=2](4.5,-2)
to (4.5,10)
to (3.5, 10)
to (3.5, 7)
to [out=down, in=down, in looseness=1.3](0.5, 8);
\draw (-0.5, 8)
to [out=down, in=up, out looseness=1.5, in looseness=0.8] (3.5,4.5)
to (3.5, 2.5)
to [out=down, in=right](2, 1.5)
to [out=right, in=up](3.5,0.5)
to (3.5, 0)
to [out=down, in=right](2,-1)
to [out=right, in=up](3.5,-2)
to [out=down, in=down, looseness=2](4.5,-2)
to (4.5,10)
 (3.5, 10)
to (3.5, 7)
to [out=down, in=down, in looseness=1.3](0.5, 8);
\draw (2, -1) to (1, -1);
\draw (2, 1.5) to (0, 1.5);
\node [scale=\labelsize] at (-1.5, 7.2) {$i$};
\node [scale=\labelsize] at (4, 5.7) {$a$};
\node [scale=\labelsize] at (2, 4) {$b$};
\node [draw, fill=white, minimum width=60pt, scale=0.6, fill opacity=1]at (-0.7, 8) {$f$};
\end{tikzpicture}
\end{aligned}
\,\,\right]\\
=\quad
\sum_{i,a,b}
\left[
\begin{aligned}
\begin{tikzpicture} [thick, scale=0.5]
\draw [fill=\fillC, fill opacity=0.8] (0,0.5) 
to [out=right, in=down, looseness=1.3](2,3.2)
to [out=up, in=up, looseness=2](1.5,3.2)
to [out=down, in=right](0.8, 2.9)
to [out=right, in=up](1.4, 2.6)
to [out=down, in=right](0.5,1)
to [out=left, in=down](0.25, 1.5)
to (-0.25, 1.5)
to [out=down, in=right](-0.5, 1)
to [out=left, in=down](-1.4, 2.6)
to [out=up, in=left](-0.8, 2.9)
to [out=left, in=down](-1.5,3.2)
to  [out=up, in=up, looseness=2](-2,3.2)
to [out=down, in=left, looseness=1.3](0,0.5);
\draw [fill=\fillcomp, fill opacity=0.8] (0.25, 1.5)
to [out=up,in=down](1,2.5)
to [out=up, in=up, looseness=2](0.5, 2.5)
to [out=down, in=down, looseness=1.1](-0.5,2.5)
to [out=up, in=up, looseness=2](-1,2.5)
to [out=down, in=up](-0.25, 1.5);
\draw [thick](0.8,3) to (0.8,5);  
\draw [thick](-0.8,3) to (-0.8,5);  
\node [draw, scale=0.4, minimum width=30, minimum height=17,fill=white, fill opacity=1] at (0, 1.4) {$f_i$};
\node [Vertex, scale=\vertexsize, vertex colour=white] at (0.8, 2.9) {};
\node [Vertex, scale=\vertexsize, vertex colour=white]at (-0.8, 2.9) {};
\draw [fill=\fillC, yscale=-1, fill opacity=0.8] (0,-7.5) 
to [out=right, in=down, looseness=1.1](2,-4.8)
to [out=up, in=up, looseness=2](1.5,-4.8)
to [out=down, in=right](0.8, -5.1)
to [out=right, in=up](1.4, -5.4)
to [out=down, in=right](0, -7.)
to [out=left, in=down](-1.4, -5.4)
to [out=up, in=left](-0.8, -5.1)
to [out=left, in=down](-1.5,-4.8)
to  [out=up, in=up, looseness=2](-2,-4.8)
to [out=down, in=left, looseness=1.1](0,-7.5);
\draw [fill=\fillcomp, yscale=-1, fill opacity=0.8](0, -6.7)
to [out=right,in=down](1.1,-5.5)
to [out=up, in=up, looseness=2](0.5, -5.5)
to [out=down, in=down, looseness=2](-0.5,-5.5)
to [out=up, in=up, looseness=2](-1.1,-5.5)
to [out=down, in=left](0, -6.7);
\node [Vertex, scale=\vertexsize, vertex colour=white] at (0.8, 5.1) {};
\node [Vertex, scale=\vertexsize, vertex colour=white]at (-0.8, 5.1) {};
\node [scale=\labelsize] at (0, 7.2) {$a$};
\node [scale=\labelsize] at (0, 6.35) {$b$};
\draw [fill=\fillC, fill opacity=0.8] (3,4)
to (3,3)
to [out=down,in=down, looseness=2](4,3)
to (4,4)
to (3,4);
\draw [fill=\fillcomp, fill opacity=0.8] (5,5.5)
to (5,4.5)
to [out=down,in=down, looseness=2](6,4.5)
to (6,5.5);
\draw [fill=\fillC, fill opacity=0.8] (7,5.5)
to (7,4.5)
to [out=down,in=down, looseness=2](8,4.5)
to (8,5.5);
\draw [fill=\fillcomp, fill opacity=0.8] (3,5.5)
to (3,4)
to (4,4)
to (4,5.5);
\node [draw, minimum width=45, minimum height=27,fill=white, scale=0.5, fill opacity=1] at (3.5, 4) {$f_i$}; 
\node [scale=\labelsize] at (3.5, 3) {$a$};
\node [scale=\labelsize] at (7.5, 4.75) {$a$};
\node [scale=\labelsize] at (5.5, 4.75) {$b$};
\end{tikzpicture}
\end{aligned}
\quad-\quad
\begin{aligned}
\begin{tikzpicture} [thick, scale=0.5]
\draw (-3, 0) 
to [out=down,in=down, looseness=2](-4, 0)
to (-4, 2)
to [out=up,in=up, looseness=2](-3, 2);
\node [Vertex, scale=\vertexsize, vertex colour=white] at (-3, 0) {};
\node [Vertex, scale=\vertexsize, vertex colour=white]at (-3, 2) {};
\draw [fill=\fillcomp, fill opacity=0.8] (-3.5,1.5)
to (-3.5,0.5)
to [out=down,in=down, looseness=2](-2.5,0.5)
to (-2.5,1.5)
to [out=up,in=up, looseness=2](-3.5, 1.5);
\draw [fill=\fillC, fill opacity=0.8] (0,1)
to (0,0)
to [out=down,in=down, looseness=2](1,0)
to (1,1)
to (0,1);
\draw [fill=\fillcomp, fill opacity=0.8] (2,2.5)
to (2,1.5)
to [out=down,in=down, looseness=2](3,1.5)
to (3,2.5);
\draw [fill=\fillC, fill opacity=0.8] (4,2.5)
to (4,1.5)
to [out=down,in=down, looseness=2](5,1.5)
to (5,2.5);
\draw [fill=\fillcomp, fill opacity=0.8] (0,2.5)
to (0,1)
to (1,1)
to (1,2.5);
\draw [fill=\fillC, fill opacity=0.8]
(-2, 2.5)
to [out=up,in=up, looseness=2](-1, 2.5)
to (-1, -0.5)
to [out=down,in=down, looseness=1.5](-2, -0.5)
to [out=up, in=right](-3, 0)
to [out=right, in=down](-2, 0.5)
to (-2, 1.5)
to [out=up, in=right](-3, 2)
to [out=right, in=down](-2, 2.5);
\node [draw, minimum width=45, minimum height=27,fill=white, scale=0.5, fill opacity=1] at (0.5, 1) {$f_i$}; 
\node [scale=\labelsize] at (0.5, 0) {$a$};
\node [scale=\labelsize] at (4.5, 1.75) {$a$};
\node [scale=\labelsize] at (2.5, 1.75) {$b$};
\node [scale=\labelsize] at (-1.5, 1) {$a$};
\node [scale=\labelsize] at (-3, 1) {$b$};
\end{tikzpicture}
\end{aligned}
\right]
\end{align*}
By Lemma~\ref{lemma: auxilliary MKP complementarity result}, this simplifies as follows:
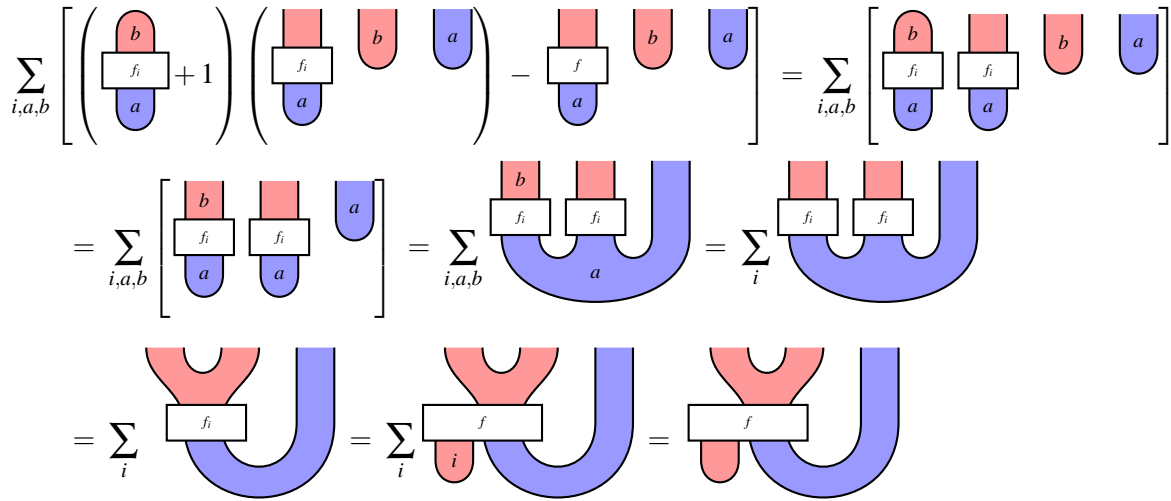
\begin{align*}
&\sum_{i,a,b}
\left[
\left(
\begin{aligned}
\begin{tikzpicture} [thick, scale=0.5]
\draw [fill=\fillC, fill opacity=0.8] (0,1)
to (0,0)
to [out=down,in=down, looseness=2](1,0)
to (1,1)
to (0,1);
\draw [fill=\fillcomp, fill opacity=0.8] (0,2)
to (0,1)
to (1,1)
to (1,2)
to [out=up, in=up, looseness=2](0,2);
\node [draw, minimum width=50, minimum height=27,fill=white, scale=0.5, fill opacity=1] at (0.5, 1) {$f_i$}; 
\node [scale=\labelsize] at (0.5, 0) {$a$};
\node [scale=\labelsize] at (0.5, 2) {$b$};
\end{tikzpicture}
\end{aligned}
+1
\right)
\left(
\begin{aligned}
\begin{tikzpicture} [thick, scale=0.5]
\draw [fill=\fillC, fill opacity=0.8] (0,1)
to (0,0)
to [out=down,in=down, looseness=2](1,0)
to (1,1)
to (0,1);
\draw [fill=\fillcomp, fill opacity=0.8] (2,2.5)
to (2,1.5)
to [out=down,in=down, looseness=2](3,1.5)
to (3,2.5);
\draw [fill=\fillC, fill opacity=0.8] (4,2.5)
to (4,1.5)
to [out=down,in=down, looseness=2](5,1.5)
to (5,2.5);
\draw [fill=\fillcomp, fill opacity=0.8] (0,2.5)
to (0,1)
to (1,1)
to (1,2.5);
\node [draw, minimum width=45, minimum height=27,fill=white, scale=0.5, fill opacity=1] at (0.5, 1) {$f_i$}; 
\node [scale=\labelsize] at (0.5, 0) {$a$};
\node [scale=\labelsize] at (4.5, 1.75) {$a$};
\node [scale=\labelsize] at (2.5, 1.75) {$b$};
\end{tikzpicture}
\end{aligned}
\right)
\quad-\quad
\begin{aligned}
\begin{tikzpicture} [thick, scale=0.5]
\draw [fill=\fillC, fill opacity=0.8] (0,1)
to (0,0)
to [out=down,in=down, looseness=2](1,0)
to (1,1)
to (0,1);
\draw [fill=\fillcomp, fill opacity=0.8] (2,2.5)
to (2,1.5)
to [out=down,in=down, looseness=2](3,1.5)
to (3,2.5);
\draw [fill=\fillC, fill opacity=0.8] (4,2.5)
to (4,1.5)
to [out=down,in=down, looseness=2](5,1.5)
to (5,2.5);
\draw [fill=\fillcomp, fill opacity=0.8] (0,2.5)
to (0,1)
to (1,1)
to (1,2.5);
\node [draw, minimum width=45, minimum height=27,fill=white, scale=0.5, fill opacity=1] at (0.5, 1) {$f$}; 
\node [scale=\labelsize] at (0.5, 0) {$a$};
\node [scale=\labelsize] at (4.5, 1.75) {$a$};
\node [scale=\labelsize] at (2.5, 1.75) {$b$};
\end{tikzpicture}
\end{aligned}
\right]
\quad=\quad
\sum_{i,a,b}
\left[
\begin{aligned}
\begin{tikzpicture} [thick, scale=0.5]
\draw [fill=\fillC, fill opacity=0.8] (0,1)
to (0,0)
to [out=down,in=down, looseness=2](1,0)
to (1,1)
to (0,1);
\draw [fill=\fillcomp, fill opacity=0.8] (2,2.5)
to (2,1.5)
to [out=down,in=down, looseness=2](3,1.5)
to (3,2.5);
\draw [fill=\fillC, fill opacity=0.8] (4,2.5)
to (4,1.5)
to [out=down,in=down, looseness=2](5,1.5)
to (5,2.5);
\draw [fill=\fillcomp, fill opacity=0.8] (0,2.5)
to (0,1)
to (1,1)
to (1,2.5);
\node [draw, minimum width=45, minimum height=27,fill=white, scale=0.5, fill opacity=1] at (0.5, 1) {$f_i$}; 
\draw [fill=\fillC, fill opacity=0.8] (-2,1)
to (-2,0)
to [out=down,in=down, looseness=2](-1,0)
to (-1,1)
to (-2,1);
\draw [fill=\fillcomp, fill opacity=0.8] (-2,2)
to (-2,1)
to (-1,1)
to (-1,2)
to [out=up, in=up, looseness=2](-2,2);
\node [draw, minimum width=45, minimum height=27,fill=white, scale=0.5, fill opacity=1] at (-1.5, 1) {$f_i$}; 
\node [scale=\labelsize] at (-1.5, 0) {$a$};
\node [scale=\labelsize] at (0.5, 0) {$a$};
\node [scale=\labelsize] at (4.5, 1.75) {$a$};
\node [scale=\labelsize] at (2.5, 1.75) {$b$};
\node [scale=\labelsize] at (-1.5, 2) {$b$};
\end{tikzpicture}
\end{aligned}
\right]
\\
&\qquad=\quad
\sum_{i,a,b}
\left[
\begin{aligned}
\begin{tikzpicture} [thick, scale=0.5]
\draw [fill=\fillC, fill opacity=0.8] (0,1)
to (0,0)
to [out=down,in=down, looseness=2](1,0)
to (1,1)
to (0,1);
\draw [fill=\fillC, fill opacity=0.8] (4,2.5)
to (4,1.5)
to [out=down,in=down, looseness=2](5,1.5)
to (5,2.5);
\draw [fill=\fillcomp, fill opacity=0.8] (0,2.5)
to (0,1)
to (1,1)
to (1,2.5);
\draw [fill=\fillC, fill opacity=0.8] (2,1)
to (2,0)
to [out=down,in=down, looseness=2](3,0)
to (3,1)
to (2,1);
\draw [fill=\fillcomp, fill opacity=0.8] (2,2.5)
to (2,1)
to (3,1)
to (3,2.5);
\node [draw, minimum width=45, minimum height=27,fill=white, scale=0.5, fill opacity=1] at (0.5, 1) {$f_i$}; 
\node [draw, minimum width=45, minimum height=27,fill=white, scale=0.5, fill opacity=1] at (2.5, 1) {$f_i$}; 
\node [scale=\labelsize] at (0.5, 0) {$a$};
\node [scale=\labelsize] at (2.5, 0) {$a$};
\node [scale=\labelsize] at (.5, 2) {$b$};
\node [scale=\labelsize] at (4.5, 1.75) {$a$};
\end{tikzpicture}
\end{aligned}
\right]
\quad=\quad
\sum_{i,a,b}
\begin{aligned}
\begin{tikzpicture} [thick, scale=0.5]
\draw [fill=\fillC, fill opacity=0.8] (4,2.5)
to (4,0.5)
to [out=down,in=down, looseness=2](3, 0.5)
to (3, 0.75)
to (2, 0.75)
to (2, 0.5)
to [out=down,in=down, looseness=2](1, 0.5)
to (1, 0.75)
to (0, 0.75)
to (0, 0.5)
to [out=down,in=down, looseness=1.2](5,0.5)
to (5,2.5);
\draw [fill=\fillcomp, fill opacity=0.8] (0,2.5)
to (0,1)
to (1,1)
to (1,2.5);
\draw [fill=\fillcomp, fill opacity=0.8] (2,2.5)
to (2,1)
to (3,1)
to (3,2.5);
\node [draw, minimum width=45, minimum height=27,fill=white, scale=0.5, fill opacity=1] at (0.5, 1) {$f_i$}; 
\node [draw, minimum width=45, minimum height=27,fill=white, scale=0.5, fill opacity=1] at (2.5, 1) {$f_i$}; 
\node [scale=\labelsize] at (2.5, -0.5) {$a$};
\node [scale=\labelsize] at (.5, 2) {$b$};
\end{tikzpicture}
\end{aligned}
\quad=\quad
\sum_i
\begin{aligned}
\begin{tikzpicture} [thick, scale=0.5]
\draw [fill=\fillC, fill opacity=0.8] (4,2.5)
to (4,0.5)
to [out=down,in=down, looseness=2](3, 0.5)
to (3, 0.75)
to (2, 0.75)
to (2, 0.5)
to [out=down,in=down, looseness=2](1, 0.5)
to (1, 0.75)
to (0, 0.75)
to (0, 0.5)
to [out=down,in=down, looseness=1.2](5,0.5)
to (5,2.5);
\draw [fill=\fillcomp, fill opacity=0.8] (0,2.5)
to (0,1)
to (1,1)
to (1,2.5);
\draw [fill=\fillcomp, fill opacity=0.8] (2,2.5)
to (2,1)
to (3,1)
to (3,2.5);
\node [draw, minimum width=45, minimum height=27,fill=white, scale=0.5, fill opacity=1] at (0.5, 1) {$f_i$}; 
\node [draw, minimum width=45, minimum height=27,fill=white, scale=0.5, fill opacity=1] at (2.5, 1) {$f_i$}; 
\end{tikzpicture}
\end{aligned}
\\
&\qquad=\quad
\sum_i
\begin{aligned}
\begin{tikzpicture} [thick, scale=0.5]
\draw [fill=\fillC, fill opacity=0.8] (4,2.5)
to (4,0.5)
to [out=down,in=down, looseness=2](2, 0.5)
to (2, 0.75)
to (1, 0.75)
to (1, 0.5)
to [out=down,in=down, looseness=1.7](5,0.5)
to (5,2.5);
\draw [fill=\fillcomp, fill opacity=0.8, draw=none] (3,2.5)
to [out=down, in=up](2,0.5)
to (1,0.5)
to [out=up, in=down](0,2.5)
to (1, 2.5)
to [out=down,in=down, looseness=2.3](2,2.5);
\draw  (3,2.5)
to [out=down, in=up](2,0.5)
to (1,0.5)
to [out=up, in=down](0,2.5)
(1, 2.5)
to [out=down,in=down, looseness=2.3](2,2.5);
\node [draw, minimum width=60, minimum height=27,fill=white, scale=0.5, fill opacity=1] at (1.6, 0.5) {$f_i$}; 
\end{tikzpicture}
\end{aligned}
\quad=\quad
\sum_i
\begin{aligned}
\begin{tikzpicture} [thick, scale=0.5]
\draw [fill=\fillC, fill opacity=0.8] (4,2.5)
to (4,0.5)
to [out=down,in=down, looseness=2](2, 0.5)
to (2, 0.75)
to (1, 0.75)
to (1, 0.5)
to [out=down,in=down, looseness=1.7](5,0.5)
to (5,2.5);
\draw [fill=\fillcomp, fill opacity=0.8, draw=none] (3,2.5)
to [out=down, in=up](2,0.5)
to (1,0.5)
to [out=up, in=down](0,2.5)
to (1, 2.5)
to [out=down,in=down, looseness=2.3](2,2.5);
\draw  (3,2.5)
to [out=down, in=up](2,0.5)
to (1,0.5)
to [out=up, in=down](0,2.5)
(1, 2.5)
to [out=down,in=down, looseness=2.3](2,2.5);
\draw [fill=\fillcomp, fill opacity=0.8] (-0.25, 0.5)
to (-0.25, -0.5)
to [out=down,in=down, looseness=2](0.75, -0.5)
to (0.75, 0.5);
\node [draw, minimum width=90, minimum height=27,fill=white, scale=0.5, fill opacity=1] at (1.0, 0.5) {$f$}; 
\node [scale=\labelsize] at (0.25, -0.5) {$i$};
\end{tikzpicture}
\end{aligned}
\quad=\quad
\begin{aligned}
\begin{tikzpicture} [thick, scale=0.5]
\draw [fill=\fillC, fill opacity=0.8] (4,2.5)
to (4,0.5)
to [out=down,in=down, looseness=2](2, 0.5)
to (2, 0.75)
to (1, 0.75)
to (1, 0.5)
to [out=down,in=down, looseness=1.7](5,0.5)
to (5,2.5);
\draw [fill=\fillcomp, fill opacity=0.8, draw=none] (3,2.5)
to [out=down, in=up](2,0.5)
to (1,0.5)
to [out=up, in=down](0,2.5)
to (1, 2.5)
to [out=down,in=down, looseness=2.3](2,2.5);
\draw  (3,2.5)
to [out=down, in=up](2,0.5)
to (1,0.5)
to [out=up, in=down](0,2.5)
(1, 2.5)
to [out=down,in=down, looseness=2.3](2,2.5);
\draw [fill=\fillcomp, fill opacity=0.8] (-0.25, 0.5)
to (-0.25, -0.5)
to [out=down,in=down, looseness=2](0.75, -0.5)
to (0.75, 0.5);
\node [draw, minimum width=90, minimum height=27,fill=white, scale=0.5, fill opacity=1] at (1.0, 0.5) {$f$}; 
\end{tikzpicture}
\end{aligned}
\end{align*}
The final diagram clearly remains unchanged under application of the projector as per Definition~\ref{Def:Mean king problem specification}, hence the result is established.
$\hfill \square$

\end{document}